\documentclass[11pt]{article}

\newcommand{\xhdr}[1]{\vspace{2mm} \noindent{\bf #1}}

\usepackage[suppress]{color-edits}  % use this to suppress the package
\addauthor{yf}{purple}    % yf for Yiding
\addauthor{bl}{blue}    % bl for Brendan
\addauthor{as}{red}      % as for Alex
\usepackage{amsfonts,amsmath,amssymb,amsthm,graphicx}
\usepackage{fullpage}
\usepackage{changepage}
\usepackage{enumerate}
\usepackage{scalerel}
\usepackage{accents}
\usepackage[margin=1in]{geometry}
\usepackage{float}
\usepackage[caption = false]{subfig}
\usepackage{hyperref}
\usepackage{enumitem}
\usepackage{authblk}

\usepackage{natbib}
\usepackage{csquotes}
\usepackage{comment}

\usepackage{bbm}

\usepackage{cleveref}
\crefname{claim}{claim}{claims}

\usepackage{tikz}
\usepackage{pgfplots}
\pgfplotsset{compat=1.18}
\usetikzlibrary{patterns}
\usepgfplotslibrary{fillbetween}
\usetikzlibrary{intersections}

\usepackage[ruled,vlined,linesnumbered]{algorithm2e}
\Crefname{algocf}{Algorithm}{Algorithms}
\allowdisplaybreaks

\usetikzlibrary{arrows, decorations.markings}

\tikzstyle{vecArrow} = [thick, decoration={markings,mark=at position
   1 with {\arrow[semithick]{open triangle 60}}},
   double distance=1.4pt, shorten >= 5.5pt,
   preaction = {decorate},
   postaction = {draw,line width=1.4pt, white,shorten >= 4.5pt}]
\tikzstyle{innerWhite} = [semithick, white,line width=1.4pt, shorten >= 4.5pt]

\allowdisplaybreaks
\theoremstyle{plain}
\newtheorem{theorem}{Theorem}[section]
\newtheorem{lemma}[theorem]{Lemma}
\newtheorem{claim}[theorem]{Claim}

\newtheorem{proposition}[theorem]{Proposition}
\newtheorem{observation}[theorem]{Observation}
\theoremstyle{plain}
\newtheorem{definition}{Definition}[section] % definition numbers are dependent on theorem numbers
\newtheorem{example}[definition]{Example}
\newtheorem{remark}[definition]{Remark}

\theoremstyle{plain}

\usepackage{xfrac}
\usepackage{thm-restate}

\newcommand{\OPT}{\texttt{OPT}}

\newcommand{\alloc}{x}
\newcommand{\price}{p}

\newcommand{\wealth}{\truewealth}
\newcommand{\pacescalar}{\alpha}
\newcommand{\val}{v}

\newcommand{\primed}{^\dagger}
\newcommand{\doubleprimed}{^\ddagger}

\newcommand{\util}{u}

\newcommand{\moneycost}{C}
\newcommand{\reals}{\mathbb{R}}

\newcommand{\twelfare}{W}
\newcommand{\wtp}{\twelfare}
\newcommand{\moneycostderivative}{R}

\newcommand{\truewealth}{w}

\newcommand{\allocs}{\mathbf{\alloc}}
\newcommand{\optalloc}{\alloc^*}

\newcommand{\valuefunc}{V}
\newcommand{\valuefuncderivative}{S}

\newcommand{\instance}{I}

\newcommand{\vali}{\val_i}

\newcommand{\allocij}{\alloc_{ij}}

\newcommand{\prices}{\mathbf{\price}}

\newcommand{\eqlb}{\texttt{NE}}
\newcommand{\eqlbs}{\mathcal{E}}

\newcommand{\valuefunctioni}{\valuefunc_i}
\newcommand{\truewealthi}{\truewealth_i}
\newcommand{\moneycosti}{\moneycost_i}

\newcommand{\ctr}{\phi}
\newcommand{\ctrs}{\boldsymbol{\ctr}}
\newcommand{\ctri}{\ctrs_i}
\newcommand{\ctrij}{\ctr_{ij}}
\newcommand{\purchaseset}{H}
\newcommand{\purchaseseti}{\purchaseset_i}
\newcommand{\pricej}{\price_j}

\newcommand{\randomprice}{\mathbf{\price}}
\newcommand{\randompricej}{\randomprice_j}
\newcommand{\payment}{t}
\newcommand{\alloci}{\alloc_i}
\newcommand{\paymenti}{\payment_i}
\newcommand{\randomalloc}{\mathbf{x}}
\newcommand{\randomallocs}{{\randomalloc}}
\newcommand{\randomalloci}{\randomallocs_i}
\newcommand{\randompayment}{\mathbf{\payment}}

\newcommand{\valuefunctionderivativei}{\valuefuncderivative_i}
\newcommand{\moneycostderivativei}{\moneycostderivative_i}
\newcommand{\reportval}{\tilde\val}
\newcommand{\reportwealth}{\tilde\wealth}
\newcommand{\randomreportval}{\tilde{\mathbf\val}}
\newcommand{\randomreportwealth}{\tilde{\mathbf\wealth}}
\newcommand{\realsinf}{\reals_+^{\infty}}
\newcommand{\strategy}{\mathbf{s}}
\newcommand{\valuefunction}{\valuefunc}
\newcommand{\PoA}{\Gamma}
\newcommand{\purePoA}{\PoA_{\texttt{Pure}}}
\newcommand{\mixedPoA}{\PoA_{\texttt{Mixed}}}
\newcommand{\bayesPoA}{\PoA_{\texttt{Bayes}}}

\newcommand{\typedist}{F}
\newcommand{\allocselectionrule}{\sigma}

\newcommand{\tildeP}{\tilde P}

\newcommand{\tildePj}{\tildeP_j}
\newcommand{\optallocij}{\optalloc_{ij}} 

\newcommand{\lagent}{A_S}
\newcommand{\hagent}{A_G}
\newcommand{\allocLB}{y}
\newcommand{\allocLBi}{\allocLB_i}
\newcommand{\allocUB}{z}
\newcommand{\allocUBi}{\allocUB_i}
\newcommand{\priceL}{\price_L}

\newcommand{\priceH}{\price_H}
\newcommand{\PriceL}{P_L}
\newcommand{\PriceH}{P_H}
\newcommand{\allocHat}{\hat{\alloc}}
\newcommand{\allocHati}{\allocHat_i}

	\DeclareMathOperator{\argmax}{argmax}

\newcommand{\condition}{\,\mid\,}

\newcommand{\prob}[2][]{\text{Pr}\ifthenelse{\not\equal{}{#1}}{_{#1}}{}\!\left[{\def\givenn{\middle|}#2}\right]}
\newcommand{\expect}[2][]{\mathbb{E}\ifthenelse{\not\equal{}{#1}}{_{#1}}{}\!\left[{\def\givenn{\middle|}#2}\right]}

\newcommand{\tparen}{\big}
\newcommand{\tprob}[2][]{\text{Pr}\ifthenelse{\not\equal{}{#1}}{_{#1}}{}\tparen[{\def\given{\tparen|}#2}\tparen]}
\newcommand{\texpect}[2][]{\mathbb{E}\ifthenelse{\not\equal{}{#1}}{_{#1}}{}\tparen[{\def\given{\tparen|}#2}\tparen]}

\newcommand{\sprob}[2][]{\text{Pr}\ifthenelse{\not\equal{}{#1}}{_{#1}}{}[#2]}
\newcommand{\sexpect}[2][]{\mathbb{E}\ifthenelse{\not\equal{}{#1}}{_{#1}}{}[#2]}

\newcommand{\indicator}[1]{{\mathbbm{1}\left\{ #1 \right\}}}

\renewcommand{\OPT}{\ensuremath{\mathrm{OPT}}}

\title{Strategic Budget Selection in a Competitive Autobidding World}

\author[1]{Yiding Feng}
\author[2]{Brendan Lucier}
\author[3]{Aleksandrs Slivkins}

\affil[1]{University of Chicago, Chicago IL, USA.}
\affil[2]{Microsoft Research, Cambridge MA, USA.}
\affil[3]{Microsoft Research, New York NY, USA.}
\affil[ ]{\normalsize{\texttt{yiding.feng@chicagobooth.edu}},~\{\normalsize{\texttt{brlucier,slivkins}}\}\normalsize{\texttt{@microsoft.com}}}

\date{}

\begin{document}

\maketitle

\begin{abstract}
We study a game played between advertisers in an online ad platform.  The platform sells ad impressions by first-price auction and provides autobidding algorithms that optimize bids on each advertiser's behalf, \bledit{subject to advertiser constraints such as budgets}.  %This defines a bidding game played by the autobidders.  but also defines a ``meta-game'' between the advertisers who are \emph{strategically} declaring their budgets.  Advertiser payoffs in the constraint-choosing ``metagame'' are determined by the equilibrium reached by the autobidders.
Crucially, these constraints are \emph{strategically} chosen by the  advertisers.
%declares constraints 
%(and possibly a maximum bid) to their autobidder.  
The chosen constraints define an ``inner'' budget-pacing game for the autobidders.
%, who compete to maximize the total value received subject to the constraints.  
Advertiser payoffs in the constraint-choosing ``metagame'' are determined by the equilibrium reached by the autobidders.

Advertiser
%s only specify budgets and linear values to their autobidders, but their true 
preferences can be more general than what is implied by their constraints: we assume only that they have weakly decreasing marginal value for clicks and weakly increasing marginal disutility for spending money.
Nevertheless, we show that
%Our main result is that,
%despite this gap between general preferences and simple autobidder constraints,
%the allocations at equilibrium are approximately efficient.
%Specifically, 
at any pure Nash equilibrium of the metagame, the resulting allocation obtains at least half of the liquid welfare of any allocation and this bound is tight.  We also obtain a $4$-approximation for any mixed Nash equilibrium 
%, and this result extends also to 
or Bayes-Nash equilibria.  These results rely on the power to declare budgets: if advertisers can specify only a (linear) value per click \bledit{or an ROI target} but not a budget constraint, the approximation factor at equilibrium can be as bad as linear in the number of advertisers.

\end{abstract}

%The autobidder takes as input a target budget and a value per outcome (e.g., per click), and employs budget-pacing to optimize the total value received subject to those constraints.  Crucially, we view each advertiser's declared budget and value as a strategic choice.  These choices define a game played by the autobidders, and advertiser payoffs are defined by the (unique) equilibrium of that game.
%Despite advertiser preferences being substantially more complex than the simple constraints they report,

\setcounter{page}{0}
\newpage

\section{Introduction}
\label{sec:intro}

In many large online platforms, it is increasingly common for users to delegate their choices to algorithmic proxies that optimize on their behalf.  Examples include autobidders in online advertising markets~\citep{GoogleAutobidder,BingAutobidder,FacebookAutobidder}, dynamic price adjustment algorithms for sales or rental platforms like Amazon and Airbnb~\citep{AmazonPricing,AirbnbPricing}, price prediction tools for flights~\citep{HopperPrediction,ExpediaPrediction}, and more.  These tools typically employ techniques from machine learning to optimize some goal subject to user-specified constraints or targets.  When many users deploy these algorithmic tools simultaneously, the algorithms are effectively competing against each other.  But in addition to the game between the algorithms, this setup defines a metagame between the users, who are choosing the parameters of their algorithms in anticipation of the competition.  This raises a natural question: \emph{how will users play the metagame, and what is the impact on the market?}

We study these questions through the lens of autobidding in online advertising markets.  Online advertising platforms sell individual advertising events (like clicks or conversions) by auction, with a separate auction held for each ad impression.  These auctions are strategically linked and optimal bidding is complicated, but advertisers can delegate the bidding details to an autobidder. We focus on online advertising in part because the use of automated bidding algorithms is very well established in that market, with all major advertising platforms offering integrated autobidding services~\citep{GoogleAutobidder,BingAutobidder,FacebookAutobidder}.

The autobidding paradigm involves three layers: advertisers, autobidders, and auctions.  See \Cref{fig:metagame} for an illustration.  Each advertiser has their own individual autobidding algorithm (often provided by the platform) that bids on their behalf.  An advertiser provides their autobidder with aggregate instructions like a budget constraint and/or a maximum bid.\footnote{In our model, maximum bid constraints are equivalent to average return-on-investment (ROI) constraints, which are common in autobidding.  See Remark~\ref{remark:ROI} in Section~\ref{sec:prelim}.}  The autobidder then participates in many individual auctions, implementing learning algorithms that use auction feedback to adjust bids in order to maximize the total value received (e.g., total clicks) while adhering to the specified constraints.  There is by now a very rich academic literature on the design of autobidding algorithms for various constraints and auction designs~\citep{Borgs-www07,balseiro2019learning,balseiro2017budget,balseiro2020dual,balseiro2023best,castiglioni2022unifying} and how a platform might design its auction rules or prices for such algorithms~\citep{balseiro2021robust,deng2021towards,golrezaei2023pricing}.

From the advertisers' perspective, the autobidder and auction layers can be seen as a cohesive mechanism that the advertisers interact with strategically. The advertiser-specified constraints (e.g., budgets) define the mechanism's message space. But we emphasize that this is not a direct-revelation mechanism, and the allowable constraints do not necessarily capture the intricacies of advertiser preferences. Indeed, the latter may be much more complex than, e.g., maximizing a linear value for clicks or other events subject to a budget.\footnote{And even if the message space does capture advertisers' preferences, these mechanisms do not generally incentivize truth-telling; see, e.g., \cite{AMP-23}.}

Thus, while the autobidders play a bidding game amongst themselves, the advertisers play a game of constraint selection, henceforth called the \emph{metagame}. The actions in this metagame are the budgets and any other constraints specified to the autobidders. These constraints in turn define the bidding game, called the \emph{inner game}, played by the autobidders. The outcome of the inner game ultimately determines the advertisers' payoffs. In this paper we focus on the metagame played by the advertisers, explore its equilibrium outcomes, and analyze the efficiency of the market.

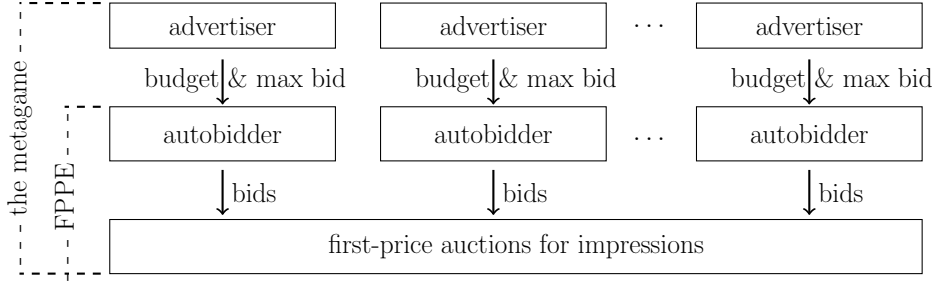
\begin{figure}
    \centering
    \begin{tikzpicture}[scale=0.6, transform shape]

\draw (0, 0) rectangle (18, 1.2);
\draw (9, 0.6) node {\LARGE first-price auctions for impressions};

\draw (0, 2.5) rectangle (5, 3.7);
\draw (2.5, 3.1) node {\LARGE autobidder};
\draw[->, thick] (2.5, 2.3) -- (2.5, 1.3);
\draw (3.2, 1.8) node {\LARGE bids};

\draw (6, 2.5) rectangle (11, 3.7);
\draw (8.5, 3.1) node {\LARGE autobidder};
\draw[->, thick] (8.5, 2.3) -- (8.5, 1.3);
\draw (9.2, 1.8) node {\LARGE bids};

\draw (13, 2.5) rectangle (18, 3.7);
\draw (15.5, 3.1) node {\LARGE autobidder};
\draw[->, thick] (15.5, 2.3) -- (15.5, 1.3);
\draw (16.2, 1.8) node {\LARGE bids};

\draw (12, 3) node {\LARGE $\dots$};

\draw (0, 5) rectangle (5, 6);
\draw (2.5, 5.5) node {\LARGE advertiser};
\draw[->, thick] (2.5, 4.75) -- (2.5, 3.75);
\draw (3, 4.3) node {\LARGE budget \& max bid};

\draw (6, 5) rectangle (11, 6);
\draw (8.5, 5.5) node {\LARGE advertiser};
\draw[->, thick] (8.5, 4.75) -- (8.5, 3.75);
\draw (9, 4.3) node {\LARGE budget \& max bid};

\draw (13, 5) rectangle (18, 6);
\draw (15.5, 5.5) node {\LARGE advertiser};
\draw[->, thick] (15.5, 4.75) -- (15.5, 3.75);
\draw (16, 4.3) node {\LARGE budget \& max bid};

\draw (12., 5.5) node {\LARGE $\dots$};

 \node[rotate=90] at (-1,1.7) {\LARGE - - - FPPE - - -};
 \node[rotate=90] at (-2,3) {\LARGE - - - the metagame - - -};

\draw [dashed, thick] (-0.2, 0) -- (-1, 0);
\draw [dashed, thick] (-0.2, 3.7) -- (-1, 3.7);
\draw [dashed, thick] (-1.2, 0) -- (-2, 0) ;
\draw [dashed, thick] (-0.2, 6) -- (-2, 6);

\end{tikzpicture} 
    \caption{The metagame of strategic budget selection.}
    \label{fig:metagame}
\end{figure}

%\subsection{Our Contributions.}
%To study the strategic behavior of advertisers in the autobidding paradigm, we
%%introduce a novel game-theoretic model --
%{consider} a \emph{metagame of strategic budget selection}.
%%\blcomment{I softened the language of the previous sentence, since the EC'23 paper also treats budget selection as strategic. But we should still highlight the novelty of our model with advertiser utilities being distinct from the mechanism message space, which is completely novel.}
%In this metagame, each advertiser \emph{strategically} declares a budget constraint (and possibly a maximum bid) to their autobidder. The chosen constraints define an ``inner'' budget-pacing game for the autobidders, who compete with each other to maximize the total value (i.e., number of events obtained) subject to the constraints. Advertiser payoffs in the original metagame are then determined by the equilibrium reached by the autobidders. Below we discuss and highlight a few details of our metagame.

\xhdr{Our model.}
For the inner game, each autobidder receives as input a budget constraint and maximum bid, either of which can be infinite.
We first need to specify the outcome of the inner game between the autobidders for a fixed choice of constraints. One potential challenge is that if there are multiple equilibria in the inner game, the metagame payoffs would be ambiguous and dependent on the details of the learning dynamics.
%therein: indeed, advertisers' payoffs in the metagame are easily unspecified and possibly depend on the inner game dynamics. 
Fortunately, if the auction format is a first-price auction, the inner game has an essentially unique equilibrium known as the \emph{first-price pacing equilibrium (FPPE)} \citep{CKPSSSW-22}.  \bledit{Moreover, common learning dynamics are known to converge quickly to this equilibrium \citep{Borgs-www07}.} Given that and the practical prevalence of first-price auctions in the online advertising market, we will focus on first-price auctions in our model and use FPPE as our predicted outcome of the inner game.\footnote{In contrast, the inner game equilibrium
%it is known that the equilibrium of the budget-pacing game in the second-price auction
may not be unique \bledit{for 2nd-price auctions \citep{CKSS-22}, and finding any equilibrium is PPAD-hard~\citep{CKK-21}.}  So extending to 2nd-price auctions necessitates taking a stance on equilibrium selection and/or the learning dynamics,
%to extend our model from the first-price auction context to the second-price context, a stance on equilibrium selection of the inner game is necessary,
which we leave for future work.}  \bledit{This approach allows us to abstract away from the details of the autobidder implementation; our results are relevant to any learning methods that converge to the FPPE.}

%\xhdr{First-price auction and inner game equilibrium uniqueness.}
%One potential challenge that arises is the possibility of non-uniqueness of equilibrium in the inner game.  When there are multiple equilibria in the inner budget-pacing game, the payoff of the advertisers in the metagame becomes underspecified. Fortunately, for first-price auctions the equilibrium of the budget-pacing game is essentially unique, and is known as the \emph{first-price pacing equilibrium (FPPE)} \citep{CKPSSSW-22}. Given the desired theoretical properties of FPPE and the practical prevalence of the first-price auction in the online advertising market, we will focus on first-price auctions and use FPPE as the unique outcome of the autobidders to determine the payoff of the advertisers.\footnote{In contrast, it is known that the equilibrium of the budget-pacing game in the second-price auction may not be unique \citep{CKSS-22}. Therefore, to extend our model from the first-price auction context to the second-price context, a stance on equilibrium selection of the inner game is necessary, which we leave for future work.}
%See \Cref{fig:metagame}.
%%for a graphical illustration of our model.  %\blcomment{Not sure what figure is intended here.}

%\xhdr{Advertisers with general utility models.}
Turning next to the advertisers, our model encompasses a broad range of advertiser preferences.
We consider general \emph{separable} utility models
%that incorporates concave increasing
with a weakly concave (not necessarily linear) valuation for events (e.g., clicks) and weakly convex increasing disutility for spending money.
%This formulation includes the case of
This includes quasi-linear utility up to a hard budget constraint, but also allows value-maximizing advertisers, 
%. But it also allows
soft budget constraints, decreasing marginal value for clicks, etc.
We emphasize the distinction between preferences and constraints: while the autobidders strive to maximize value given the specified constraints, the advertisers are optimizing for their more general preferences.

%We emphasize that while the advertiser utility model is general, they must declare a linear value (i.e., maximum bid) and a hard budget constraint to their autobidders, which may not perfectly capture their true preferences. Furthermore, while autobidders function as \emph{value-maximizers}, designed to maximize the total value received while adhering to the specified constraints, advertisers are \emph{utility-maximizers} based on their true preferences.

% \asmargincomment{moved here}
To measure the efficiency of the overall market, we focus on the \emph{liquid welfare}, which is the summation of each advertiser's willingness-to-pay given her own type and allocation. This general definition
%Our general definition of liquid welfare
encompasses the classic definition of welfare for linear agents and the original definition of liquid welfare for budgeted agents, respectively \citep{vic-61,DP-14}. Liquid welfare is also closely related to \emph{compensating variation} from economics \citep{hic-75}.%
\footnote{\label{footnote:compensating variation}Compensating variation is the transfer of money required, after some market change, to return agents to their original utility levels.  In our scenario, the ``change" is an allocation of resources, and the compensating variation is the maximum amount the agents would pay for that allocation; i.e., the liquid welfare.}
%Given the advertisers' types, the liquid welfare is determined by the allocation; 
\bledit{We emphasize that liquid welfare is determined by the allocation and the advertisers' true preferences, not the payments or the declared constraints.}
In our results, we compare to the optimal liquid welfare achieved by \emph{any} allocation, henceforth called the \OPT.

\xhdr{Approximate efficiency in the metagame.}
In our baseline model, we consider a full-information environment where each advertiser's type (i.e., valuation and disutility for spending money) is fixed and publicly known.  We first consider pure Nash equilibria of the metagame: the advertisers simultaneously declare their constraints to their respective autobidders %(deterministically for pure Nash or randomly for mixed Nash),
%each advertiser deterministically (for pure) or randomly (for mixed) declares constraints to their autobidder simultaneously,
then outcomes and payoffs are determined by the FPPE of the inner game.
%the inner FPPE of the autobidders.

%As the main results of this paper, w
\bledit{While the FPPE is essentially unique given the choices of the advertisers, the same is not true for the metagame of constraint selection.  We show through a sequence of examples that the metagame may have multiple pure Nash equilibria, or none at all.  Intuitively, a unilateral change of budget by one advertiser can substantially change the equilibrium behavior of \emph{all} autobidders in the inner game, leading to a rich strategic environment for the advertisers.  Nevertheless,} 
in our main result we prove that at any pure Nash equilibrium (if one exists), the resulting liquid welfare is a 2-approximation to the
$\OPT$
%optimal liquid welfare
and this is tight.\footnote{We say that liquid welfare at equilibrium is an $\alpha$-approximation to the \OPT, for some approximation factor $\alpha\geq 1$, if the liquid welfare is at least $\OPT/\alpha$ for any vector of advertiser types.}
%Moreover, at any mixed Nash equilibrium (which is guaranteed to exist), the resulting liquid welfare is at most a 4-approximation to the
%\OPT.
%optimal liquid welfare.
%
%We also consider an extension of our model to Bayesian environments, where each advertiser's type is independently drawn from a prior, and show a $4$-approximation to the expected \OPT\  at any Bayesian Nash equilibrium.

\bledit{
\begin{theorem}
\label{thm:intro.PNE}
%At any Bayesian Nash equilibrium of the metagame, the expected liquid welfare is a $4$-approximation to the expected \OPT.  
At any pure Nash equilibrium of the metagame, the liquid welfare is a $2$-approximation to the \OPT\ and
this approximation factor is tight.
\end{theorem}
}

We also show that at any mixed Nash equilibrium (which is guaranteed to exist), the resulting expected liquid welfare is at most a $4$-approximation to the \OPT.  In fact, this approximation result extends to Bayesian environments, where each advertiser's private type is independently drawn from a publicly-known prior.  In this extended setting, we can likewise show that the expected liquid welfare at any Bayesian Nash equilibrium is at most a $4$-approximation to the expected \OPT.

\bledit{
\begin{theorem}
\label{thm:intro.BNE}
At any mixed or Bayesian Nash equilibrium of the metagame, the expected liquid welfare is a $4$-approximation to the expected \OPT.  
\end{theorem}
}

\bledit{We reiterate that our approximation bounds hold for any (distribution over) separable preferences of the advertisers.\footnote{To the best of our knowledge, the only other $O(1)$-approximation for such separable utilities in a multi-item auction is due to~\cite{BCHNL-21}, who study direct bidding in second-price auctions; their bound applies to pure Nash equilibria under an additional ROI-optimality refinement and linear valuations.  See Section~\ref{sec:related} for further discussion.}
We view the generality of our bounds as a benefit of the design paradigm of mechanisms with proxy autobidders, which provide an interface for advertisers with complex preferences to engage with a simple and robust auction format.}

\xhdr{Importance of budget constraints.} Our approximation results highlight an important practical implication:
\begin{displayquote}
\emph{Simple instructions for autobidders are sufficient to achieve good market outcomes even when advertisers have much more complex preferences.}
\end{displayquote}
An interesting follow-up question arises: Can these instructions be further simplified? Is it necessary to provide advertisers with the option to specify both hard budgets and maximum bids? To answer this, we study variants of the metagame where advertisers exclusively declare budget constraints or maximum bid constraints, respectively.
%Our analysis for these variants demonstrates that the hard
We demonstrate that the budget constraints are crucial. If advertisers exclusively specify maximum bids but not budgets, the resulting approximation factor at equilibrium can be as bad as linear in the number of advertisers. On the other hand, all our approximation results extend to the variant of the metagame with no maximum bid.

\begin{theorem}
    The results of Theorem~\ref{thm:intro.PNE} and Theorem~\ref{thm:intro.BNE} continue to hold in a model where advertisers can specify only budget constraints.  However, in a model where advertisers can specify only maximum bids, there exist instances where the liquid welfare at Bayesian Nash equilibrium is at most an $\Omega(n)$-approximation to the expected \OPT.
\end{theorem}

%where advertisers exclusively declare hard budgets but not maximum bids.
Intuitively, the presence of budgets allows advertisers to hedge against unexpectedly high expenditures, which can have a potentially catastrophic impact on their utility. 

%type distribution independently.
%We generalize our analysis for the mixed Nash equilibrium in the full information environments. Specifically, we show that at any Bayesian Nash equilibrium, the resulting liquid welfare is at most a 4-approximation to the expected \OPT.
%optimal liquid welfare.

%\yfcomment{Placeholder: mention new pure Nash characterization for single-item instances.}

\xhdr{Characterization of equilibria for homogeneous items.}
\bledit{Finally, to shed light on the equilibrium structure, we also explicitly characterize the pure Nash equilibria in the special case where all ad impressions are homogeneous.  We show that a pure Nash equilibrium always exists in this homogeneous case and can be found in polynomial time.  In fact, there are two distinct \emph{types} of equilibria, which we call low-price equilibria and high-price equilibria.  There may be multiple equilibria of each type, corresponding to a range of implementable per-unit prices.  In a low-price equilibrium, all advertisers specify budget constraints that target their minimal utility-maximizing allocations at a market-clearing price.  In a high-price equilibrium, one or more advertisers forego budget constraints and instead impose high maximum bids, effectively forcing a higher price.  This multiplicity demonstrates the rich strategic landscape of the metagame, even in the special case where all ad impressions are identical.}

%At a high level, the intuition behind our approximation
\xhdr{Summary of Techniques.}
Let us now return to our approximation results and provide some additional intuition into the proof ideas.  At any pure Nash equilibrium, each advertiser will face a Pareto curve that represents the tradeoff between the value she receives and the budget constraint she declares to her autobidder. Increasing the declared budget allows the advertiser to win more, but at a potentially decreasing bang-per-buck.  %Ideally, we would like to show that this tradeoff is not too extreme, so each advertiser has a choice of budget that achieves high welfare without suffering too high a payment.  
The first step of our analysis is to analyze this tradeoff. 
%Roughly, we aim to show that 
%
{
To that end, we will employ a useful interpretation of FPPE due to~\cite{CKPSSSW-22}: the inner FPPE of autobidders corresponds to a market-clearing outcome that assigns a price to each impression, so that each autobidder obtains a preferred bundle under those prices.
}
%We would like to say that each advertiser can provide her autobidder with enough budget to obtain a high-welfare bundle.
Since increasing one autobidder's budget has downstream effects on how other autobidders will behave (even keeping the specified constraints of other advertisers fixed), adjusting one's budget can cause the market-clearing prices to change.
%; see \Cref{lem:fppe increasing budget}.
{
This gives us the perspective we need to understand an advertiser's tradeoff when choosing a budget: the sensitivity of FPPE prices and market-clearing allocations to budget changes.
}

{
Unfortunately, this sensitivity can be unbounded in general.  An advertiser might need to increase their spend by an arbitrary amount to secure a target increase in allocation.
}
\bledit{For example, consider two advertisers competing for a single impression type.  If the advertisers set finite budgets and no maximum bids, then at FPPE they will spend their budgets exactly and split the item in proportion to their declared budgets.  (See Section~\ref{sec:example} for examples and intuition for FPPE.)  This means that if the first advertiser sets a budget of $1$, the second advertiser would need to set a budget of $9$ to receive $90\%$ of the item (at a price of $10$ per unit), but a budget of $99$ to receive $99\%$ of the item (at a price of $100$ per unit).  In such a scenario, large changes in price might be required to implement a small change in allocation.}

\bledit{
However, a key insight is that this hyper-inflation only occurs when one advertiser obtains a very large allocation of an impression type.}
As long as an advertiser is not winning an excessive share of the impressions, the market-clearing prices faced by her autobidder will not be overly sensitive to small changes in her allocation.  This is our main technical lemma, Lemma~\ref{lem:poa pure allocation price relation}, which relates the rate of substitution between valuation utility and payment disutility of an advertiser at equilibrium to the FPPE prices and the share of each impression she obtains.  It implies that if an advertiser is obtaining significantly less value than she would in the optimal allocation, then either (a) the prices are low and her allocation is small, in which case by Lemma~\ref{lem:poa pure allocation price relation} she can improve her utility by increasing her budget (contradicting the equilibrium assumption), or (b) her autobidder must be facing high market-clearing prices. But in the latter case, we can employ a standard trick for bounding the price of anarchy: the lost liquid welfare from the advertiser can be charged against revenue collected from the other advertisers who are paying these high prices. {Here we use the fact that since liquid welfare measures total willingness to pay, the revenue collected is always a lower bound on liquid welfare.}  Putting these pieces together yields our $2$-approximation result for pure Nash equilibrium.

{
The extension to mixed (and Bayesian) Nash equilibrium requires a slightly different approach.  Since the messages (and hence outcomes and prices) may be random, it is less straightforward to characterize each advertiser's tradeoff between allocation and spend.  So instead of quantifying that tradeoff precisely, we consider a specific budget-setting strategy that each agent will consider.  Namely, each advertiser considers the \emph{expected} FPPE market-clearing prices (over the specified constraints of others) that would be obtained \emph{if she were not present in the market}.  She then calculates the expected cost of her expected optimal allocation under this price distribution; this will be the deviating budget she considers.  We show that if an advertiser is receiving low expected value at equilibrium and the FPPE prices are low in expectation, this deviation will be utility-improving, a contradiction.  A low-value advertiser must therefore be facing high expected prices, and we can once again charge any shortfall of liquid welfare to the revenue generated by those high prices.
}

\xhdr{Organization.} We start by formalizing the model and providing necessary preliminaries and
notations in \Cref{sec:prelim}. In \Cref{sec:example} we study two simple examples and analyze their pure Nash equilibrium in the metagame. In \Cref{sec:poa pure,sec:poa mixed}, we present the approximation results of the liquid welfare at equilibrium. We study the extension of our model in the Bayesian environments in \Cref{sec:bayesian}, and several variants of our model under restrictive message space (i.e., advertisers can only exclusively declare hard budgets, or maximum bids) in \Cref{sec:value/budget reporting only}.  Finally, in Section~\ref{sec:single item pure nash} we characterize and prove existence of pure Nash equilibria for single-item (aka., homogeneous) instances.
%\yfcomment{Placeholder: mention new pure Nash characterization for single-item instances.}

\subsection{Additional Related Work}
\label{sec:related}

\paragraph{Incentives in Autobidding.} A recent line of work has studied the incentives of value-maximizing advertisers to truthfully describe their preferences to their autobidders. \cite{MP-23} show that the advertisers may benefit by misreporting their preference to autobidders when the platform does not have commitment power and can change the auction rule. \cite{LT-22} conduct an empirical investigation using numerical examples, illustrating the potential advantages for advertisers when misreporting to autobidders in a second-price auction.  Conceptually closest to our work is \cite{AMP-23}, who introduce the concept of auto-bidding incentive compatibility (AIC)
%. They establish that the second-price auction, truthful auction, and first-price auction are not AIC, whereas
and show that first-price auction with uniform bidding satisfies AIC (while many other auction formats do not).
%It is important to emphasize the the significant distinctions between our work and \citet{AMP-23}.
Crucially, all of these works assume that both advertisers and their autobidders are value-maximizers with similar constraints, so the advertisers are able to declare their true preferences to their autobidders.  In contrast, in our model the advertisers are utility-maximizers and they can only give their autobidders simple constraints (hard budget, maximum bid) that may not be rich enough to capture their true complex preferences.  In particular, the induced mechanism is not direct-revelation and hence incentive compatibility does not apply; we focus instead on Nash equilibria.
%so we focus instead on equilibria of constraint-selection.
%The resulting mechanism is therefore not a direct revelation , and we study the equilibria of the resulting game.
%This separation between advertiser utilities and the message space of the autobidding mechanism is a primary feature of our model.
%with a general class of utility model, and they can only declares simple constraints (hard budget, maximum bid) to their autobidders, which not necessary capture their true complex preference. In other words, AIC is not well-defined in our model. Finally, our paper also investigates the approximate efficiency of overall markets.

\xhdr{Metagame between no-regret learning agents.}
\cite{KN-22a} explores the concept of a metagame involving rational agents who utilize regret-minimizing learning algorithms to play games on their behalf,
%Unlike our model where advertisers with complex preferences have a simple message space that may not capture their true preferences, the model of \citet{KN-22a} assumes that the message space encompasses all possible preferences. The focus of \citet{KN-22a} is investigating
and investigates
whether the agents
%of these learning algorithms have incentives
are incentivized
to manipulate or misrepresent their true preferences in various classic games.
In \cite{KN-22b}, the authors delve into a similar problem but within the context of auctions involving two agents with linear utility, %Specifically, they show the incentive for
showing that agents have incentive to misreport their true values in a second-price auction when both agents use multiplicative-weights learning algorithms. On the other hand, they find that truthful reporting forms a Nash equilibrium in the first-price auction when both agents use mean-based learning algorithms.
Compared to both \cite{KN-22a} and \cite{KN-22b}, this present paper considers a different auction setup with budget constraints, focuses on cases where advertiser preferences and autobidder instructions are different, and abstracts away from the learning process by assuming that the autobidders converge to bidding according to an FPPE of the simultaneous auction game. 
%and focuses on the case where advertiser preferences are more complex than the space of allowable instructions to autobidders.
%, and we analyze equilibria of the resulting game.
%has a distinct characteristic in the metagame: advertisers with complex preferences have a simple message space that may not capture their true preferences. Furthermore, our paper considers a more general auction setup with a general number of agents and heterogeneous items, in contrast to \citet{KN-22b} which primarily focuses on auctions involving two agents and identical items.
%Lastly, our paper also investigates the approximate efficiency of overall markets.

\xhdr{Price of anarchy for non-quasi linear agents} There is a long line of literature about the price of anarchy (PoA) -- the approximation between worst equilibrium and best outcome -- in the context of liquid welfare for non-quasi linear agents. Besides works already mentioned above, \cite{DP-14} introduce the concept of liquid welfare for agents with hard budget constraints. They prove the PoA for posting market clearing prices and the clinching auction. \cite{LX-15} design a sampling mechanism with a better PoA guarantee. All mechanisms studied in \cite{DP-14,LX-15} are truthful mechanisms. For non-truthful mechanisms, \cite{CST-15,CV-16,CV-18} study the PoA of the simultaneous Kelly mechanism, while \cite{AFGR-17} studies the PoA of the simultaneous first-price auction and second-price auction with no over-bidding. \cite{BCHNL-21} studies the PoA under ROI-optimal pure Nash equilibrium of the simultaneous second price auction, establishing a $2$-approximation and that such a pure Nash equilibrium always exists. \cite{ABM-19,LMP-22,meh-22,deng2022efficiency} study the PoA for the autobidders under various auctions. \cite{FT-23,GLLLS-23,LPSZ-23} study the dynamic of no-regret learning/budget-pacing players (autobidders) and provide the liquid welfare guarantees. It is important to highlight that
%previous research on autobidding has neglected the strategic behavior of
these works treat the autobidder constraints as exogenous and do not model them as strategic choices, focusing rather on the auction and autobidders' interaction. Additionally, except \cite{BCHNL-21}, the aforementioned studies assume that agents have a linear disutility for spending money up to a fixed budget. In contrast, our paper (like \citealp{BCHNL-21}) considers agents with a general convex disutility function for spending money with a hard budget.

\section{Model and Preliminaries}
\label{sec:prelim}

\paragraph{Agent models.}
There are $n\geq 2$ agents (advertisers)
and $m$ divisible items (impressions).
The outcome for agent $i$ is
$(\alloci, \paymenti)$,
where $\alloci = (\alloc_{i1}, \dots, \alloc_{im})
\in [0, 1]^m$
is the allocation for each item $j$ and
$\paymenti\in\reals_+$ is the payment.
Given allocation $\alloci$,
agents receives
$\sum_{j\in[m]}\ctrij\allocij$
number of clicks,\footnote{From this point on we will assume the event of interest is ``clicks'' for expositional convenience.}
where $\ctrij\in \reals_+$
is the \emph{click-through rate}\footnote{The click-through rate can also be interpreted as conversion rate or other related concepts in different applications.}
of each item $j$ for agent $i$.

Agent $i$'s  von Neumann–Morgenstern utility is
parameterized by
her \emph{type
$(\valuefunctioni,
\truewealthi, \moneycosti)$}:
Given outcome $(\alloci,
\paymenti)$,
agent $i$'s utility $\util_i(\alloci,\paymenti)$
is defined as
\begin{align*}
    \util_i(\alloci,\paymenti)
    \triangleq
    \left\{
    \begin{array}{ll}
    \valuefunctioni\left(
    \sum_{j\in[m]}
    \ctrij\allocij\right)
     -
    \moneycosti\left(\paymenti\right)&
    \quad \text{if } \paymenti \leq \truewealthi\\
    -\infty     &
     \quad \text{if } \paymenti > \truewealthi
    \end{array}
    \right.
\end{align*}
where
$\valuefunctioni:\reals_+\rightarrow \reals_+$
is the \emph{valuation function}
mapping from the total number of received clicks
to agent's valuation;
$\truewealthi\in \realsinf$
is the \emph{hard budget};\footnote{We use notation $\realsinf$
to denote the set of all non-negative real numbers and infinite,
i.e., $\realsinf = \reals_+ \cup \{\infty\}$.}
and
$\moneycosti:\reals_+ \rightarrow\reals_+$
is the \emph{money cost function}
mapping from the payment to the disutility for spending money.\footnote{An interpretation of
money cost function $\moneycosti$ is as follows:
Agent $i$ has an outside option value for her money,
e.g., spending on another advertising platform.
Then $\moneycosti(\paymenti)$
is the utility hat buyer $i$ foregoes by paying
$\paymenti$ to the current advertising platform (seller).}
We assume $\valuefunctioni$
is differentiable,
weakly concave, weakly increasing, and $\valuefunctioni(0) = 0$;
and money cost function
$\moneycosti$
is differentiable,
weakly convex,
weakly increasing, and $\moneycosti(0) = 0$.
\bledit{We will also assume that either $w_i < \infty$ or $\moneycosti(t_i) > 0$ for some $t_i$.}\footnote{This rules out agents with no value for money or spending constraint. It ensures liquid welfare is well-defined.}
%\yfedit{We also assume that every agent $i$'s money cost function satisfies $\lim_{\paymenti\rightarrow\infty}\moneycosti(\paymenti) = \infty$.}\footnote{This assumption is restrictive only if $w_i = \infty$, and by convexity of $C_i$ it is equivalent to requiring that $C_i$ is not identically zero when $w_i = \infty$.} 
%if 
%she has no hard budget, i.e., $\truewealthi = \infty$.}
We write $\valuefunctionderivativei(\sum_j\ctrij\allocij)$
and
$\moneycostderivativei(\paymenti)$
as the derivative of $\valuefunctioni$
and $\moneycosti$
at $\sum_j\ctrij\allocij$
and $\paymenti$, respectively.
\yfdelete{Finally, for every agent $i$,
we normalize $\moneycosti$ such that $\moneycostderivativei(0) = 1$,
i.e.,  the first infinitesimal unit of money is worth one infinitesimal unit of value.}
Agents with this utility model are called \emph{general agents}.

%\asdelete{In addition to agents with the aforementioned general utility model, hereafter \emph{general agents}, it is worth noting two classic utility models that can be viewed as special cases of the general model.}

Three classic models can be viewed as special cases of general agents.
\begin{enumerate}
    \item \textsl{(linear utility)}
    An agent $i$ with linear utility with type $\vali\in \reals_+$, hereafter
    \emph{linear agent},
    has
    linear valuation function
    $\valuefunctioni(\sum_{j}\ctrij\allocij)
    = \val_i\cdot \sum_{j}\ctrij\allocij$
    where $\val_i$ is her \emph{value per click},
    no hard budget (aka., $\truewealthi = \infty$),
    and identity money cost function $\moneycosti(\paymenti) = \paymenti$.

    \item \textsl{(budgeted utility)}
    An agent $i$ with budgeted utility with type $(\vali, \truewealthi)\in\reals_+\times \realsinf$, hereafter
    \emph{budgeted agent}, has
    linear valuation function
    $\valuefunctioni(\sum_{j}\ctrij\allocij)
    = \val_i\cdot \sum_{j}\ctrij\allocij$
    where $\val_i$ is her \emph{value per click},
    hard budget $\truewealthi$,
    and identity money cost function $\moneycosti(\paymenti) = \paymenti$.

    \yfedit{
    \item \textsl{(value-maximizing utility)}
    A \emph{value-maximizing} agent $i$ with type $(\vali, \truewealthi) \in \reals_+ \times \reals_+$,
    has linear valuation function $\valuefunctioni(\sum_{j}\ctrij\allocij)
    = \val_i\cdot \sum_{j}\ctrij\allocij$
    where $\val_i$ is her \emph{value per click},
    hard budget $\truewealthi$,
    and zero money cost function $\moneycosti(\paymenti) = 0$ for $t_i \leq w_i$.
    {Value-maximizing agents} have been studied extensively in 
    the recent autobidding literature~\citep{balseiro2021landscape}.
    }
\end{enumerate}
Clearly, every linear agent is a budgeted agent,
and every budgeted agent is a general agent.

\xhdr{First-price pacing equilibrium for budgeted agents (autobidders)}
Before we introduce our solution concept for general agents, we first revisit a pivotal concept: the \emph{first-price pacing equilibrium}, introduced by \cite{CKPSSSW-22} to study budgeted agents (autobidders) in an advertising market.\footnote{\yfedit{For ease of presentation, we use an equivalent definition of the first-price pacing equilibrium. See \Cref{apx:FPPE original definition} for the original definition from \cite{CKPSSSW-22} and the proof of equivalence.}}

\begin{definition}[First-price pacing equilibrium]
\label{def:FPPE}
For budgeted agents with types $\{(\val_i, \truewealth_i)\}_{i\in[n]}$, a \emph{first-price pacing equilibrium (FPPE)} is a tuple $(\price, \alloc, \payment)$ of per-unit price $\pricej\in \reals_+$ for each item $j$, allocation $\alloci\in[0,1]^m$, and payment $\paymenti\in\reals_+$ for each agent $i$ that satisfies the following properties:
\begin{enumerate}
    \item \textsl{(highest bang-per-buck)} if $\allocij > 0$, then $j \in \argmax_{j'\in[m]}\frac{\val_i\ctr_{ij'}}{\price_{j'}}$
    and $\vali\ctrij \geq \pricej$;
    \item \textsl{(supply feasibility)} $\sum_{i\in[n]}\allocij \leq 1$ and equality holds if $\pricej > 0$;
    \item \textsl{(payment calculation)} $\paymenti = \sum_{j\in[m]}\allocij\pricej$;
    \item \textsl{(budget feasibility)} $\paymenti \leq \truewealthi$
    and equality holds if $\max_{j\in[m]}\frac{\val_i\ctr_{ij}}{\pricej} > 1$.
\end{enumerate}
\end{definition}

{As noted by~\cite{CKPSSSW-22}, the per-unit prices, allocation, and payments in FPPE can be interpreted as the outcome of a budget-pacing game for first-price auctions, as follows.} Each budgeted agent (autobidder) $i$ first determines an item-independent \emph{pacing multiplier $\pacescalar_i = 1 \wedge (\min_{j\in[m]}\sfrac{\price_{j}}{\vali\ctr_{ij}}) \in[0, 1]$}. The agent then submits bid $\pacescalar_i\vali\ctrij$ for each item $j$.\footnote{Here the goal of each autobidder is to maximize the total value (or equivalently total number of clicks due to the linear valuation function) received subject to the budget constraint. Under this interpretation, the pacing multiplier can be considered as the Lagrangian  multiplier of the budget constraint.} Then the allocation is the highest-bids-win and the payment is computed under the first-price format. As a sanity check, under this interpretation, the ``highest bang-per-buck'' property is satisfied due to the definition of pacing scalar $\pacescalar_i$ and the highest-bids-win allocation construction. 

%\asdelete{In addition to the aforementioned interpretation,}
The FPPE can also be interpreted as a (supply-unaware) competitive equilibrium: if we fix the per-unit price $\pricej$ for each item $j$, the allocation $\alloci$ maximizes the utility of agent $i$ subject to her budget $\truewealthi$, without taking into account the supply feasibility.  We will occasionally refer to this interpretation.

%\asdelete{Going forward, we will occasionally refer to this competitive equilibrium interpretation. However, our primary focus will be on motivating our model using the  former first-pricing pacing interpretation.}
%
%As previously mentioned, the FPPE has been studied as a natural allocation and payment outcome for budgeted agents (autobidders) in the digital advertising market \citep{CKPSSSW-22,AMP-23}.

%A desirable characteristic of the FPPE is its theoretical existence and uniqueness.

\begin{lemma}[\citealp{CKPSSSW-22}]
\label{lem:FPPE uniqueness}
    For any set of budgeted agents, FPPE exists. Moreover, the per-unit price of each item and utility of each budgeted agent are unique.
\end{lemma}
Given \Cref{lem:FPPE uniqueness}, we denote
$\price(\{(\val_i, \wealth_i)\}_{i\in[n]})$,
$\alloc^\allocselectionrule(\{(\val_i, \wealth_i)\}_{i\in[n]})$,
and $\payment^\allocselectionrule(\{(\val_i, \wealth_i)\}_{i\in[n]})$
by the unique per-unit prices
and the corresponding allocation, payment (under tie-breaking rule $\allocselectionrule$)
in the FPPE for budgeted agents with types $\{(\val_i, \wealth_i)\}_{i\in[n]}$,
respectively.\footnote{\yfedit{We assume that the seller decides a tie-breaking rule $\allocselectionrule$ exogenously ex ante. 
In real-world applications, whatever tie-breaking rule is selected by the auction designer, the autobidders are expected to effectively ``implement'' the appropriate market-clearing tie-breaking rule by making micro-adjustments to their bids across time, in order to hit their budgets in aggregate. Our single-shot model abstracts away from this dynamic behavior, but it motivates our focus on FPPE allocations.
Importantly, our results
are independent of the specific choice of $\allocselectionrule$. When it is clear from the context, we omit mentioning $\allocselectionrule$ in our notation.}}

Let us present two examples to illustrate the concept of FPPE; we reuse these examples in Section \ref{sec:example}.
%Below we present two examples that serve to illustrate the concept of FPPE. These examples will also be utilized to illustrate our main model in Section \ref{sec:example}.
\begin{example}[Two linear agents and single item]
\label{example:linear agents}
Suppose there are two linear agents and one item. Let us assume that the click-through rates are the same for both agents, i.e., $\ctr_{11} = \ctr_{21} = 1$. Additionally, we assume that the value per click $\val_1$ for agent 1 is greater than or equal to the value per click $\val_2$ for agent 2, i.e., $\val_1 \geq \val_2$.

In the FPPE, the unique per-unit price for the item is $\price_1 = \val_1$. If $\val_1 > \val_2$, the allocation in the FPPE is also unique, with $\alloc_{11} = 1$ and $\alloc_{21} = 0$. On the other hand, if $\val_1 = \val_2$, any allocation $\alloc_{11}, \alloc_{21}\in[0, 1]$ satisfying $\alloc_{11}+\alloc_{21} = 1$, along
with the unique per-unit price $\price_1 = \val_1$, forms an FPPE.
\end{example}

\begin{example}[Two budgeted agents and two items]
\label{example:budgeted agents}
Suppose there are two budgeted agents and two items. Let us assume that the click-through rate $\ctr_{ij}$ is given by $\frac{1}{2} + \frac{1}{2}\indicator{i=j}$ for each $i\in[2]$ and $j\in[2]$, i.e., each agent $i$ favors item $i$ than the other item. The value per click is the same for both agents, i.e., $\val_1 = \val_2 = 1$. Both agents have a budget constraint of $\truewealth_1 = \truewealth_2 = \frac{1}{2}$.

In the FPPE, both the per-unit price and the allocation are unique. Specifically, we have $\price_1 = \price_2 = \frac{1}{2}$, and the allocation $\alloc_{ij} = \indicator{i = j}$ for each $i\in[2]$ and $j\in[2]$. In other words, both agents receive their favorite items as per the allocation.
\end{example}

\xhdr{Metagame of strategic budget selection for general agents.}
An FPPE is defined only for budgeted agents.  For more general agents, we will take inspiration from autobidding platforms in practice and imagine the agents are provided an interface to report a budgeted agent's type that will specify the behavior of an autobidder.
%In reality, agents (advertisers) may have more complex utility models that go beyond simple budgeted utility. However, extracting such detailed information is challenging for the seller (platform). As a result,
%The prevailing approach in the digital advertising industry is to provide advertisers with an interface that {allows them to give instructions to autobidding agents who will act as value-maximizing budgeted agents.  Equivalently, we can view each advertiser's report as a budgeted agent's type that specifies the behavior of the autobidder.  %The profile of all reports is then used to determine the allocation and payment.}
%treats them as budgeted agents and asks them to report their "fake types" based on budgeted utility. These reported types are then used to determine the allocation and payment.
%Taking inspiration from this practice, we consider
This defines a \emph{metagame of strategic budget selection} for agents with general utility models:

\begin{definition}[Metagame of strategic budget selection]
\label{def:metagame}
Each general agent $i$ decides on a message $(\reportval_i, \reportwealth_i) \in (\realsinf)^2$ as her
%fake type for the budgeted utility reported
report
to the seller.\footnote{Throughout the paper, we use notation $\tilde{}$ to denote the budgeted utility model that an agent reports to the seller.} Given reported message profile $\{(\reportval_i, \reportwealth_i)\}_{i\in[n]}$, the seller implements  allocation $\alloc(\{(\reportval_i, \reportwealth_i)\}_{i\in[n]})$ and payment $\payment(\{(\reportval_i, \reportwealth_i)\}_{i\in[n]})$ induced by the FPPE, assuming that agents have budgeted utility with types $\{(\reportval_i, \reportwealth_i)\}_{i\in[n]}$.
\end{definition}

Going forward, we will refer to the metagame of budgeted utility reporting as simply the \emph{metagame} for the sake of brevity. Similarly, the FPPE induced by a message profile will be referred to as the \emph{inner FPPE}.
Given the structure of FPPE, the message $(\reportval_i, \reportwealth_i)$ reported by agent $i$ in the metagame can also be interpreted as the constraints on the maximum bid $\reportval_i$ and constraints on the maximum payment (hard budget) $\reportwealth_i$ that agent $i$ specifies to her autobidder. 
The agent has the option to report $\reportval_i = \infty$ ($\reportwealth_i = \infty$), indicating the absence of any constraint on the maximum bid or maximum payment, respectively.

\bledit{
\begin{remark}[Return on Investment Constraints]
\label{remark:ROI}
Another common type of autobidding constraint is an (aggregate) ROI constraint, which bounds the ratio between the total number of allocated clicks and the total payment.  We note that, for FPPE, a maximum bid $\reportval_i$ is equivalent to an aggregate ROI constraint.  Indeed, due to the ``highest-bang-per-buck'' property of FPPE, all items allocated to an autobidder have the same ROI, which is equal to the equilibrium bid.  So an ROI constraint of the form $t_i \leq \gamma \sum_j \phi_{ij}x_{ij}$ excludes equilibrium bids higher than $\gamma$, and a maximum bid of $\reportval_i$ guarantees an average payment of at most $\reportval_i$ per click.
%i.e., the ratio between the total number of allocated clicks and the total payment is at most $\reportval_i$. Specifically, due to the ``highest-bang-per-buck'' property, for each autobidder in FPPE, all items allocated to her achieve the same ROI, and thus the aggregate ROI is satisfied as desired.
\end{remark}
}

%\yfedit{The maximum bid $\reportval_i$ can also be interpreted as an aggregate return-of-investment (ROI) constraint, i.e., the ratio between the total number of allocated clicks and the total payment is at most $\reportval_i$. Specifically, due to the ``highest-bang-per-buck'' property, for each autobidder in FPPE, all items allocated to her achieve the same ROI, and thus the aggregate ROI is satisfied as desired.}

% An interpretation of the metagame is as follows.
% In the digital advertising market, a platform has a set of $n$ advertisers (agents)
% and $m$ impressions (items).
% The platform knows the click-through rate $\ctrij$ for each pair of advertiser $i$ and impression $j$.
% However, the platform does not know all other parameters in
% the general utility model $\{(\ctri, \valuefunctioni, \truewealthi, \moneycosti)\}_{i\in[n]}$
% of the advertisers.
% Hence,
% the platform pretends each advertiser
% as if she has a budgeted utility,
% asks her to report her value per click as well as her hard budget,
% and then implements the allocation and payment
% induced by FPPE given advertisers' reports.

With slight abuse of notations, we use $\util_i(\reportval_i, \reportwealth_i, \reportval_{-i}, \reportwealth_{-i})$,
$\alloc_i(\reportval_i, \reportwealth_i, \reportval_{-i}, \reportwealth_{-i})$, $\payment_i(\reportval_i, \reportwealth_i, \reportval_{-i}, \reportwealth_{-i})$ to represent the utility, allocation and payment of agent $i$ in the metagame when agent $i$ reports message $(\reportval_i, \reportwealth_i)$ and other agents report message $(\reportval_{-i}, \reportwealth_{-i}) \triangleq \{(\reportval_{i'}, \reportwealth_{i'})\}_{i'\in[n]\backslash\{i\}}$.\footnote{Throughout the paper, we use notation $-i$ to denote other $n-1$ agents excluding agent $i$.}
Similarly, we use $\pricej(\reportval_i, \reportwealth_i, \reportval_{-i}, \reportwealth_{-i})$ to represent the per-unit price of item $j$ of the inner FPPE given message profile $\{(\reportval_i, \reportwealth_i), (\reportval_{-i}, \reportwealth_{-i})\}$.
Now we present a useful lemma that characterizes the relationship between agents' utility, payment, and per-unit prices. Essentially, it suggests that for each agent, her payment along with the per-unit prices of the inner FPPE serve as sufficient statistics for computing her utility.
Its proof is straightforward given the definitions of FPPE and metagame, and is deferred to \Cref{apx:computeutility}.
\begin{restatable}{lemma}{computeutility}
\label{lem:compute utility}
    In the metagame, for every agent $i$ with type $\{\valuefunctioni, \truewealthi,\moneycosti\}$
    and every message profile,
    % $\{(\reportval_i, \reportwealth_i), (\reportval_{-i}, \reportwealth_{-i})\}$,
    let $\paymenti$ be the payment of agent $i$ and $\price$
    be the per-unit prices of the inner FPPE.
    Then agent $i$'s utility $\util_i$ satisfies
    \begin{align*}
        \util_i =
    \left\{
    \begin{array}{ll}
     \valuefunctioni\left(
        \left(
        \max_{j\in[m]}\frac{\ctrij}{\pricej}\right)\cdot \paymenti\right) -
        \moneycosti\left(
        \paymenti
        \right)
        &
    \quad \text{if } \paymenti \leq \truewealthi\\
    -\infty     &
     \quad \text{if } \paymenti > \truewealthi
    \end{array}
    \right.
    \end{align*}
\end{restatable}

% \Cref{lem:compute utility} suggests that for each agent $i$, her payment $\paymenti$ along with the per-unit prices $\price$ serve as sufficient statistics for computing her utility.

\xhdr{Equilibria in the metagame.}
We are interested in Nash equilibria for agents with public types.~
%\emph{pure Nash equilibrium} and \emph{mixed Nash equilibrium} for the metagame
%$\{(\valuefunctioni, \truewealthi, \moneycosti)\}_{i\in[n]}$.
\footnote{In \Cref{sec:bayesian}, we extend our model to the Bayesian environment with private types.}
%from the full information environments to the Bayesian environments where agents' \emph{private} types are drawn from prior distributions.}

\begin{definition}%[Pure Nash equilibrium]
%\asmargincomment{combined defns for pure and mixed, saved a few lines.}
Consider agents with types $\{(\valuefunctioni, \truewealthi, \moneycosti)\}_{i\in[n]}$.
A \emph{pure Nash equilibrium} is a message profile
    $\{(\reportval_i, \reportwealth_i)\}_{i\in[n]}$ such that
    for every agent $i$ and every
    message $(\reportval_i\primed, \reportwealth_i\primed)$,
    \begin{align*}
        \util_i(\reportval_i, \reportwealth_i, \reportval_{-i}, \reportwealth_{-i})
        \geq
        \util_i(\reportval_i\primed, \reportwealth_i\primed, \reportval_{-i}, \reportwealth_{-i}).
    \end{align*}
%\end{definition}
%\begin{definition}[Mixed Nash equilibrium]
 %   For agents with types $\{(\valuefunctioni, \truewealthi, \moneycosti)\}_{i\in[n]}$,
 A \emph{mixed Nash equilibrium} is a randomized message profile\footnote{Throughout the paper, we use bold symbols (e.g., $\randomreportval,\randomreportwealth$) to denote random variables and their corresponding distributions.} $\{(\randomreportval_i, \randomreportwealth_i)\}_{i\in[n]}$ such that
the random messages are mutually independent, and for every agent $i$ and every message $(\reportval_i\primed, \reportwealth_i\primed)$,
    \begin{align*}
        \expect[
        (\reportval_{i},\reportwealth_{i})
        \sim (\randomreportval_{i}, \randomreportwealth_{i}),
        (\reportval_{-i},\reportwealth_{-i})
        \sim (\randomreportval_{-i}, \randomreportwealth_{-i})
        ]
        {\util_i(\reportval_i, \reportwealth_i, \reportval_{-i}, \reportwealth_{-i})}
        \geq
        \expect[
        (\reportval_{-i},\reportwealth_{-i})
        \sim (\randomreportval_{-i}, \randomreportwealth_{-i})
        ]{
        \util_i(\reportval_i\primed, \reportwealth_i\primed, \reportval_{-i}, \reportwealth_{-i})
        }
    \end{align*}
%    Moreover, random messages $(\randomreportval_i, \randomreportwealth_i)$ and $(\randomreportval_{-i}, \randomreportwealth_{-i})$ are independent,
\end{definition}
As a sanity check, in both pure/mixed Nash equilibrium, the utility of every agent $i$ is non-negative, since zero utility is always guaranteed by reporting message $(\reportval_i = 0, \reportwealth_i = 0)$.

\xhdr{Liquid welfare.} We evaluate a particular allocation in terms of its \emph{liquid welfare}. 

%\xhdr{Liquid welfare and price of anarchy.} We evaluate the performance of the metagame by \emph{price of anarchy}  in terms of the \emph{liquid welfare}.}

\begin{definition}[Liquid welfare]
For agents with types $\{(\valuefunctioni, \truewealthi, \moneycosti)\}_{i\in[n]}$, the \emph{liquid welfare} of a (possibly) randomized allocation $\randomalloc$ is defined as
\begin{align*}
    \twelfare(\randomalloc) \triangleq \sum\nolimits_{i\in[n]} \wtp_i(\randomalloc_i),
\text{ where }
  \wtp_i(\randomalloc_i) \triangleq \min\left\{\truewealthi,\;
\moneycost_i^{-1}\left(
\expect[\alloci\sim\randomalloc_{i}]{
\valuefunctioni\left({\textstyle \sum_{j\in[m]}}\;
\ctrij\allocij \right)}
\right)
\right\}.
\end{align*}
Here $\wtp_i(\randomalloc_i)$ is
agent $i$'s \emph{willingness to pay} for allocation $\randomalloc_i$.
\footnote{Since our assumption on money cost function $\moneycosti$ implies that either $\truewealthi < \infty$ or $\lim_{\paymenti\rightarrow \infty} \moneycosti(\paymenti) = \infty$, and $\expect[\alloci\sim\randomalloc_{i}]{
\valuefunctioni(\sum_{j\in[n]}
\ctrij\allocij)}$
is bounded for every feasible randomized allocation $\randomalloci$, 
agent $i$'s willingness to pay
$\wtp_i(\randomalloci)$ is well-defined and finite.}
% i.e., $\wtp_i(\randomalloci)\not = \infty$.}
\end{definition}

%\begin{align*}
%  \wtp_i(\randomalloc_i) \triangleq \min\left\{\truewealthi,\;
%\moneycost_i^{-1}\left(
%\expect[\alloci\sim\randomalloc_{i}]{
%\valuefunctioni\left({\textstyle \sum_{j\in[m]}}\;
%\ctrij\allocij \right)}
%\right)
%\right\}
%\end{align*}

%\asdelete{This general definition of liquid welfare encompasses the classic definition of welfare for linear agents and the original definition of liquid welfare for budgeted agents, respectively \citep{vic-61,DP-14}. It is also closely related to the notion of \emph{compensating variation} in economics \citep{hic-75} (see \Cref{footnote:compensating variation}). Finally,}
%\asmargincomment{we've said it in the Intro!}

Within the set of feasible randomized allocations, the optimal allocation that maximizes the liquid welfare is deterministic, thanks to the weak concavity of the valuation function $\valuefunctioni$.

%Based on the definition of transferable welfare, w

We evaluate the metagame via an approximation of liquid welfare. Specifically, we compare liquid welfare of the worst equilibria against that of the best allocation, over all instances.\footnote{\yfedit{Though the inner FPPE of the metagame assumes that each autobidder uses a single pacing multiplier and conducts ``linear bidding'', the price of anarchy compares its equilibrium efficiency with the unrestricted optimal allocation. Our approximation results hold in spite of the restriction to linear bidding in the inner FPPE.}}

%We evaluate the performance of the metagame by measuring the approximation of liquid welfare. This is done by comparing the worst allocation among all possible equilibria with the optimal allocation, and taking the supremum over all instances.
\begin{definition}[Price of anarchy]
    The \emph{price of anarchy (PoA)} of the metagame $\purePoA$
    (resp., $\mixedPoA$)
    under pure (resp., mixed) Nash equilibrium is
    \begin{align*}
        \purePoA \triangleq
        \sup_{n, m, \ctr}
        \sup_{\{\valuefunctioni, \truewealthi, \moneycosti\}_{i\in[n]}}
        \frac
        {\max_{\alloc}\twelfare(\alloc)}
        {\inf_{\alloc \in \texttt{Pure}}\twelfare(\alloc)},
        \quad
        \mixedPoA \triangleq
        \sup_{n, m, \ctr}
        \sup_{\{\valuefunctioni, \truewealthi, \moneycosti\}_{i\in[n]}}
        \frac
        {\max_{\alloc}\twelfare(\alloc)}
        {\inf_{\randomalloc \in \texttt{Mixed}}\twelfare(\randomalloc)}
    \end{align*}
    where $\texttt{Pure}$ (resp., $\texttt{Mixed}$)
    is the set of the deterministic (resp., randomized) allocation profile
    induced by all pure (mixed) Nash equilibrium
    given types $\{\valuefunctioni, \truewealthi, \moneycosti\}_{i\in[n]}$.
\end{definition}

\section{Warmup: Pure Nash Equilibrium in Examples}
\label{sec:example}

{Let us revisit \Cref{example:linear agents,example:budgeted agents} and analyze pure Nash equilibria therein. We establish the following:}
% \asdelete{To reinforce the understanding of our metagame, in this section, we revisit the previously mentioned examples of linear agents and budgeted agents (see \Cref{example:linear agents,example:budgeted agents}) and analyze the pure Nash equilibrium of these two examples in the metagame. Before delving into the detailed discussions, we would like to summarize the main message conveyed in this section.}

\begin{proposition}
\label{prop:non existencee}
In the metagame, pure Nash equilibria {might not exist}, even for budgeted agents. {When they do exist,} 
%Even if a pure Nash equilibrium exists, 
the resulting per-unit prices may not be unique, even for linear agents, {due to the multiplicity of equilibria.}
\end{proposition}

We interpret this proposition as follows. First, although the inner FPPE of the metagame always exists and its induced per-unit prices are unique (as per \Cref{lem:FPPE uniqueness}), incorporating the strategic behavior of the advertisers {in the meta-game}
% \asdelete{by allowing them to strategically declare budget constraint and maximum bid constraint to their autobidder} 
complicates the allocation and payment outcome. {Second, due to the non-existence result, even budgeted agents may have incentive to misreport their types.}\footnote{{For this point, it is crucial that the budgeted agents are utility-maximizers. Indeed, in the meta-game for budgeted \emph{value}-maximizing agents, truthful reporting is a Nash equilibrium (this follows from \citealp{AMP-23}).}} \yfedit{Finally, despite the non-existence result, we show in \Cref{sec:single item pure nash} that a pure Nash equilibrium always exists for single-item instances and can be computed in polynomial time.}
% \asdelete{Another immediate corollary of the non-existence result is that even budgeted agents, as utility maximizers, have an incentive to misreport their true types to their autobidder.}

% \asdelete{\begin{corollary}
% \label{coro:budgeted non AIC}
%     For budgeted agents, it is possible that reporting their types truthfully
%     is not a pure Nash equilibrium in the metagame.
% \end{corollary}

% \begin{remark}
%     \citet{AMP-23} shows that for \emph{value-maximizing} agents with hard budget constraints, truthfully reporting to value-maximizing autobidders is an equilibrium in the first-price auction with uniform bidding (i.e., FPPE).
%     In contrast, \Cref{coro:budgeted non AIC} suggests that when agents are \emph{utility-maximizing} with hard budget constraints, reporting the true budgets to value-maximizing autobidders may not be an equilibrium.
% \end{remark}
% }

% \asdelete{In the following, \Cref{sec:pure nash auxiliary lemmas} presents two auxiliary lemmas as useful tools for our analysis of the pure Nash equilibrium in \Cref{example:linear agents,example:budgeted agents}. In \Cref{sec:pure nash linear agents}, we illustrate that the pure Nash equilibrium is not unique in \Cref{example:linear agents}, while in \Cref{sec:pure nash budgeted agents}, we illustrate that the pure Nash equilibrium does not exist in \Cref{example:budgeted agents}.}

% \asdelete{\subsection{Two Auxiliary Lemmas}
% \label{sec:pure nash auxiliary lemmas}
% In this subsection, we present two auxiliary lemmas (\Cref{lem:reporting infinite value,lem:agent in favor tie-breaking}) that are useful}
{We start with two auxiliary lemmas}
on verifying the existence of profitable deviations in a pure Nash equilibrium. The first lemma suggests that 
% \asdelete{for verifying the existence of profitable deviations} 
for a given agent $i$, it suffices to consider deviations $(\reportval_i\primed,\reportwealth_i\primed)$ with a restriction that $\reportval_i\primed = \infty$, i.e., no maximum bid constraint for her autobidder. Loosely speaking, restricting to deviation with $\reportval_i\primed = \infty$ simplifies the analysis since the agent would exhaust her reported budget~$\reportwealth_i\primed$. Consequently, \Cref{lem:compute utility} ensures that $\reportwealth_i\primed$ along with the per-unit prices of the inner FPPE serve as sufficient statistics for computing her utility.
Its proof is deferred to \Cref{apx:reportinginfinitevalue}.
\begin{restatable}{lemma}{reportinginfinitevalue}
\label{lem:reporting infinite value}
    In the metagame, for every agent $i$ and every pure Nash equilibrium $\{(\reportval_i, \reportwealth_i), (\reportval_{-i}, \reportwealth_{-i})\}$, 
    %it satisfies that
    \begin{align*}
        \util_i(\reportval_i, \reportwealth_i, \reportval_{-i}, \reportwealth_{-i})
        =
        \max_{\reportval_i\primed,\reportwealth_i\primed}
        \util_i(\reportval_i\primed,\reportwealth_i\primed, \reportval_{-i}, \reportwealth_{-i})
        =
        \max_{\reportwealth_i\primed}
        \util_i(\infty,\reportwealth_i\primed, \reportval_{-i}, \reportwealth_{-i}).
    \end{align*}
\end{restatable}

The second lemma suggests that, for a given agent $i$, it suffices 
% \asdelete{when verifying the existence of profitable deviations for agent $i$, it is sufficient} 
to consider the tie-breaking rule of the inner FPPE that favors agent $i$.
Its proof is straightforward given \Cref{lem:reporting infinite value}, see \Cref{apx:favortiebreaking}.
\begin{restatable}{lemma}{favortiebreaking}
\label{lem:agent in favor tie-breaking}
    In the metagame with tie-breaking rule $\allocselectionrule$ for the inner FPPE, for every agent $i$ and every pure Nash equilibrium $\{(\reportval_i, \reportwealth_i), (\reportval_{-i}, \reportwealth_{-i})\}$, it satisfies that
    \begin{align*}
        \util_i^\allocselectionrule(\reportval_i, \reportwealth_i, \reportval_{-i}, \reportwealth_{-i})
        =
        \max_{\allocselectionrule'}
        \util_i^{\allocselectionrule'}(\reportval_i, \reportwealth_i, \reportval_{-i}, \reportwealth_{-i})
    \end{align*}
    where
    % $\util_i^\allocselectionrule(\cdot,\cdot,\cdot,\cdot)$
    $\util_i^\allocselectionrule(\reportval_i, \reportwealth_i, \reportval_{-i}, \reportwealth_{-i})$
    and
    $\util_i^{\allocselectionrule'}(\reportval_i, \reportwealth_i, \reportval_{-i}, \reportwealth_{-i})$
    % $\util_i^{\allocselectionrule\primed}(\cdot,\cdot,\cdot,\cdot)$
    are agent $i$'s utility in the metagame
    when tie-breaking rules $\allocselectionrule$, $\allocselectionrule\primed$ are selected for the inner FPPE, respectively;
    and the maximization on the right-hand side is taken over all randomized and deterministic tie-breaking rules.
\end{restatable}

\subsection{\texorpdfstring{\Cref{example:linear agents}}{}: Non-Uniqueness of Pure Nash Equilibrium}
\label{sec:pure nash linear agents}

% \asdelete{In this subsection, we revisit the linear agents instance defined in \Cref{example:linear agents}:
% Consider a scenario with two linear agents and one item. Let us assume that the click-through rates are the same for both agents, i.e., $\ctr_{11} = \ctr_{21} = 1$. Additionally, we assume that the value per click $\val_1$ for agent 1 is greater than or equal to the value per click $\val_2$ for agent 2, i.e., $\val_1 \geq \val_2$.
% For this example,}\ascomment{No need to restate the example, IMO.}

{We revisit the linear agents instance from \Cref{example:linear agents}.} We prove that there exists an efficient pure Nash equilibrium (\Cref{claim:linear efficient}),
% \asdelete{for all $\val_1 \geq \val_2$,},
and also an inefficient pure Nash equilibrium if $\val_2$ is close to $\val_1$ (\Cref{claim:linear inefficient}). In both equilibria, agents report their value per click (aka., maximum bid) truthfully, while strategically declaring their budgets.

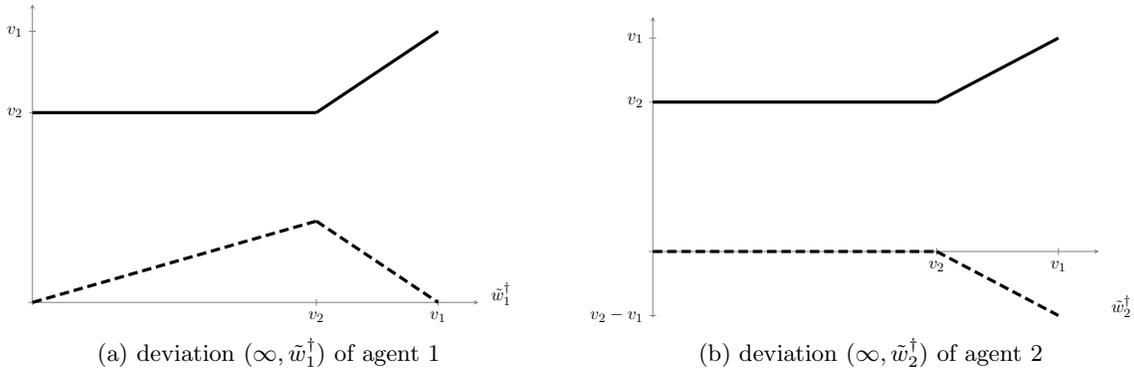
\begin{figure}[hbt]
\centering
\subfloat[deviation $(\infty,\reportwealth_1\primed)$ of agent 1]{
\begin{tikzpicture}[scale=0.6, transform shape]
\begin{axis}[
axis line style=gray,
axis lines=middle,
xlabel = {$\reportwealth_1\primed$},
xlabel style={at={(axis description cs:1.05,-0.02)}, anchor=south},
xtick={0, 0.7, 1},
ytick={0, 0.7, 1},
xticklabels={0, $\val_2$, $\val_1$},
yticklabels={0, $\val_2$, $\val_1$},
xmin=-0.01,xmax=1.1,ymin=-0.01,ymax=1.1,
width=0.7\textwidth,
height=0.5\textwidth]

\addplot[domain=0:0.7, line width=0.7mm] (x, {0.7});
\addplot[domain=0.7:1, line width=0.7mm] (x, {x});

\addplot[domain=0:0.7, dashed, dash pattern=on 6pt off 3pt on 6pt off 3pt, line width=0.7mm] (x, {x/0.7*1-x});
\addplot[domain=0.7:1, dashed, dash pattern=on 6pt off 3pt on 6pt off 3pt, line width=0.7mm] (x, {1-x});

\end{axis}

% \begin{axis}[
%     axis line style=gray,
%     axis x line=none, % no x-axis line
%     axis y line=right, % y-axis line on the right
%     xmin=-0.01, xmax=1., % x-axis limits
%     ymin=-0.01, ymax=1.1, % y-axis limits
%     width=0.6455\textwidth,
%     height=0.5\textwidth, % size of the plot
%     ticks=none, % remove ticks for the secondary axis
% ]

% \addplot[domain=0:0.7, dotted, line width=0.7mm] (x, {x/0.7});
% \addplot[domain=0.7:1, dotted, line width=0.7mm] (x, {1});
% % Your plot commands here
% \end{axis} 

\end{tikzpicture}
\label{fig:linear agent 1 efficient deviation}
}
\quad
\subfloat[deviation $(\infty,\reportwealth_2\primed)$ of agent 2]{
\begin{tikzpicture}[scale=0.6, transform shape]
\begin{axis}[
axis line style=gray,
axis lines=middle,
xlabel = {$\reportwealth_2\primed$},
xlabel style={at={(axis description cs:1.05,-0.02)}, anchor=south},
xtick={0, 0.7, 1},
ytick={-0.3, 0, 0.7, 1},
xticklabels={0, $\val_2$, $\val_1$},
yticklabels={$\val_2 - \val_1$, 0, $\val_2$, $\val_1$},
xmin=-0.01,xmax=1.1,ymin=-0.31,ymax=1.1,
width=0.7\textwidth,
height=0.5\textwidth]

\addplot[domain=0:0.7, line width=0.7mm] (x, {0.7});
\addplot[domain=0.7:1, line width=0.7mm] (x, {x});
% \addplot[domain=0.7:1, thick] (x, {x});

\addplot[domain=0:0.7, dashed, dash pattern=on 6pt off 3pt on 6pt off 3pt, line width=0.7mm] (x, {x/0.7*0.7-x});
\addplot[domain=0.7:1, dashed, dash pattern=on 6pt off 3pt on 6pt off 3pt, line width=0.7mm] (x, {0.7-x});
% \addplot[domain=0.7:1, dashed, thick] (x, {0.7-x});

\end{axis}

% \begin{axis}[
%     axis line style=gray,
%     axis x line=none, % no x-axis line
%     axis y line=right, % y-axis line on the right
%     ytick={0.2, 0.5},
%     yticklabels={1, $1$},
%     xmin=-0.01, xmax=0.7, % x-axis limits
%     ymin=-0.01, ymax=0.8, % y-axis limits
%     width=0.6245\textwidth,
%     height=0.5\textwidth, % size of the plot
%     ticks=none, % remove ticks for the secondary axis
% ]

% \addplot[domain=0:0.7, dotted, line width=0.7mm] (x, {x});
% % Your plot commands here
% \end{axis} 

\end{tikzpicture}
\label{fig:linear agent 2 efficient deviation}
}
\caption{Graphical illustration of each agent $i$'s deviation in \Cref{claim:linear efficient} for \Cref{example:linear agents}. The solid (resp., dashed) line is the per-unit price (resp., utility of agent $i$).}
\label{fig:linear efficient deviation}
\end{figure}

\begin{claim}
\label{claim:linear efficient}
In \Cref{example:linear agents}, a pure Nash equilibrium {of the metagame} is achieved when agent 1 reports $\reportval_1 = \val_1$ and $\reportwealth_1 = \val_2$, while agent 2 reports $\reportval_2 = \val_2$ and $\reportwealth_2 = \val_2$.
\end{claim}
\begin{proof}
In the following argument, we assume $\val_1 > \val_2$. However, the same argument can be applied in the case of $\val_1 = \val_2$ due to \Cref{lem:agent in favor tie-breaking}. To save space, we will omit the details of this case.

Given message profile $((\reportval_1 = \val_1, \reportwealth_1 = \val_2), (\reportval_2 = \val_2, \reportwealth_2 = \val_2))$, the inner FPPE allocates the entire item to agent 1 with per-unit price $\val_2$, i.e., $\alloc_{11} = 1$, $\alloc_{21} = 0$ and $\price_1 = \val_2$.

We verify the non-existence of profitable deviation. Invoking \Cref{lem:reporting infinite value}, it is sufficient to consider deviation $(\reportval_i\primed,\reportwealth_i\primed)$ with $\reportval_i\primed = \infty$ for each agent $i$. In such a deviation, the ``budgeted feasibility'' property of FPPE ensures that agent $i$ exhausts her budget and pays exactly $\reportwealth_i\primed$. We now proceed to analyze each agent $i$ individually.
For agent $1$, it can be verified that the per-unit price $\price_1\primed$ and allocation $\alloc\primed_{11}$ of inner FPPE under her deviation $(\infty, \reportwealth_1\primed)$ are
\begin{align*}
    \price_1\primed =
    \begin{cases}
       \val_2  &  \text{ if } \reportwealth_1\primed \leq \val_2\\
       \reportwealth_1\primed  & \text{ if } \reportwealth_2\primed \geq \val_2
    \end{cases},
    \qquad
    \alloc\primed_{11} =
    \frac{\reportwealth_1\primed}{\price_1\primed}
    =
    \begin{cases}
    \frac{\reportwealth_1\primed}{\val_2}     &  \text{ if } \reportwealth_1\primed \leq \val_2\\
    1     & \text{ if } \reportwealth_2\primed \geq \val_2
    \end{cases}
\end{align*}
and her utility is maximized at $\reportwealth_1\primed = \val_2$, which results in the same utility as she obtains in equilibrium. See \Cref{fig:linear agent 1 efficient deviation} for a graphical illustration.
To avoid repetition, we omit a similar argument for agent 2.  See \Cref{fig:linear agent 2 efficient deviation} for a graphical illustration.
\end{proof}

\begin{figure}[hbt]
\centering
\subfloat[deviation $(\infty,\reportwealth_1\primed)$ of agent 1]{
\begin{tikzpicture}[scale=0.55, transform shape]
\begin{axis}[
axis line style=gray,
axis lines=middle,
xlabel = {$\reportwealth_1\primed$},
xlabel style={at={(axis description cs:1.05,-0.02)}, anchor=south},
xtick={0, 0.24221453, 0.53045, 0.7, 1},
ytick={0, 0.3460210084, 0.7, 1},
xticklabels={0, $\gamma\val_1$, $\val_2-\gamma\val_2$, $\val_2$, $\val_1$},
yticklabels={0, $\frac{\reportwealth_1}{\reportwealth_1 + \reportwealth_2}\val_1 - \reportwealth_1$, $\val_2$, $\val_1$},
xmin=-0.01,xmax=1.1,ymin=-0.01,ymax=1.1,
width=0.7\textwidth,
height=0.5\textwidth]

\addplot[domain=0:0.53045, line width=0.7mm] (x, {0.16955 + x});
\addplot[domain=0.53045:0.7, line width=0.7mm] (x, {0.7});
\addplot[domain=0.7:1, line width=0.7mm] (x, {x});

\addplot[domain=0:0.53045, dashed, dash pattern=on 6pt off 3pt on 6pt off 3pt, line width=0.7mm] (x, {x/(0.16955 + x)*1-x});
\addplot[domain=0.53045:0.7, dashed, dash pattern=on 6pt off 3pt on 6pt off 3pt, line width=0.7mm] (x, {x/0.7*1-x});
\addplot[domain=0.7:1, dashed, dash pattern=on 6pt off 3pt on 6pt off 3pt, line width=0.7mm] (x, {x/x*1-x});

\addplot[domain=0:0.24221453, dotted, gray] (x, {0.3460210084});
\addplot[domain=0.24221453:1, dotted, gray] (x, {0.3460210084});

\end{axis}

% \begin{axis}[
%     axis line style=gray,
%     axis x line=none, % no x-axis line
%     axis y line=right, % y-axis line on the right
%     xmin=-0.01, xmax=1., % x-axis limits
%     ymin=-0.01, ymax=1.1, % y-axis limits
%     width=0.6455\textwidth,
%     height=0.5\textwidth, % size of the plot
%     ticks=none, % remove ticks for the secondary axis
% ]

% \addplot[domain=0:0.53045, dotted, line width=0.7mm] (x, {x/(0.16955 + x)});
% \addplot[domain=0.53045:0.7, dotted, line width=0.7mm] (x, {x/0.7});
% \addplot[domain=0.7:1, dotted, line width=0.7mm] (x, {x/x});
% % Your plot commands here
% \end{axis} 

\end{tikzpicture}
\label{fig:linear agent 1 inefficient deviation}
}
\quad
\subfloat[deviation $(\infty,\reportwealth_2\primed)$ of agent 2]{
\begin{tikzpicture}[scale=0.55, transform shape]
\begin{axis}[
axis line style=gray,
axis lines=middle,
xlabel = {$\reportwealth_2\primed$},
xlabel style={at={(axis description cs:1.05,-0.02)}, anchor=south},
xtick={0, 0.16955017301, 0.757785467, 1},
ytick={-0.3, 0, 0.1186851211, 0.7, 1},
xticklabels={0, $\gamma\val_2$, $\val_1-\gamma\val_1$, $\val_1$},
yticklabels={$\val_2 - \val_1$, 00, $\frac{\reportwealth_2}{\reportwealth_1 + \reportwealth_2}\val_2 - \reportwealth_2$, $\val_2$, $\val_1$},
xmin=-0.01,xmax=1.1,ymin=-0.31,ymax=1.1,
width=0.7\textwidth,
height=0.5\textwidth]

\addplot[domain=0:0.757785467, line width=0.7mm] (x, {0.24221453 + x});
\addplot[domain=0.757785467:1, line width=0.7mm] (x, {1});

\addplot[domain=0:0.757785467, dashed, dash pattern=on 6pt off 3pt on 6pt off 3pt, line width=0.7mm] (x, {x/(0.24221453 + x)*0.7-x});
\addplot[domain=0.757785467:1, dashed, dash pattern=on 6pt off 3pt on 6pt off 3pt, line width=0.7mm] (x, {x/1*0.7-x});

\addplot[domain=0:1, dotted, gray] (x, {0.1186851211});

\end{axis}

% \begin{axis}[
%     axis line style=gray,
%     axis x line=none, % no x-axis line
%     axis y line=right, % y-axis line on the right
%     ytick={0.2, 0.5},
%     yticklabels={1, $1$},
%     xmin=-0.01, xmax=1, % x-axis limits
%     ymin=0.3, ymax=1.1, % y-axis limits
%     width=0.6455\textwidth,
%     height=0.4\textwidth, % size of the plot
%     ticks=none, % remove ticks for the secondary axis
% ]

% \addplot[domain=0:0.7, dotted, line width=0.7mm] (x, {x});
% % Your plot commands here
% \end{axis} 

\end{tikzpicture}
\label{fig:linear agent 2 inefficient deviation}
}
\caption{Graphical illustration of each agent $i$'s deviation in \Cref{claim:linear inefficient} for \Cref{example:linear agents}. The solid (dashed) line is the per-unit price (utility of agent $i$).}
\label{fig:linear inefficient deviation}
\end{figure}
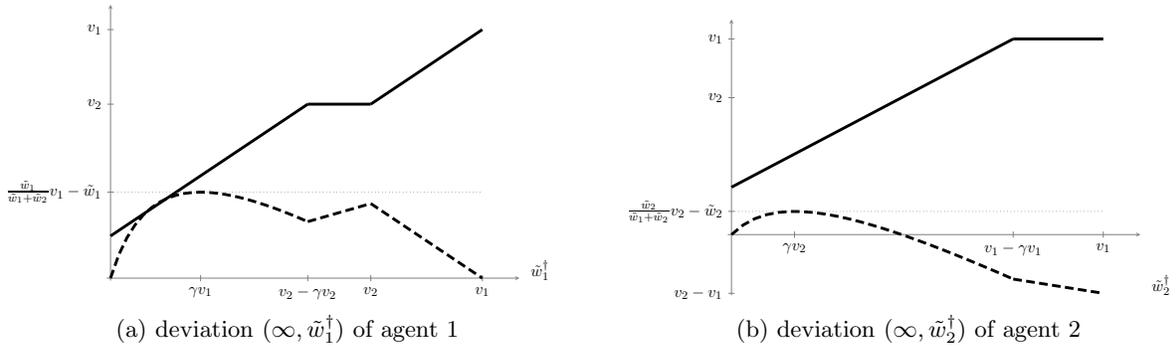

\begin{claim}
\label{claim:linear inefficient}
\label{claim:quasi linear inefficient equilibrium}
{Consider the metagame in \Cref{example:linear agents}
with} $\frac{\sqrt{5} - 1}{2}\val_1 \leq \val_2 \leq \val_1$. Let $\gamma = \frac{\val_1\val_2}{(\val_1 + \val_2)^2}$. Then 
%, in \Cref{example:linear agents}, 
a pure Nash equilibrium is achieved when agent 1 reports $\reportval_1 = \val_1$ and $\reportwealth_1 = \gamma \val_1$, while agent 2 reports $\reportval_2 = \val_2$ and $\reportwealth_2 = \gamma \val_2$.
\end{claim}
\begin{proof}
In the following argument, we assume $\val_1 > \val_2$. However, the same argument can be applied in the case of $\val_1 = \val_2$ due to \Cref{lem:agent in favor tie-breaking}. To save space, we will omit the details of this case.

Given message profile $((\reportval_1 = \val_1, \reportwealth_1 = \gamma\val_1), (\reportval_2 = \val_2, \reportwealth_2 = \gamma\val_2))$, the per-unit price $\price_1$ of the inner FPPE is $\price_1 = \reportwealth_1 + \reportwealth_2 = \gamma (\val_1 + \val_2) \leq \val_2 < \val_1$. Moreover, the item is allocated to both agents in proportion to their respective values, i.e., $\alloc_{i1} = \sfrac{\val_i}{(\val_1 + \val_2)}$.

We verify the non-existence of profitable deviation. Invoking \Cref{lem:reporting infinite value}, it is sufficient to consider deviation $(\reportval_i\primed,\reportwealth_i\primed)$ with $\reportval_i\primed = \infty$ for each agent $i$. In such a deviation, the ``budgeted feasibility'' property of FPPE ensures that agent $i$ exhausts her budget and pays exactly $\reportwealth_i\primed$. We now proceed to analyze each agent $i$ individually.
For agent $1$, it can be verified that the per-unit price $\price_1\primed$ and allocation $\alloc\primed_{11}$ of inner FPPE under her deviation $(\infty, \reportwealth_1\primed)$ are
\begin{align*}
    \price_1\primed =
    \begin{cases}
       \reportwealth_1\primed + \reportwealth_2  &  \text{ if } \reportwealth_1\primed \leq \val_2 - \reportwealth_2\\
       \val_2  & \text{ if } \val_2 - \reportwealth_2 \leq \reportwealth_2\primed \leq \val_2\\
       \reportwealth_1\primed & \text{ if } \reportwealth_2\primed \geq \val_2
    \end{cases},
    \qquad
    \alloc\primed_{11} =
    \frac{\reportwealth_1\primed}{\price_1\primed}
    =
    \begin{cases}
       \frac{\reportwealth_1\primed}{\reportwealth_1\primed + \reportwealth_2}  &  \text{ if } \reportwealth_1\primed \leq \val_2 - \reportwealth_2\\
       \frac{\reportwealth_1\primed}{\val_2}  & \text{ if } \val_2 - \reportwealth_2 \leq \reportwealth_2\primed \leq \val_2\\
       1 & \text{ if } \reportwealth_2\primed \geq \val_2
    \end{cases}
\end{align*}
By considering the first-order condition, we observe that agent $1$'s utility, under the mentioned deviation, has two local maximizers: $\reportwealth_1\primed = \gamma\val_1$ and $\reportwealth_1\primed = \val_2$. However, according to the claim assumption that $\sfrac{(\sqrt{5} - 1)\val_1}{2}\leq \val_2$, her utility is maximized at $\reportwealth_1\primed = \gamma\val_1$, which coincides with her utility in the equilibrium. See \Cref{fig:linear agent 1 inefficient deviation} for a graphical illustration.
To avoid repetition, we omit a similar argument for agent 2. See \Cref{fig:linear agent 2 inefficient deviation} for a graphical illustration.
\end{proof}

\begin{remark}
    At the equilibrium of the metagame, each agent $i$ faces a Pareto curve (see \Cref{fig:linear efficient deviation,fig:linear inefficient deviation}) that describes how much value they receive as their reported budget $\reportwealth_i$ increases.  Increasing her reported budget $\reportwealth_i$ causes more impressions and thus clicks to be won, but at a potentially decreasing bang-per-buck. In particular, as we mentioned in \Cref{sec:prelim}, FPPE can be viewed as a competitive equilibrium, i.e., market-clearing outcome that assigns a price to each impression.  Since increasing budget has downstream effects on how other autobidders will behave (even keeping the reports of other advertisers fixed), increasing one's budget can cause prices to increase.
\end{remark}

\subsection{\texorpdfstring{\Cref{example:budgeted agents}}{}: Non-Existence of Pure Nash Equilibrium}
\label{sec:pure nash budgeted agents}

% \asdelete{In this subsection, we revisit the budgeted agents instance defined in \Cref{example:budgeted agents}:
% Consider a scenario with two budgeted agents and two items. Let us assume that the click-through rate $\ctr_{ij}$ is given by $\frac{1}{2} + \frac{1}{2}\indicator{i=j}$ for each $i\in[2]$ and $j\in[2]$, i.e., each agent $i$ favors item $i$ than the other item. Furthermore, the value per click is the same for both agents, i.e., $\val_1 = \val_2 = 1$. Both agents have a budget constraint of $\truewealth_1 = \truewealth_2 = \frac{1}{2}$.
We revisit the budgeted agents instance from \Cref{example:budgeted agents} and prove non-existence of pure Nash equilibrium.
% Therefore, both agents have an incentive to misreport their true types to their autobidders.}

\begin{claim}
\label{claim:budget}
{In the metagame from \Cref{example:budgeted agents}, pure Nash equilibrium does not exist. Consequently, both agents have incentive to misreport their true types.}
\end{claim}

% \asdelete{The formal proof of \Cref{claim:budget} relies on the following case-by-case argument:}
{To prove this claim,} we enumerate all feasible allocations and argue that each of them cannot be induced by a pure Nash equilibrium. We distinguish three cases with different deviation strategy for each, see 
\Cref{claim:budget non favorate allocation,claim:budget single allocation,claim:budget two allocation}.
% \asdelete{Given different structure of the allocation, we use different argument by constructing specific deviation strategy. See \Cref{claim:budget non favorate allocation}, \Cref{claim:budget single allocation} and \Cref{claim:budget two allocation} for each case of allocations.}
In the first case, we argue that allocations where each agent receives a positive fraction of her less favored item cannot be induced by a pure Nash equilibrium.
\begin{claim}
    \label{claim:budget non favorate allocation}
    In \Cref{example:budgeted agents}, there exists no pure Nash equilibrium whose induced allocation $\alloc$ satisfies $\alloc_{12} > 0$ and $\alloc_{21} > 0$.
\end{claim}
\begin{proof}
    We prove this by contradiction. Suppose there exists a pure Nash equilibrium as desired. Let $\price_1, \price_2$ be the per-unit price of the inner FPPE. Since $\alloc_{12} > 0$ in FPPE, the ``highest bang-per-buck'' property implies $\frac{\ctr_{12}}{\price_2} \geq \frac{\ctr_{11}}{\price_1}$ and thus $\price_2 \leq \frac{1}{2}\price_1$.
    Similarly, $\alloc_{21} > 0$ in FPPE implies $\price_1 \leq \frac{1}{2}\price_2$.
    Thus, the inner FPPE has zero per-unit prices, i.e., $\price_{1} = \price_{2} = 0$,
    which can only be achieved from message profile $\{(\reportval_i, \reportwealth_i)\}$ where $\reportval_i = 0$ or $\reportwealth_i = 0$ for each agent $i$.
    It is straightforward to verify that $(\reportval_i\primed = \epsilon, \reportwealth_i\primed = \epsilon)$ is a profitable deviation for each agent $i$ with sufficiently small $\epsilon$, which leads to a contradiction.
\end{proof}

In the second case, we argue that the allocation where each agent receives her favored item cannot be induced by a pure Nash equilibrium. Note that this is the allocation induced by FPPE if both agents report their types truthfully.
\begin{claim}
    \label{claim:budget single allocation}
    In \Cref{example:budgeted agents}, there exists no pure Nash equilibrium whose induced allocation $\alloc$ satisfies $\alloc_{11} = \alloc_{22} = 1$.
\end{claim}
\begin{proof}
    We prove this by contradiction. Suppose there exists a pure Nash equilibrium as desired. Let $\price_1, \price_2$ be the per-unit price of the inner FPPE.
    Without loss of generality,\footnote{Due to the symmetric of the instance, the same argument can be applied to $\price_1 \leq \price_2$ and $\price_2 > 0$ as well. The remaining case of $\price_1 = \price_2 = 0$ is already covered in the proof of \Cref{claim:budget non favorate allocation}.} we assume $\price_1 \geq \price_2$ and $\price_1 > 0$.
    The utility of agent 1 is $\util_1 = 1 - \price_1$.
    Consider the following profitable deviation $(\reportval_i\primed = \infty, \reportwealth_i\primed = \frac{1}{2}\price_2)$.
    It can be verified that the per-unit price and allocation
    of inner FPPE under such a deviation are
    $\price_1\primed = \frac{1}{2}\price_2$ and $\price_2\primed = \price_2$,
    and $\alloc_{11} = \alloc_{22} = 1$.
    Consequently, agent 1's utility under such a deviation is $\util_1\primed = 1 - \frac{1}{2}\price_2 \geq \util_1$, which leads to a contradiction.
\end{proof}

In the final case, we argue that allocations where one agent $i$ receives her favored item and a positive fraction of less favored item cannot be induced by a pure Nash equilibrium. The detailed formal proof for this case is complex and will be deferred to \Cref{apx:budgettwoalloc}. At a high-level, we utilize the relation  $\price_i = 2 \price_{1-i}$ on per-unit prices of inner FPPE due to the ``highest bang-per-buck'' property, and $\alloc_{ii} = 1, \alloc_{i, 1-i} > 0$. We then argue that depending on the magnitude of $\price_i$, either agent has a profitable deviation.
\begin{restatable}{claim}{budgettwoalloc}
    \label{claim:budget two allocation}
    In \Cref{example:budgeted agents}, there exists no pure Nash equilibrium whose induced allocation $\alloc$ satisfies $\alloc_{ii} = 1$ and $\alloc_{i,1-i} > 0$ for some agent $i$.
\end{restatable} 

%\section{Main Results: Price of Anarchy in the Metagame}
%\input{Paper/PoA-results}

%\subsection{Tight PoA Bound Under Pure Nash Equilibrium}
\section{Main Result for Pure Nash Equilibrium}
\label{sec:poa pure}
{In this section, we analyze pure equilibria of the metagame, and present a tight bound on the price of anarchy.}
%we present the tight PoA bound under pure Nash equilibrium.
\begin{theorem}
\label{thm:poa pure}
    In the metagame,
    the price of anarchy under pure Nash equilibrium is $\purePoA = 2$.
\end{theorem}

% \asdelete{The lower bound in \Cref{thm:poa pure} is shown in the following example.}

\begin{example}[Lower bound of PoA under pure Nash equilibrium]
\label{example:pure lower bound}
    Consider a scenario with two agents (one linear agent and one budgeted agent) and one item. Let us assume that the click-through rates are the same for both agents, i.e., $\ctr_{11} = \ctr_{21} = 1$. Agent $1$ has a budget utility model with type $\val_1 = K$ and $\truewealth_1 = 1$; while agent $2$ has a linear utility model with type $\val_2 = 1$. Here we assume $K$ is a sufficiently large constant.

    The optimal liquid welfare is $2 - \sfrac{1}{K}$. This is achieved through an allocation where budgeted agent $1$ receives a $\frac{1}{K}$-fraction of the item, and linear agent $2$ receives a $\frac{K-1}{K}$-fraction of the item.
    By employing a similar argument to the one presented in \Cref{claim:linear efficient}, we can be verified that a pure Nash equilibrium is achieved when both agents report their types truthfully: $\reportval_1 = \val_1 = K$, $\reportwealth_1 = \truewealth_1 = 1$, $\reportval_2 = \val_2 = 1$, and $\reportwealth_2 = \infty$. In this equilibrium, the per-unit price of the inner FPPE is $\price_1 = 1$, and agent $1$ receives the entire item. Consequently, the achieved liquid welfare is 1. Letting $K$ approach infinity, the lower bound of PoA under pure Nash equilibrium is obtained as desired.
\end{example}

%Before the proof of the upper bound in \Cref{thm:poa pure},
{In the rest of this section we prove the upper bound in \Cref{thm:poa pure}: $\purePoA \leq 2$. First,} we introduce a characterization of the allocation for each agent and per-unit prices of the inner FPPE.
\begin{lemma}
\label{lem:poa pure allocation price relation}
    In the metagame, for every pure Nash equilibrium, suppose $\price, \alloc, \payment$
    are the per-unit prices, allocation and payment of the induced FPPE.
    For every agent $i$,
    \yfedit{
    let $\purchaseseti \triangleq \{j\in[m]: \alloc_{ij} > 0\}$ be the subset of items for which agent $i$ receives a strictly positive fraction.
    Suppose agent $i$ does not exhaust her true budget $\truewealth_i$,
    i.e., $\payment_i < \truewealth_i$, then
\begin{align*}
\arraycolsep=1.4pt\def\arraystretch{2.2}
    \frac
        {\valuefuncderivative_i
        \left(\sum_{j\in[m]}\ctrij\allocij\right)}
    {\moneycostderivative_i
        \left(\payment_i\right)}
        \leq
        \left\{
        \begin{array}{ll}
         \frac{\sum_{j\in \purchaseseti}\price_j}{\sum_{j\in \purchaseseti}(1-\allocij)\ctrij}
     \qquad \qquad        &  \text{if $\purchaseseti \not= \emptyset$}\\
         \min_{j\in[m]}\frac{\price_j}{\ctrij}    & \text{otherwise (i.e., $\purchaseseti = \emptyset$)}
        \end{array}
     \right.
\end{align*}
}where $\valuefuncderivative_i$ and $\moneycostderivative_i$ are the derivative of valuation function $\valuefunc_i$ and money cost function $\moneycost_i$ defined in \Cref{sec:prelim}, respectively.
\end{lemma}

\yfedit{One interpretation of \Cref{lem:poa pure allocation price relation} is as follows.  First consider a hypothetical setting where item prices are fixed and the agent is acting as a price-taker, as in ``standard'' market equilibrium. Then it would be optimal for the agent to select items until the marginal price equals the marginal value-per-unit (i.e., $\valuefuncderivative(\cdot)$) divided by the marginal cost-per-unit (i.e., $\moneycostderivative(\cdot)$), which can be formulated as the inequality in \Cref{lem:poa pure allocation price relation} with the discounting term $(1-\allocij)$ in the denominator of the right-hand side removed. However, in our setting, the agent is not acting as a price-taker: her behavior distorts the prices, and hence distorts the relationship between her allocation and the prices. \Cref{lem:poa pure allocation price relation} shows that this distortion is proportional to her allocation at equilibrium.  If an agent is taking almost all of the items they care about ($\allocij$ close to 1 for all $j$ in $\purchaseseti$), the denominator on the right-hand side is close to 0 and the distortion is very large. In contrast, as long as a constant fraction of the items remains, the distortion is small.  This is useful for our efficiency analysis: roughly speaking, an agent who gets a small allocation is acting approximately like a price-taker (and thus an approximate first-welfare-theorem analysis applies). On the other hand, an agent who gets a large allocation is anyway making a large contribution to the liquid welfare.
}

\yfdelete{To gain a better understanding behind \Cref{lem:poa pure allocation price relation}, consider a simple instance with one item and linear utilities. In this instance, each agent $i$ has a linear utility model with type $\val_i$, i.e., $\valuefuncderivative_i(\sum_{j\in[m]}\ctrij\allocij) = \val_i$ and $\moneycostderivative_i(\payment_i) = 1$. Moreover, suppose click-through-rate $\ctr_{i1} = 1$ for each agent~$i$. For this instance, the inequality in \Cref{lem:poa pure allocation price relation} can be simplified as $\alloc_{i1} \geq 1 - \sfrac{\price_1}{\val_i}$ which serves as a lower bound for agent $i$'s allocation as a function of per-unit price $\price_1$ of the FPPE and her value per click $\val_i$.}

At a high level, the proof of \Cref{lem:poa pure allocation price relation} argues that when the inequality in the lemma statement is violated, agent $i$ can increase her utility by increasing her reported budget for a sufficiently small amount. To support this argument, we require the following technical lemma concerning the sensitivity of the per-unit prices with respect to the reported budgets in FPPE. Its proof is deferred to \Cref{apx:increasingBudget}.

\begin{restatable}{lemma}{increasingBudget}
\label{lem:fppe increasing budget}
\label{lem:fppe change budget}
In an FPPE, when the budget of an arbitrary agent $i$ is increased from $\wealth_i$ to $\wealth_i\primed$, the per-unit price $\price_j$ for every item $j$ weakly increases. Furthermore, the resulting revenue $\sum_{j\in[m]}\pricej$ increases by at most $\wealth_i\primed$.
\end{restatable}

In the following, we prove \Cref{lem:poa pure allocation price relation} and then  \Cref{thm:poa pure}.

%\begin{proof}[Proof of \Cref{lem:poa pure allocation price relation}]

\subsection{Proof of \texorpdfstring{\Cref{lem:poa pure allocation price relation}}{}}

Fix an arbitrary pure Nash equilibrium and suppose $\price, \alloc, \payment$
are the per-unit prices, allocation and payment of the inner FPPE, respectively.
Fix an arbitrary agent $i$ who does not exhaust her true budget $\truewealth_i$,
i.e., $\payment_i < \truewealth_i$. We consider two cases separately.

\paragraph{Case a. Suppose $\purchaseseti \not=\emptyset$.}
\yfedit{Let $\Delta_i \triangleq \min_{j\in[m]}\frac{\price_j}{\ctrij}$.
Due to the ``highest bang-per-buck'' property of FPPE, 
we know that $\Delta_i = \frac{\price_j}{\ctrij}$
for every $j\in \purchaseseti$.
In the remaining argument, we further assume that ${\valuefuncderivative_i(\sum_{j\in[m]}\ctrij\allocij)} / {\moneycostderivative_i
    (\sum_{j\in[m]}\price_j\allocij)} > \Delta_i$.
Otherwise,  the lemma statement is satisfied since:
\begin{align*}
    \frac
    {\valuefuncderivative_i
    \left(\sum_{j\in[m]}\ctrij\allocij\right)}
{\moneycostderivative_i
    \left(\sum_{j\in[m]}\price_j\allocij\right)}
    \leq
    \Delta_i
    =
    \frac{\sum_{j\in \purchaseseti}\price_j}{\sum_{j\in \purchaseseti}\ctrij}
    \leq
    \frac{\sum_{j\in \purchaseseti}\price_j}{\sum_{j\in \purchaseseti}(1-\allocij)\ctrij}
\end{align*}

Using the same argument as the one presented in \Cref{lem:reporting infinite value},
the same inner FPPE and thus the same utility for every agent are induced
if we hold all other agents' reports fixed and let agent~$i$ report $\reportval_i \triangleq \frac{\valuefuncderivative_i(\sum_{j\in[m]}\ctrij\allocij)}{\moneycostderivative_i
    \left(\sum_{j\in[m]}\price_j\allocij\right)}$ and $\reportwealth_i \triangleq \paymenti$.} By the definition of $\purchaseseti$, we know that $\reportwealth_i =
\sum_{j\in \purchaseseti}\pricej\allocij$ where we shrink the summation over set $[m]$ to $\purchaseseti$.
The utility $\util_i$ of agent $i$ in the equilibrium can be computed as
\begin{align*}
\util_i
&\overset{(a)}{=}
    \valuefunc_i\left(
    \left(\max_{j\in[m]}
    \frac{\ctrij}{\price_j}
    \right)
    \cdot
    \paymenti
    \right)
    - \moneycost_i\left(
        \paymenti
    \right)
    \\
    &\overset{(b)}{=}
    \valuefunc_i\left(
    \frac{\paymenti}{\Delta_i}
    \right)
    - \moneycost_i\left(
        \paymenti
    \right)
    \\
    &\overset{(c)}{=}
    \valuefunc_i\left(
    \frac{
    \reportwealth_i}
    {\Delta_i}
    \right)
    - \moneycost_i\left(
        \reportwealth_i
    \right)
\end{align*}
where
equality~(a) holds due to \Cref{lem:compute utility};
equality~(b) holds due to the definition of $\Delta_i$;
and
equality~(c) holds due to the construction of $\reportwealth_i$.

Now consider a deviation of agent $i$ which keeps her reported value $\reportval_i\primed \triangleq \reportval_i$ as the same, while
increasing her reported budget by $\epsilon$,
i.e., $\reportwealth_i\primed \triangleq \reportwealth_i + \epsilon$
for sufficiently small and positive $\epsilon <
\min\{\truewealth_i - \reportwealth_i,
\min_{j\in \purchaseseti}
\reportval_i\ctrij - \price_j\}$.
Let $\price\primed, \alloc\primed, \payment\primed$
be the per-unit prices, allocation and payment of the new inner FPPE under such a deviation.
\Cref{lem:fppe increasing budget} implies that $
    \sum_{j\in \purchaseseti}\price_j\primed \overset{}{\leq}
    \epsilon + \sum_{j\in \purchaseseti}\price_j$.
Furthermore, agent $i$'s utility $\util_i\primed$ after her deviation can be lowerbounded as
\begin{align*}
\util_i\primed
&\overset{(a)}{=}
    \valuefunc_i\left(
    \left(\max_{j\in[m]}
    \frac{\ctrij}{\price_j\primed}
    \right)
    \cdot \paymenti\primed
    \right)
    - \moneycost_i\left(
       \paymenti\primed
    \right)
    \\
    &\geq
    \valuefunc_i\left(
    \frac{\sum_{j\in \purchaseseti}\ctrij}
    {\sum_{j\in \purchaseseti}
    \price_j\primed}
    \cdot
    \paymenti\primed
    \right)
    - \moneycost_i\left(
        \paymenti\primed
    \right)
    \\
    &\overset{(b)}{=}
    \valuefunc_i\left(
    \frac{\sum_{j\in \purchaseseti}\ctrij}
    {\sum_{j\in \purchaseseti}
    \price_j\primed}
    \cdot
    \reportwealth_i\primed
    \right)
    - \moneycost_i\left(
        \reportwealth_i\primed
    \right)
    \\
    &\geq
    \valuefunc_i\left(
    \frac{\sum_{j\in \purchaseseti}\ctrij}
    {\epsilon+ \sum_{j\in \purchaseseti}
    \price_j}
    \cdot
    \left(
    \reportwealth_i
    +
    \epsilon
    \right)
    \right)
    - \moneycost_i\left(
        \reportwealth_i
        +
        \epsilon
    \right)
    \\
    &\overset{(c)}{=}
    \valuefunc_i\left(
    \frac{\sum_{j\in\purchaseseti}\price_j}
    {\epsilon+\sum_{j\in\purchaseseti}\price_j}
    \frac{
    \reportwealth_i
    +
    \epsilon
    }{\Delta_i}
    \right)
    - \moneycost_i\left(
        \reportwealth_i
        +
        \epsilon
    \right)
\end{align*}
where
equality~(a) holds due to \Cref{lem:compute utility} and the construction of $\reportwealth_i\primed$ so that $\paymenti\primed \leq \reportwealth_i\primed \leq \truewealth_i$;
equality~(b) holds due to the construction of $\reportval_i\primed$, the choice of $\epsilon$,
and the ``budget feasibility'' property of FPPE;
and equality~(c) holds due to the definition of $\Delta_i$ and the ``highest bang-per-buck'' property of FPPE,
which implies $\frac{\price_j}{\ctrij} = \Delta_i$ for every $j\in \purchaseseti$.

Since $\util_i$ is agent $i$'s utility in the equilibrium,
utility $\util_i\primed$ under her deviation should not be profitable,
i.e., $\util_i \geq \util_i\primed$.
Using the bounds of $\util_i$ and $\util_i\primed$ obtained above, the difference between agent $i$'s utility before and after her deviation can be upperbounded as
\begin{align*}
    \util_i - \util_i\primed
    &\leq \valuefunc_i\left(
    \frac{\reportwealth_i}
    {\Delta_i}
    \right)
    - \moneycost_i\left(
        \reportwealth_i
    \right)
    -
    \left(
    \valuefunc_i\left(
    \frac{\sum_{j\in\purchaseseti}\price_j}
    {\epsilon+\sum_{j\in\purchaseseti}\price_j}
    \frac{
    \reportwealth_i
    +
    \epsilon
    }{\Delta_i}
    \right)
    - \moneycost_i\left(
        \reportwealth_i
        +
        \epsilon
    \right)
    \right)
    \\
    &=
    \moneycost_i(\reportwealth_i + \epsilon)
    -
    \moneycost_i(\reportwealth_i)
    -
     \left(
     \valuefunc_i\left(
    \frac{\sum_{j\in\purchaseseti}\price_j}
    {\epsilon+\sum_{j\in\purchaseseti}\price_j}
    \frac{
    \reportwealth_i
    +
    \epsilon
    }{\Delta_i}
    \right)
    -
    \valuefunc_i\left(
    \frac{
    \reportwealth_i
    }{\Delta_i}
    \right)
     \right)
    \\
    &\overset{(a)}{\leq}
    \epsilon
    \cdot
    \moneycostderivative_i(\reportwealth_i + \epsilon)
    -
    \left(
    \frac{\sum_{j\in\purchaseseti}\price_j}
    {\epsilon+\sum_{j\in\purchaseseti}\price_j}
    \frac{
    \reportwealth_i
    +
    \epsilon
    }{\Delta_i}
    -
    \frac{
    \reportwealth_i
    }{\Delta_i}
    \right)
    \cdot
     \valuefuncderivative_i
     \left(
     \frac{\sum_{j\in\purchaseseti}\price_j}
    {\epsilon+\sum_{j\in\purchaseseti}\price_j}
    \frac{
    \reportwealth_i
    +
    \epsilon
    }{\Delta_i}
    \right)
     \\
     &\overset{}{=}
     \epsilon
    \cdot
    \moneycostderivative_i(\reportwealth_i + \epsilon)
    -
    \frac{\epsilon}{\Delta_i}
    \cdot
    \frac{\sum_{j\in\purchaseseti}\pricej -
    \reportwealth_i}{\sum_{j\in\purchaseseti}\price_j + \epsilon}
    \cdot
     \valuefuncderivative_i
     \left(
     \frac{\sum_{j\in\purchaseseti}\price_j}
    {\epsilon+\sum_{j\in\purchaseseti}\price_j}
    \frac{
    \reportwealth_i
    +
    \epsilon
    }{\Delta_i}
    \right)
     \\
     &\overset{(b)}{=}
     \epsilon
    \cdot
    \moneycostderivative_i(\reportwealth_i + \epsilon)
    -
    \epsilon
    \cdot
    \frac{\sum_{j\in\purchaseseti}(1-\allocij)\ctrij
    }
    {\sum_{j\in\purchaseseti}\price_j + \epsilon}
    \cdot
     \valuefuncderivative_i
     \left(
     \frac{\sum_{j\in\purchaseseti}\price_j}
    {\epsilon+\sum_{j\in\purchaseseti}\price_j}
    \frac{
    \reportwealth_i
    +
    \epsilon
    }{\Delta_i}
    \right)
\end{align*}
where inequality~(a) holds due to the concavity of
and the convexity of money cost function $\moneycost_i$;
and equality~(b) holds since $\frac{\price_j}{\ctrij} = \Delta_i$
for every $j\in \purchaseseti$, and thus $\frac{1}{\Delta_i}(\sum_{j\in\purchaseseti}\pricej - \reportwealth_i) =
\frac{1}{\Delta_i}(\sum_{j\in\purchaseseti}\pricej - \paymenti)
=
\frac{1}{\Delta_i}(\sum_{j\in\purchaseseti}\pricej - \sum_{j\in\purchaseseti}\pricej\allocij) =
\sum_{j\in\purchaseseti}(1-\allocij)\ctrij$.
After rearranging the terms,
we have
\begin{align*}
       \frac
{\valuefuncderivative_i
     \left(
     \frac{\sum_{j\in\purchaseseti}\price_j}
    {\epsilon+\sum_{j\in\purchaseseti}\price_j}
    \frac{
    \reportwealth_i
    +
    \epsilon
    }{\Delta_i}
    \right)}
{\moneycostderivative_i(\reportwealth_i + \epsilon)}
\leq
\frac{\sum_{j\in \purchaseseti}\price_j + \epsilon}{\sum_{j\in \purchaseseti}(1-\allocij)\ctrij}
\end{align*}
for all positive $\epsilon$ that is sufficiently small.
Finally, letting $\epsilon$ approach zero,
the lemma statement is obtained as
\begin{align*}
&\lim_{\epsilon\rightarrow 0}
\moneycostderivative_i(\reportwealth_i + \epsilon)
\overset{(a)}{=}
\moneycostderivative_i(\reportwealth_i)
=
\moneycostderivative_i\left(\paymenti\right)
\\
    &\lim_{\epsilon\rightarrow 0}
    \valuefuncderivative_i
     \left(
     \frac{\sum_{j\in\purchaseseti}\price_j}
    {\epsilon+\sum_{j\in\purchaseseti}\price_j}
    \frac{
    \reportwealth_i
    +
    \epsilon
    }{\Delta_i}
    \right)
    \overset{(b)}{=}
    \valuefuncderivative_i
     \left(
    \frac{
    \reportwealth_i
    }{\Delta_i}
    \right)
%     =
%     \valuefuncderivative_i
% \left(\frac{\sum_{j\in[m]}\pricej\allocij}{\Delta_i}\right)
    \overset{(c)}{=}
\valuefuncderivative_i
\left(\sum_{j\in[m]}\ctrij\allocij\right)
\end{align*}
where
equalities~(a) (b) hold since both $\moneycost_i$ and $\valuefunc_i$ are differentiable;
and equality~(c) holds due to the construction of $\reportwealth_i$ and the fact that $\frac{\price_j}{\ctrij} = \Delta_i$
for every $j\in \purchaseseti$ and $\allocij = 0$ for every $j\notin\purchaseseti$.
%\end{proof}

\newcommand{\jPrime}{j\primed}

\paragraph{Case b. Suppose $\purchaseseti =\emptyset$.}
\yfedit{In this case, $\sum_{j\in[m]}\ctrij\allocij = 0$, $\paymenti = 0$ 
and agent $i$ has zero utility.
We consider the following contradiction argument:
suppose there exists an item $j\primed\in[m]$ such that $\frac{\price_{\jPrime}}{\ctr_{i\jPrime}} < \frac{\valuefuncderivative_i(0)}{\moneycostderivative_i(0)}$.

Consider a deviation of agent $i$ reports $\reportval_i\primed \triangleq \frac{\valuefuncderivative_i(0)}{\moneycostderivative_i
    (0)}$ and $\reportwealth_i\primed \triangleq \epsilon$ for sufficiently small and positive $\epsilon < \reportval_i\ctr_{i\jPrime} - \price_{\jPrime}$.
Let $\price\primed, \alloc\primed, \payment\primed$
be the per-unit prices, allocation and payment of the new inner FPPE under such a deviation.
\Cref{lem:fppe increasing budget} implies that $
    \price_{\jPrime}\primed \overset{}{\leq}
    \epsilon + \price_{\jPrime}$.
Furthermore, agent $i$'s utility $\util_i\primed$ after her deviation can be lowerbounded as
\begin{align*}
\util_i\primed
&\overset{(a)}{=}
    \valuefunc_i\left(
    \left(\max_{j\in[m]}
    \frac{\ctrij}{\price_j\primed}
    \right)
    \cdot \paymenti\primed
    \right)
    - \moneycost_i\left(
       \paymenti\primed
    \right)
    \\
    &\geq
    \valuefunc_i\left(
    \frac{\ctr_{i\jPrime}}
    {
    \price_{\jPrime}\primed}
    \cdot
    \paymenti\primed
    \right)
    - \moneycost_i\left(
        \paymenti\primed
    \right)
    \\
    &\overset{(b)}{=}
    \valuefunc_i\left(
    \frac{\ctr_{i\jPrime}}
    {
    \price_{\jPrime}\primed}
    \cdot
    \reportwealth_i\primed
    \right)
    - \moneycost_i\left(
        \reportwealth_i\primed
    \right)
    \\
    &\geq
    \valuefunc_i\left(
    \frac{\ctr_{i\jPrime}}
    {\epsilon +    \price_{\jPrime}}
    \cdot
    \epsilon
    \right)
    - \moneycost_i\left(
        \epsilon
    \right)
    \\
    &\overset{(c)}{>}0
\end{align*}
where
equality~(a) holds due to \Cref{lem:compute utility};
equality~(b) holds due to the construction of $\reportval_i\primed$, the choice of $\epsilon$,
and the ``budget feasibility'' property of FPPE;
and strict inequality~(c) holds for sufficiently small $\epsilon$ since $\frac{\price_{\jPrime}}{\ctr_{i\jPrime}} < \frac{\valuefuncderivative_i(0)}{\moneycostderivative_i(0)}$.
Finally, note that $\util_i\primed > 0$ leads to a contradiction, since agent~$i$ has zero utility in the equilibrium.
}
%We are now prepared to prove \Cref{thm:poa pure}.
%\begin{proof}[Proof of \Cref{thm:poa pure}]

\subsection{Proof of \texorpdfstring{\Cref{thm:poa pure}}{}}

    Fix an arbitrary pure Nash equilibrium and suppose $\price, \alloc, \payment$
    are the per-unit prices, allocation and payment of the inner FPPE, respectively.
    Let
    $\optalloc$ be the optimal allocation that maximizes the liquid welfare.
    Consider the following partition $A_1\bigsqcup A_2\bigsqcup A_3$
    of agents based on $\alloc$, $\price$, and $\optalloc$:
    \begin{align*}
        A_1 &\triangleq
        \left\{i\in[n]: \sum_{j\in[m]}\pricej\allocij = \truewealth_i\right\}
        \\
        A_2 &\triangleq
        \left\{i\in[n]: i \not\in A_1
        \land
        \sum_{j\in[m]} \ctrij\optalloc_{ij} \leq
        \sum_{j\in[m]} \ctrij\allocij
        \right\}
        \\
        A_3 &\triangleq
        \left\{i\in[n]: i \not\in A_1
        \land
        \sum_{j\in[m]} \ctrij\optalloc_{ij} >
        \sum_{j\in[m]} \ctrij\allocij
        \right\}
    \end{align*}
    In words, $A_1$ corresponds to every agent $i$ who exhausts her true budget $\truewealth_i$ in the equilibrium;
    and $A_2$, $A_3$ correspond to the remaining agents {divided according to whether their individual allocation is larger in the equilibrium outcome or in the optimal allocation}.

    In the following, we compare the willingness to pay (liquid welfare contribution) for agents from $A_1, A_2, A_3$ separately.

    \smallskip
    For every agent $i\in A_1$, note that
    % \begin{align*}
        $\wtp_i(\optalloc_i)
        \overset{}{\leq} \truewealth_i
        \overset{}{=}
        \wtp_i(\alloci)$
    % \end{align*}
    where the inequality holds due to the definition of $\wtp_i$,
    and the equality holds due to the definition of $A_1$.

    For every agent $i\in A_2$, note that
    % \begin{align*}
        $\wtp_i(\optalloc_i) \leq \wtp_i(\alloci)$
    % \end{align*}
    due to the fact that $\wtp_i$ is increasing and $\sum_{j\in[m]} \ctrij\optalloc_{ij} \leq
        \sum_{j\in[m]} \ctrij\allocij$ in the definition of $A_2$.

    For every agent~ $i\in A_3$,
    let $\purchaseseti \triangleq \{j\in[m]:\allocij > 0\}$
    be the subset of items for which agent~$i$ receives a strictly positive fraction.
    \yfedit{
    Consider two cases. 
    First, suppose $\purchaseseti = \emptyset$. Note that
    \begin{align*}
        \frac{\valuefunc_i\left(\sum_{j\in[m]} \ctrij\optalloc_{ij}
        \right)}{
        \moneycost_i\left(\sum_{j\in[m]} \pricej\optalloc_{ij}
        \right)
        }
        &\overset{(a)}{\leq}
        \frac{\left(\sum_{j\in[m]} \ctrij\optalloc_{ij}
        \right)\cdot \valuefuncderivative_i(0)}{
        \left(\sum_{j\in[m]} \pricej\optalloc_{ij}
        \right)
        \cdot \moneycostderivative_i(0)
        }
        \overset{(b)}{\leq}
        \frac{\left(\sum_{j\in[m]} \ctrij\optalloc_{ij}
        \right)}{
        \left(\sum_{j\in[m]} \pricej\optalloc_{ij}
        \right)
        }
        \cdot
        \left(
        \min_{j\in[m]}\frac{\pricej}{\ctrij}
        \right)
        \leq 1
    \end{align*}
    where inequality~(a) holds due to the concavity (convexity) of valuation function $\valuefunc_i$
    (money cost function $\moneycost_i$);
    and inequality~(b) holds due to \Cref{lem:poa pure allocation price relation}.
    The above inequality further implies
    \begin{align*}
        \wtp_i(\optalloc_i)
        &\overset{}{=}
        \moneycost_i^{-1}
        \left(\valuefunc_i\left(\sum_{j\in[m]} \ctrij\optalloc_{ij}
        \right)
        \right)
        \leq 
        \sum_{j\in[m]} \pricej\optalloc_{ij}
    \end{align*}
    }Next, suppose $\purchaseseti \not=\emptyset$. Let $\Delta_i \triangleq \frac{\sum_{j\in \purchaseseti}\price_j}{\sum_{j\in \purchaseseti}\ctrij}$.
Due to the ``highest bang-per-buck'' property of FPPE,
$\frac{\price_j}{\ctrij} \geq \Delta_i$ for every item $j\in[m]$,
and equality holds for $j\in \purchaseseti$.
    Note that
    \begin{align*}
        \wtp_i(\optalloc_i)
        &\overset{(a)}{\leq}
        \wtp_i(\alloci) +
        \frac{
        \valuefunc_i\left(\sum_{j\in[m]} \ctrij\optalloc_{ij}
        \right) -
        \valuefunc_i\left(
        \sum_{j\in[m]} \ctrij\allocij
        \right)
        }{
        \moneycostderivative_i(\wtp_i(\alloci))
        }
        \\
        &\overset{(b)}{\leq}
         \wtp_i(\alloci) +
        \frac{
        \valuefunc_i\left(\sum_{j\in[m]} \ctrij\optalloc_{ij}
        \right) -
        \valuefunc_i\left(
        \sum_{j\in[m]} \ctrij\allocij
        \right)
        }{
        \moneycostderivative_i\left(\paymenti\right)
        }
        \\
        &\overset{(c)}{\leq}
         \wtp_i(\alloci)
         +
         \left(
         \sum_{j\in[m]} \ctrij(\optalloc_{ij} -\allocij)
         \right)
         \cdot \frac{
        \valuefuncderivative_i\left(
        \sum_{j\in[m]} \ctrij\allocij
        \right)
         }{
        \moneycostderivative_i\left(
        \paymenti\right)
         }
         \\
        &\overset{(d)}{\leq}
        \wtp_i(\alloci) +
       \left(
         \sum_{j\in[m]} \ctrij(\optalloc_{ij} -\allocij)
         \right)
         \cdot
    \frac{\sum_{j\in \purchaseseti}\price_j}{\sum_{j\in \purchaseseti}(1-\allocij)\ctrij}
    \\
        &\overset{(e)}{\leq}
        \wtp_i(\alloci) +
       \left(
         \sum_{j\in[m]} \ctrij\optalloc_{ij}
         \right)
         \cdot
    \frac{\sum_{j\in \purchaseseti}\price_j}{\sum_{j\in \purchaseseti}\ctrij}
    \\
        &\overset{(f)}{=}
        \wtp_i(\alloci) +
       \left(
         \sum_{j\in[m]} \ctrij\optalloc_{ij}
         \right)
         \cdot \Delta_i
        \\
        &\overset{(g)}{\leq}
        \wtp_i(\alloci) +
        \sum_{j\in[m]} \pricej\optalloc_{ij}
    \end{align*}
    where inequality~(a) holds due to the definition of $\wtp_i$ and the convexity of money cost function $\moneycost_i$;
    inequality~(b) holds since
    $\moneycostderivative_i$ is weakly increasing implied by
    the convexity of
    $\moneycost_i$,
    and
    $\wtp_i(\alloci) \geq \paymenti$ which is implied by the definition of $\wtp_i$ and
    the fact that
    agent $i$'s utility is non-negative in the equilibrium;
    inequality~(c) holds due to the concavity of valuation function $\valuefunc_i$;
    inequality~(d) holds due to \Cref{lem:poa pure allocation price relation};
    and inequality~(e) holds
    since $\frac{\sum_{j\in[m]\ctrij(\optalloc_{ij}-\allocij)}}{\sum_{j\in[m]}\ctrij(1-\allocij)}
    \leq \frac{\sum_{j\in[m]\ctrij\optalloc_{ij}}}{\sum_{j\in[m]}\ctrij}$
    by algebra;
    equality~(f) holds due to the definition of $\Delta_i$;
    and inequality~(g) holds since
    $\frac{\price_j}{\ctrij} \geq \Delta_i$ for every item~$j\in[m]$.

    \smallskip
    Putting all pieces together, we have the following upper bound of
    the optimal liquid welfare $\wtp(\optalloc)$,
    \begin{align*}
    % \label{eq:welfare upper bound single item concave value hard budget}
    \begin{split}
        \wtp(\optalloc)
        &=
        \sum_{i\in[n]} \wtp_i(\optalloc_i)
        % \\
        \leq \sum_{i\in A_3}\sum_{j\in[m]} \pricej\optalloc_{ij}
        +
        \sum_{i\in A_1\sqcup A_2 \sqcup A_3}
        \wtp_i(\alloci)
        \\
        &
        \overset{(a)}{\leq}
        \sum_{j\in[m]}\pricej
        +
        \sum_{i\in A_1\sqcup A_2 \sqcup A_3}
        \wtp_i(\alloci)
        % \\
        \overset{(b)}{\leq}
        2\wtp(\allocs)
        \end{split}
    \end{align*}
    where inequality~(a) holds since $\sum_{i\in A_3}\optalloc_{ij} \leq 1$
    for every item $j\in[m]$;
    and inequality~(b) holds since
    $\sum_{j\in[m]}\pricej \leq \wtp(\allocs)$
    since all agents receive non-negative utility in the equilibrium.
%\end{proof}

%\subsection{PoA Bound Under Mixed Nash Equilibrium}
\section{Main Result for Mixed Nash Equilibrium}
\label{sec:poa mixed}
%In this subsection, we present the PoA bound under mixed Nash equilibrium.

{Now we analyze the PoA for the mixed equilibria of the metagame.}

\begin{theorem}
\label{thm:poa mixed}
    In the metagame,
    the price of anarchy $\mixedPoA$ under mixed Nash equilibrium lies in [2, 4].
\end{theorem}

Note that \Cref{example:pure lower bound} also serves as a lower bound for \Cref{thm:poa mixed}, so what remains is to prove the upper bound.  {We first provide some intuition into the high-level approach.
%Intuitively, we prove the upper bound in \Cref{thm:poa mixed} by explicitly constructing a deviation for each agent $i$ based on the optimal allocation $\optalloc$ that maximizes the transferable welfare.  This differs from our approach
In our analysis of pure Nash equilibrium, we analyzed the impact of small budget adjustments given the messages of the other agents.  However, in a mixed Nash equilibium, the messages (and hence outcomes and prices) may be random, so the impact of local adjustments less clear.  We instead consider a specific budget-setting strategy that each agent will consider as a deviation.  Namely, each agent considers the \emph{expected} per-unit prices of the inner FPPE if she were not present, and then sets a budget equal to the expected total payment of her part of the optimal allocation under those prices.}  We utilize \Cref{lem:compute utility} and \Cref{lem:fppe increasing budget} to characterize the per-unit prices of the inner FPPE and the agent's utility under such a deviation. Consequently, we obtain an upper bound on the optimal liquid welfare.
\begin{proof}[Proof of \Cref{thm:poa mixed}]
   Fix an arbitrary mixed Nash equilibrium and suppose $\randomprice, \randomalloc, \randompayment$
    are the randomized per-unit prices, allocation and payment of the inner FPPE over the randomness of all agents' reported message, respectively.
    Let
    $\optalloc$ be the optimal allocation that maximizes the liquid welfare.
    In the following, we compare the willingness to pay (liquid welfare contribution) for agents separately.

    Fix agent $i$. We prove an upper bound
    on its contribution $\wtp_i(\optalloc_i)$ in the optimal liquid welfare,
    \begin{align}
    \label{eq:multi item mixed wtp bound}
        \tfrac{1}{2}\,\wtp_i(\optalloc_i) \leq
        \moneycost_i^{-1}(\util_i) +
        {\textstyle \sum_{j\in[m]}}\;
        \expect[\pricej\sim\randompricej]{\pricej}\optallocij,
    \end{align}
    where $\util_i$ is the expected utility of agent $i$ in the equilibrium.
    It is important to note that inequality~\eqref{eq:multi item mixed wtp bound}
    immediately implies the upper bound of PoA as desired since
    \begin{align*}
        \wtp(\optalloc)
        = \sum_{i\in[n]}
        \wtp_i(\optalloc_i)
        \leq
        2
        \left(
        \sum_{i\in[n]}
        \moneycost_i^{-1}(\util_i)
        +
        \sum_{j\in[m]}\expect[\pricej\sim\randompricej]{\pricej}\optallocij
        \right)
        \leq
        4\wtp(\allocs)
    \end{align*}
    where the last inequality holds since $\moneycost_i^{-1}(\util_i)
    \leq
    \moneycost_i^{-1}
    \left(
    \expect[\alloci\sim\randomalloc_i]{\valuefunc_i
    \left(
    \sum_{j\in[m]}\allocij\ctrij
    \right)
    }
    \right)
    =
    \wtp_i(\alloci)$,\footnote{This proof does not utilize the differentiability of money cost function $\moneycost_i$. Therefore, to simplify the presentation, we redefine $\moneycost_i^{\clubsuit}(\paymenti) \triangleq \moneycost_i(\paymenti)$ if $\paymenti \leq \truewealth_i$ and $\moneycost_i^{\clubsuit}(\paymenti) = \infty$ if $\paymenti > \truewealth_i$; and assume $\truewealth_i^{\clubsuit} = \infty$ for each agent $i$ without loss of generality.} and
    $\sum_{i\in[n]}\sum_{j\in[m]}\expect[\pricej\sim\randompricej]{\pricej}\optallocij
    \leq \sum_{j\in[m]}\expect[\pricej\sim\randompricej]{\pricej}
    \leq  \wtp(\allocs)$ due to the non-negative utility of every agent in the equilibrium.

    We prove inequality~\eqref{eq:multi item mixed wtp bound}
    by constructing a specific deviation based on $\optalloc_i$ for agent $i$.
    Let $\prices\doubleprimed$ be the randomized per-unit prices of the FPPE when agent $i$ reporting $\reportval_i\doubleprimed = \infty$ and $\reportwealth_i\doubleprimed = 0$, while all other agents reports the same messages as the ones in the equilibrium.
    It is important to note that
    \Cref{lem:fppe increasing budget} implies that randomized price $\randompricej\doubleprimed$ is first-order stochastically dominated by randomized price $\randompricej$ for every item $j\in[m]$, and thus
    $\expect[\pricej\doubleprimed\sim\randompricej\doubleprimed]{\pricej\doubleprimed}\leq \expect[\pricej\sim\randompricej]{\pricej}$.

    Let $\tildePj\doubleprimed \triangleq \expect[\pricej\doubleprimed\sim\randompricej\doubleprimed]{\pricej\doubleprimed}$
    for every item $j\in[m]$.
    Consider the following deviation for agent $i$
    where agent $i$ deterministically reports
      $\reportval_i\primed = \infty,
        \reportwealth_i\primed =
        \sum_{j\in[m]}
        \tildePj\doubleprimed
        \optalloc_{ij} $.
    Let $\randomalloc\primed,\randomprice\primed$ be the randomized allocation and per-unit prices of the inner FPPE after this deviation of agent $i$.
    Since agent $i$ reports $\val_i\primed = \infty$,
    she always exhausts her reported budget$~\reportwealth_i\primed$.

    In the remaining analysis,
    we couple the reported message of all other agents except agent~$i$ in the FPPE inducing randomized per-unit prices $\randomprice\doubleprimed$, and the FPPE inducing randomized
    allocation~$\randomalloc\primed$ and per-unit prices $\randomprice\primed$.
    It is important to note that
    \Cref{lem:fppe increasing budget} implies
    \begin{align*}
        \pricej\doubleprimed \leq \pricej\primed
        \text{ for every $j\in[m]$, and }
        \sum_{j\in[m]}\pricej\primed \leq
        \sum_{j\in[m]}\pricej\doubleprimed + \reportwealth_i\primed
    \end{align*}
    for every realization of agents' message profile.

    Define auxiliary random variable
    $\boldsymbol\tau_i \triangleq
    \frac{\reportwealth_i\primed}
    {\sum_{j\in[m]}
    \randompricej\doubleprimed\optallocij} \geq 0$.
    By definitions of $\boldsymbol\tau_i$ and $\reportwealth_i\primed$, $\expect{\frac{1}{\boldsymbol\tau_i}} = 1$.
    The expected utility $\util_i$
    of agent $i$ in the equilibrium can be lowerbounded as her expected utility under this deviation. Namely,
    \begin{align*}
        \util_i
        &\overset{(a)}{\geq}
        \expect
        [\price\primed\sim\randomprice\primed]
        {
        \valuefunc_i\left(
        \left(
        \max_{j\in[m]}
        \frac{\ctrij}{\pricej\primed}
        \right)
        \reportwealth_i\primed
        \right)
        -
        \moneycost_i\left(\reportwealth_i\primed\right)}
        \\
        &\overset{(b)}{\geq}
        \expect
        [\price\primed\sim\randomprice\primed,\tau_i\sim\boldsymbol\tau_i]
        {
        \valuefunc_i\left(
        \left(
        \max_{j\in[m]}
        \frac{\ctrij}{\pricej\primed}
        \right)
        \left(
        \sum_{j\in[m]}
        \frac{1}{\tau_i + 1}
        \tau_i\pricej\doubleprimed\optallocij
        +
        \frac{\tau_i}{\tau_i+1}
        \left(\pricej\primed - \pricej\doubleprimed\right)
        \right)
        \right)}
        -
        \moneycost_i\left(
        \sum_{j\in[m]}
        \tildePj\doubleprimed
        \optalloc_{ij}
        \right)
        \\
        &\overset{(c)}{\geq}
        \expect
        [\price\primed\sim\randomprice\primed,\tau_i\sim\boldsymbol\tau_i]
        {
        \valuefunc_i\left(
        \left(
        \max_{j\in[m]}
        \frac{\ctrij}{\pricej\primed}
        \right)
        \left(
        \sum_{j\in[m]}
        \frac{\tau_i}{\tau_i + 1}
        \pricej\primed
        \optallocij
        \right)
        \right)}
        -
        \moneycost_i\left(
        \sum_{j\in[m]}
        \tildePj\doubleprimed
        \optalloc_{ij}
        \right)
        \\
        &\overset{(d)}{\geq}
        \expect
        [\price\primed\sim\randomprice\primed,\tau_i\sim\boldsymbol\tau_i]
        {
        \valuefunc_i\left(
        \sum_{j\in[m]}
        \frac{\tau_i}{\tau_i + 1}
        \frac{\ctrij}{\pricej\primed}
        \pricej\primed
        \optallocij
        \right)}
        -
        \moneycost_i\left(
        \sum_{j\in[m]}
        \tildePj\doubleprimed
        \optalloc_{ij}  \right)
        \\
        &\overset{(e)}{\geq}
        \expect
        [\tau_i\sim\boldsymbol\tau_i]
        {
        \frac{\tau_i}{\tau_i + 1}
        \valuefunc_i\left(
        \sum_{j\in[m]}
        {\ctrij}
        \optallocij
        \right)}
        -
        \moneycost_i\left(
        \sum_{j\in[m]}
        \tildePj\doubleprimed
        \optalloc_{ij}\right)
        \\
        &\overset{(f)}{\geq}
        {
        \frac{1}
        {\expect{\frac{1}{\boldsymbol\tau_i}}
        +1 }
        \valuefunc_i\left(
        \sum_{j\in[m]}
        {\ctrij}
        \optallocij
        \right)}
        -
        \moneycost_i\left(
        \sum_{j\in[m]}
        \tildePj\doubleprimed
        \optalloc_{ij}  \right)
        \\
        &\overset{(g)}{=}
        \frac{1}{2}
        \valuefunc_i\left(
        \sum_{j\in[m]}
        {\ctrij}
        \optallocij
        \right)
        -
        \moneycost_i\left(
        \sum_{j\in[m]}
        \tildePj\doubleprimed
        \optalloc_{ij}  \right)
    \end{align*}
    where inequality~(a) holds due to \Cref{lem:compute utility};
    inequality~(b) holds since $\reportwealth_i\primed = \tau_i\sum_{j\in[m]}\pricej\doubleprimed\optallocij$
    and $\wealth_i\primed \geq \sum_{j\in[m]}\pricej\primed - \pricej\doubleprimed$ for every realization of
    $\pricej\primed,\pricej\doubleprimed,\tau_i$;
     inequality~(c) holds since $\optallocij \leq 1$;
     inequality~(d) holds by algebra;
     {inequality~(e) holds due to the concavity of
     valuation function $\valuefunc_i$ and $\valuefunc_i(0) = 0$};
     inequality~(f) holds due to Jensen's inequality;
     and equality~(g) holds since $\expect{\frac{1}{\boldsymbol\tau_i}} = 1$
     by definition.

    We are ready to prove inequality~\eqref{eq:multi item mixed wtp bound} as follows,
    \begin{align*}
        \frac{1}{2}\wtp_i(\optalloc_i)
        &\overset{(a)}{=}
        \frac{1}{2}\moneycost_i^{-1}\left(
        \valuefunc_i\left(
        \sum_{j\in[m]}
        {\ctrij}
        \optallocij
        \right)
        \right)
        \\
        &\overset{(b)}{\leq}
        \moneycost_i^{-1}\left(
        \frac{1}{2}
        \valuefunc_i\left(
        \sum_{j\in[m]}
        {\ctrij}
        \optallocij
        \right)
        \right)
        \\
        &\overset{(c)}{\leq}
        \moneycost_i^{-1}\left(
        \util_i +
        \moneycost_i\left(\sum_{j\in[m]}
        \tildePj\doubleprimed
        \optalloc_{ij} \right)
        \right)
        \\
        &\overset{(d)}{\leq}
        \moneycost_i^{-1}\left(
        \util_i  \right)
        +
        \sum_{j\in[m]}
        \tildePj\doubleprimed
        \optalloc_{ij}
        \\
        &\overset{(e)}{\leq}
        \moneycost_i^{-1}\left(
        \util_i  \right)
        +
        \sum_{j\in[m]}
        \expect[\pricej\sim\randompricej]{\pricej}
        \optalloc_{ij}
    \end{align*}
    where
    equality~(a) holds due to the definition of $\wtp_i$;
    inequalities~(b) and (d) hold due to the concavity of $\moneycost_i^{-1}$;
    inequality~(c) holds due to the monotonicity of $\moneycost_i^{-1}$ and the lower bound of the expected utility $\util_i$ obtained above;
    and inequality~(e) holds since $\randomprice_j$
    first order-stochastically dominates $\randomprice_j\doubleprimed$ and thus $\tildePj\doubleprimed=\expect[\pricej\doubleprimed\sim\randompricej\doubleprimed]{\pricej\doubleprimed}\leq \expect[\pricej\sim\randompricej]{\pricej}$
    for every item $j$.
\end{proof}  

\section{Extension I: Metagame in Bayesian Environments}
\label{sec:bayesian}

In this section, we explore the generalization of our model and results to from full information environments to Bayesian environments. In this \emph{Bayesian metagame} extension, rather than assuming that each agent $i$ has a fixed type $(\valuefunctioni, \truewealthi, \moneycosti)$, we assume that each agent $i$'s type is independently drawn from a type distribution $\typedist_i$.
In the following, we formally define the solution concept -- Bayesian Nash equilibrium and price of anarchy under Bayesian Nash equilibrium.
Finally, we prove that the price of anarchy under Bayesian Nash equilibrium is between $[2, 4]$ in \Cref{thm:poa bayes}.

In Bayesian metagame, a strategy $\strategy_i$ of agent $i$ is a stochastic
mapping from agent $i$'s type
$(\valuefunctioni, \truewealthi, \moneycosti)$
to a randomized message $(\randomreportval_i, \randomreportwealth_i)$.
The \emph{Bayesian Nash equilibrium} is defined as follows.

\begin{definition}[Bayesian Nash equilibrium]
For agents with type distributions $\{\typedist_i\}_{i\in[n]}$,
a Bayesian Nash equilibrium is a strategy profile $\{\strategy_i\}_{i\in[n]}$
such that for every agent $i$, every realized type
$(\valuefunctioni, \truewealthi, \moneycosti)$,
and every message
$(\reportval_i\primed, \reportwealth_i\primed)$,
\begin{align*}
    &\expect[ (\reportval_{i},\reportwealth_{i})
        \sim
        \strategy_i(\valuefunctioni, \truewealthi, \moneycosti)]{
    \expect[(\valuefunction_{-i}, \truewealth_{-i}, \moneycost_{-i})
    \sim F_{-i}]
    {
    \expect[ (\reportval_{-i},\reportwealth_{-i})
        \sim
        \strategy_{-i}(\valuefunction_{-i}, \truewealth_{-i}, \moneycost_{-i})]
        {\util_i(\reportval_i, \reportwealth_i, \reportval_{-i}, \reportwealth_{-i})}
    }
    }
    \\
    &\qquad\qquad\qquad\qquad\geq
    \expect[(\valuefunction_{-i}, \truewealth_{-i}, \moneycost_{-i})
    \sim F_{-i}]
    {
    \expect[ (\reportval_{-i},\reportwealth_{-i})
        \sim
        \strategy_{-i}(\valuefunction_{-i}, \truewealth_{-i}, \moneycost_{-i})]
        {\util_i(\reportval_i\primed, \reportwealth_i\primed, \reportval_{-i}, \reportwealth_{-i})}
    }
\end{align*}
\end{definition}

Clearly, the Bayesian metagame is a generalization of our baseline model, as every problem instance in the baseline model where agents have fixed types can be viewed as a problem instance where agents' types are drawn from single point-mass distributions in the Bayesian metagame. As a result of this equivalence, Bayesian Nash equilibrium also encompasses mixed Nash equilibrium as a generalization.

Similar to the baseline model, we evaluate the performance of the Bayesian metagame by measuring the approximation of liquid welfare. This is done by comparing the worst expected liquid welfare among all possible equilibria with the optimal expected liquid welfare over the randomness of agents' type and messages, and taking the supremum over all instances.
\begin{definition}[Price of anarchy in Bayesian environments]
    The \emph{price of anarchy (PoA)} of the metagame $\bayesPoA$
    under Bayesian Nash equilibrium is
    \begin{align*}
        \bayesPoA \triangleq
        \sup_{n, m, \ctr}
        \sup_{\{\typedist_i\}_{i\in[n]}}
        \frac
        {
        \expect[\{(\valuefunctioni,\truewealthi,\moneycosti)\}_{i\in[n]}\sim\{\typedist_i\}_{i\in[n]}]{
        \max_{\alloc}\twelfare(\alloc\condition \{(\valuefunctioni,\truewealthi,\moneycosti)\}_{i\in[n]})
        }
        }
        {\inf_{\strategy \in \texttt{Bayes}}\hat\twelfare(\strategy\condition \{\typedist_i\}_{i\in[n]})}
    \end{align*}
    where $\texttt{Bayes}$
    is the set of strategy profiles
    in all Bayesian Nash equilibrium
    given type distributions $\{\typedist\}_{i\in[n]}$,
    $\twelfare(\alloc\condition \{(\valuefunctioni,\truewealthi,\moneycosti)\}_{i\in[n]})$
    is the liquid welfare of allocation $\alloc$ given agents' types $\{(\valuefunctioni,\truewealthi,\moneycosti)\}_{i\in[n]}$,
    and
    $\hat\twelfare(\strategy\condition \{\typedist_i\}_{i\in[n]})$
    is the expected liquid welfare under strategy profile $\strategy$ defined as
    \begin{align*}
        \hat\twelfare(\strategy\condition \{\typedist_i\}_{i\in[n]})
        &\triangleq
        \sum_{i\in[n]} \hat\wtp_i(\strategy\condition \{\typedist_i\}_{i\in[n]}), \\
    \hat\wtp_i(\strategy\condition \{\typedist_\ell\}_{\ell\in[n]}) 
        &\triangleq
        \expect[(\valuefunctioni,\truewealthi,\moneycosti)\sim\typedist_i]{
        \wtp_i(\randomalloc_i(\strategy,(\valuefunctioni,\truewealthi,\moneycosti), \typedist_{-i}))
        }
    \end{align*}
    for every agent $i\in[n]$. Here
    $\randomalloc_i(\strategy,(\valuefunctioni,\truewealthi,\moneycosti), \typedist_{-i})$
    is the randomized allocation of agent $i$ with realized type $(\valuefunctioni,\truewealthi,\moneycosti)$
    when agents report under strategy $\strategy$,
    and the randomness is over agent $i$'s message, other agents' types and their messages.
\end{definition}

The main result of this section is the bound on the PoA under Bayesian Nash equilibrium. The proof follows a similar approach to that of \Cref{thm:poa mixed} and is deferred to \Cref{apx:poabayes} for completeness.

\begin{restatable}{theorem}{poabayes}
\label{thm:poa bayes}
In the Bayesian metagame, the PoA under Bayesian Nash equilibrium lies in $[2, 4]$.
\end{restatable} 

\section{Extension II: Metagame under Restrictive Message Space}
\label{sec:value/budget reporting only}

In this section, our goal is to address the following question: \emph{In the metagame, what information encoded in agents' messages is essential and cannot be disregarded?} Answering this question holds significant practical implications. For instance, in digital advertising markets, the platform (seller) possesses control over the design of the auction interface for advertisers (agents). As a reminder, our metagame draws inspiration from such an interface where advertisers declare both budget and maximum bid constraints to their autobidder.

In the proof of \Cref{lem:reporting infinite value}, we establish that every message $(\reportval_i,\reportwealth_i)$ can be dominated by a message $(\reportval_i\primed, \reportwealth_i\primed)$ with reported value $\reportval_i\primed = \infty$. Similarly, in the proofs of other results in previous sections, we often construct deviation strategy with reported value $\reportval_i\primed = \infty$ and carefully design reported budget $\reportwealth_i\primed$. Loosely speaking, this indicates that for agents, \emph{the strategic decision of their reported budget $\reportwealth$ holds greater importance than their reported value (aka maximum bid) $\reportval$.}

Motivated by this intuition, we proceed to investigate two variants of the metagame, wherein agents report either only budgets or only values. In Section \ref{sec:budget reporting metagame}, we extend our analysis of the price of anarchy to the variant metagame where agents only report budgets. On the other hand, in Section \ref{sec:value reporting metagame}, we present a negative result demonstrating that the price of anarchy under Bayesian Nash equilibrium can be as high as linear in the number of agents in the variant metagame where agents only report values.

\subsection{Metagame with Budget Reporting Only}
\label{sec:budget reporting metagame}

In this subsection, we examine a variant of the metagame where agents exclusively report their budgets. The definitions of pure Nash equilibrium, mixed Nash equilibrium, and Bayesian Nash equilibrium are adjusted accordingly to accommodate this variant.

\begin{definition}[Metagame with budget reporting only]
Each agent $i$ decides on a message $\reportwealth_i \in \realsinf$ reported to the seller. Given reported message profile $\{\reportwealth_i\}_{i\in[n]}$, the seller implements allocation $\alloc(\{\reportwealth_i\}_{i\in[n]})$ and payment $\payment(\{\reportwealth_i\}_{i\in[n]})$ induced by the FPPE, assuming that agents have budgeted utility with types $\{(\infty, \reportwealth_i)\}_{i\in[n]}$.
\end{definition}

In words, in this variant of the metagame with budget reporting only, the seller treats each agent~$i$ with general utility model as a budgeted agent with value per click $\val_i = \infty$ and budget $\truewealth_i = \reportwealth_i$ reported from the agent. In this variant, we obtain the same price of anarchy (PoA) bounds as the original metagame where agents report both values and budgets.

\begin{proposition}
\label{prop:poa budget reporting}
In the metagame with budget reporting only, the price of anarchy $\purePoA$ under pure Nash equilibrium is 2, and
the price of anarchy $\mixedPoA$ ($\bayesPoA$) under mixed (Bayesian) Nash equilibrium is between [2, 4].
\end{proposition}
\begin{proof}
The PoA guarantees under mixed Nash equilibrium and Bayesian Nash equilibrium follow exactly the same argument as the ones in \Cref{thm:poa mixed,thm:poa bayes} where deviations $(\reportval_i\primed, \reportwealth_i\primed)$ are constructed such that $\reportval_i\primed = \infty$ and thus remain feasible in this variant of the metagame.

For the PoA guarantee under pure Nash equilibrium, note that the constructed deviation $(\reportval_i\primed, \reportwealth_i\primed)$ in \Cref{thm:poa pure} can be modified as $(\reportval_i\doubleprimed, \reportwealth_i\doubleprimed)$ where $\reportval_i\doubleprimed = \infty$ and $\reportwealth_i\doubleprimed = \reportwealth_i\primed$. It is straightforward to verify that the same argument applies and the upper bound of PoA holds. Regarding the lower bound, note that in \Cref{example:pure lower bound}, a pure Nash equilibrium is achieved when agent 1 reports $\reportwealth_1 = 1$ and agent 2 reports $\reportwealth_2 = 0$. In this equilibrium, the per-unit price of the inner FPPE is $\price_1 = 1$ and agent 1 receives the entire item. Consequently, the same lower bound of PoA under pure Nash equilibrium is obtained.
\end{proof}

\xhdr{Another variant for agents with linear valuation functions.}
In the remaining of this subsection, we restrict out attention to agents with linear valuation functions, i.e., $\valuefunctioni(\sum_{j}\ctrij\allocij) = \val_i\cdot \sum_{j}\ctrij\allocij$ where $\val_i$ is her \emph{value per click}. However, agents may still have hard budget $\truewealth_i$ and differentiable, weakly increasing, weakly convex money cost function $\moneycost_i$. For such agents, we denote $(\val_i, \truewealth_i, \moneycost_i)$ by their types. We consider another variant of the metagame with budget reporting only as follows.

\begin{definition}[Metagame with budget reporting only and known linear valuations]
Each agent $i$ with type $(\val_i, \truewealth_i, \moneycost_i)$ decides on a message $\reportwealth_i \in \realsinf$ reported to the seller. Given reported message profile $\{\reportwealth_i\}_{i\in[n]}$, the seller implements allocation $\alloc(\{\reportwealth_i\}_{i\in[n]})$ and payment $\payment(\{\reportwealth_i\}_{i\in[n]})$ induced by the FPPE, assuming that agents have budgeted utility with types $\{(\val_i, \reportwealth_i)\}_{i\in[n]}$.
\end{definition}

In words, in this variant of the metagame with budget reporting only, the seller knows the value per click $\val_i$ of each agent $i$ and treats this agent with type $(\val_i, \truewealth_i, \moneycost_i)$ as a budgeted agent with value per click $\val_i$ and budget $\truewealth_i = \reportwealth_i$ reported from the agent. In this variant, we still obtain the same price of anarchy (PoA) bounds as the original metagame where agents report both values and budgets.

\begin{proposition}
\label{prop:poa budget reporting linear value}
In the metagame with budget reporting only and known linear valuations, the price of anarchy $\purePoA$ under pure Nash equilibrium is 2, and
the price of anarchy $\mixedPoA$ ($\bayesPoA$) under mixed (Bayesian) Nash equilibrium is between [2, 4].
\end{proposition}
\begin{proof}
The PoA guarantee under pure Nash equilibrium follows exactly the same argument as the ones in \Cref{thm:poa pure} where deviations $(\reportval_i\primed, \reportwealth_i\primed)$ are constructed such that $\reportval_i\primed = \val_i$ when the valuation function is linear, and thus remain feasible in this variant of the metagame.

For the PoA guarantees under mixed Nash equilibrium and Bayesian Nash equilibrium, consider the constructed deviation $(\reportval_i\primed, \reportwealth_i\primed)$ where $\reportval_i\primed = \infty$ defined in \Cref{thm:poa mixed,thm:poa bayes}. Note that when the realized per-unit prices $\price$ of the inner FPPE satisfy $\max_{j\in[m]}\frac{\val_i\ctrij}{\pricej} \geq 1$, it can be verified that replacing reported value $\reportval_i\primed = \infty$ in the constructed deviation with reported value $\reportval_i\doubleprimed = \val_i$ does not alter the inner FPPE. Consequently, the same realized utility for agent $i$ is guaranteed. Conversely, when the realized per-unit prices $\price$ of the inner FPPE satisfy $\max_{j\in[m]}\frac{\val_i\ctrij}{\pricej} < 1$, the agent receives a non-positive realized utility under deviation $(\reportval_i\primed, \reportwealth_i\primed)$, while a non-negative realized utility is guaranteed under deviation $(\reportval_i\doubleprimed, \reportwealth_i\doubleprimed)$ with $\reportval_i\doubleprimed = \val_i$ and $\reportwealth_i\doubleprimed = \reportwealth_i\primed$. Therefore, we conclude that this new deviation $(\reportval_i\doubleprimed, \reportwealth_i\doubleprimed)$ can replace $(\reportval_i\primed, \reportwealth_i\primed)$ in the original analysis in \Cref{thm:poa mixed,thm:poa bayes}, and remains feasible in this variant of the metagame.
\end{proof}

\subsection{Metagame with Value Reporting Only}
\label{sec:value reporting metagame}

Let us examine a variant of the metagame where agents only report their values (i.e., the maximum bid).

\begin{definition}[Metagame with value reporting only]
Each agent $i$ decides on a message $\reportval_i \in \realsinf$ reported to the seller. Given reported message profile $\{\reportval_i\}_{i\in[n]}$, the seller implements allocation $\alloc(\{\reportval_i\}_{i\in[n]})$ and payment $\payment(\{\reportval_i\}_{i\in[n]})$ induced by the FPPE, assuming that agents have budgeted utility with types $\{(\reportval_i, \infty)\}_{i\in[n]}$.
\end{definition}

In words, in this variant of the metagame with value reporting only, the seller treats each agent~$i$ with general utility model as a budgeted agent with value per click $\val_i = \reportval_i$ reported from the agent and budget $\truewealth_i = \infty$ reported from the agent. In this variant, we present the following negative result on the price of anarchy under Bayesian Nash equilibrium.

\begin{proposition}
\label{prop:poa value reporting}
In the metagame with value reporting only, the price of anarchy $\bayesPoA$ under Bayesian Nash equilibrium is at least $\Omega(n)$, even for budgeted agents.
\end{proposition}

The proof of \Cref{prop:poa value reporting} relies on the following Bayesian instance with budgeted agents.

\begin{example}
\label{example:poa value reporting}
Consider a scenario with $n \triangleq N + 2$ budgeted agents and one item. Let us assume that the click-through rates are the same for all agents, i.e., $\ctr_{i1} = 1$ for every $i\in[n]$. The first $N$ agents, $i \in [N]$, have deterministic values $\val_i = N^2$ and budgets $\truewealth_i = \frac{1}{2}$. On the other hand, the last two agents, $i = N+1$ and $i = N+2$, have a value of $\val_i = 2$ with probability $1-\epsilon$, a value of $\val_i = 0$ with probability $\epsilon$, and deterministic budgets $\truewealth_i = 1$.
\end{example}

\begin{proof}[Proof of \Cref{prop:poa value reporting}]
Consider the Bayesian instance defined in \Cref{example:poa value reporting}.
The optimal expected liquid welfare is $\frac{N}{2}$. This is achieved through an allocation where each of the first $N$ budgeted agents receives a $\frac{1}{N}$-fraction of the item.

Consider the following Bayesian Nash equilibrium: The first $N$ agents, $i\in[N]$, report $\reportval_i = \frac{1}{2}$. The last two agents, $i = N + 1$ and $i = N+2 $ report $\reportval_i = 1$ when their true value is $\val_i = 2$, and report $\reportval_i = 0$ when their true value is $\val_i = 0$. In this equilibrium, the per-unit price is $\price_1 = 1$ when either of the last two agents $i$ has true value $\val_i = 2$ and reports $\reportval_i = 1$; and it is $\price_1 = \frac{1}{2}$ otherwise. In the former case, the last two agents shares the item, while in the later case, the first $N$ agents shares the item. Regardless of the tie-breaking rule of the inner FPPE, the expected liquid welfare is at most
\begin{align*}
    \underbrace{(1 - \epsilon)^2 \cdot 2}_{\text{the last two agents both have $\val_i = 2$}}
    +
    \underbrace{2\epsilon(1-\epsilon) \cdot 1}_{\text{one of the last two agents has $\val_i = 2$}}
    +
    \underbrace{\epsilon^2 \cdot N/2}_{\text{the last two agents both have $\val_i = 0$}}
\end{align*}
Letting $\epsilon$ approach zero and $N$ approach infinity, the lower bound $\Omega(n)$ of PoA under Bayesian Nash equilibrium is obtained as desired.

Finally, we verify the validity of this equilibrium by demonstrating that all agents have no profitable deviations. For each of the first $N$ agents, $i\in [N]$, deviating to a smaller reported value $\reportval_i\primed < \reportval_i$ reduces their utility to zero, while deviating to a higher reported value $\reportval_i\primed > \reportval_i$ results in a violation of their hard budget constraint when the last two agents report 0, thereby reducing their utility as well. As for the last two agents, $i = N + 1$ and $i = N + 2$, if their realized value is $\val_i = 0$, reporting $\val_i\primed = 0$ is optimal. Conversely, if their realized value is $\val_i = 2$, deviating to a smaller reported value $\reportval_i\primed < \reportval_i$ reduces their winning probability to $\epsilon$, while deviating to a higher reported value $\reportval_i\primed > \reportval_i$ leads to a violation of their hard budget constraint when the other of the last two agents reports 0, resulting in a reduction of their utility in both cases. \footnote{This proof assumes that the tie-breaking rule of the inner FPPE satisfies that when both of the last two agents report $\reportval_i = 1$, the fraction of the item allocated to each of them is at least $2\epsilon$.}
\end{proof}

\yfedit{
\section{Pure Nash Equilibrium for Single-Item Instances}
\label{sec:single item pure nash}

\newcommand{\reserveagent}{i^*}

To shed light on agents' strategic behavior in the metagame,
this section focuses on single-item instances.
All results and analysis hold for more general instances with homogeneous items:
each agent~$i$ has the same click-through rate $\ctrij$ for all item $j$,
i.e., $\ctrij = \ctr_{ij'}$ for every $j, j'\in[m]$, $i\in[n]$.

By restricting to single-item instances, the remaining of this section drops subscript index $j$ for the item. Moreover, without loss of generality, we assume $\ctri = 1$ for all agents and drop it as well.

Before presenting the main result of this section, we introduce the following two 
auxiliary notations that are used in the equilibrium characterization (\Cref{thm:pure nash single item}): fix a per-unit price $\price\in[0, \infty)$, define 
\begin{align*}
    \allocLBi(\price) &\triangleq 
    \frac{\truewealth_i}{\price} \wedge
    \min\left\{\alloci\in[0, 1]: \valuefunctionderivativei(\alloci)\cdot (1-\alloci) \leq \price \cdot \moneycostderivativei(\price\alloci)\right\} \\
    \allocUBi(\price) &\triangleq 
    \frac{\truewealth_i}{\price} \wedge
    \max\left\{
    \alloci\in[0, 1]: \valuefunctionderivativei(\alloci) \geq \price \cdot \moneycostderivativei(\price\alloci)
    \right\}
\end{align*}
Loosely speaking, $\allocLBi(\price)$ is the smallest allocation 
such that the agent $i$ has no incentive to weakly increase her reported budget in the metagame with induced per-unit price $\price$.
Specifically, the inequality in the definition $\allocLBi(\price)$
is exactly the inequality of \Cref{lem:poa pure allocation price relation} for single-item instances.
On the other side, $\allocUBi(\price)$ is the largest allocation such that
the agent $i$ has a weakly positive marginal utility when facing a fixed per-unit price $\price$.
If $\valuefunctionderivativei(\alloci) < \price\cdot \moneycostderivativei(\price\alloci)$
for all $\alloci\in[0, 1]$, we set $\allocUBi(\price) = 0$.
Since valuation function $\valuefunctioni$ (money cost function $\moneycosti$) is differentiable, weakly increasing, weakly concave (convex) and $\valuefunctioni(0) = 0$ ($\moneycosti(0)=0$), we make the following observation about $\allocLBi(\price)$ and $\allocUBi(\price)$.

\begin{observation}
\label{lem:allocLB allocUB}
For every agent $i\in[n]$, functions $\allocLBi(\price)$ and $\allocUBi(\price)$ satisfy the following properties:
\begin{enumerate}
    \item Both $\allocLBi(\price)$ and $\allocUBi(\price)$ are continuous and weakly decreasing in $\price$.
    % \item $\allocUBi(\price)$ is continuous and weakly decreasing in $\price$.
    \item Both $\allocLBi(0) = \allocUBi(0) = 1$, 
            $\lim_{\price\rightarrow \infty}\allocLBi(\price) = \lim_{\price\rightarrow \infty}\allocUBi(\price) = 0$,
        and $\allocLBi(\alloci) \leq \allocUBi(\alloci)$ for all $\alloci\in\reals_+$.
    % \item If $\allocLBi(\price) > 0$, then $\valuefunctionderivativei(\allocLBi(\price))\cdot (1 - \allocLBi(\price)) = \price \cdot \moneycostderivativei(\price\allocLBi(\price))$.
\end{enumerate}
\end{observation}

% \smallskip
% \vspace{10pt}
\noindent We now present the main result of the section.
\begin{theorem}
\label{thm:pure nash single item}
    In the metagame with a single item, pure Nash equilibrium always exists. 
    Specifically, there exist two types of equilibrium:\yfmargincomment{For some specific tie-breaking rule, I think there may also exist other pure Nash, which will disappear when we switch to other tie-breaking rule. In contrast, all equilibrium described below hold for all tie-breaking rule.}
    \begin{enumerate}
        \item \textsl{(Low-price equilibrium)}
        Define non-empty subinterval $\PriceL$ as 
        \begin{align*}
            \PriceL \triangleq \left\{\price\in[0,\infty):\sum_{i\in[n]}\allocLBi(\price) = 1\right\}
        \end{align*}
        For every $\price \in \PriceL$, there exists pure Nash equilibrium whose inner FPPE has per-unit price $\price$. Such equilibrium can be induced by reported message profile $\{(\reportval_i,\reportwealth_i)\}_{i\in[n]}$ constructed as 
        \begin{align*}
            \forall i\in[n]:\qquad \reportval_i \triangleq \infty,~
            \reportwealth_i \triangleq \price  \allocLBi(\price)
        \end{align*}
        Moreover, $\min \PriceL$ is the lowest per-unit price in all pure Nash equilibrium.

        \item \textsl{(High-price equilibrium)}
        % Suppose all agents have linear valuation functions $\{\valuefunctioni\}_{i\in[n]}$ and linear money cost functions $\{\moneycosti\}_{i\in[n]}$.
        Define non-empty subinterval $\PriceH$ as 
        \begin{align*}
            \PriceH \triangleq \left\{\price\in[0,\infty):
            \sum_{i\in[n]}\allocLBi(\price) \leq 1 ~\text{and}~
            \sum_{i\in[n]}\allocUBi(\price) \geq 1
            % ~\text{and}~ 
            % \exists i\in[n]:~
            % \allocUBi = 0
            \right\}
        \end{align*}
        For every $\price \in \PriceH$, 
        if there exists agent $\reserveagent$ such that $\allocLB_{\reserveagent}(\price)+\sum_{i\in[n]:i\not= \reserveagent} \allocUBi(\price) \geq 1$, then there exists pure Nash equilibrium whose inner FPPE has per-unit price~$\price$. Such equilibrium can be induced by reported message profile $\{(\reportval_i,\reportwealth_i)\}_{i\in[n]}$ constructed as 
        \begin{align*}
             i = \reserveagent:\qquad &\reportval_i \triangleq \price,~
            \reportwealth_i \triangleq \infty,
            \\
            \forall i\not= \reserveagent:\qquad &\reportval_i \triangleq \infty,~
            \reportwealth_i \triangleq \price  \allocHati
        \end{align*}
        where $\{\allocHati\}_{i\not= \reserveagent}$ is an arbitrary solution such that 
        $\sum_{i\not=\reserveagent}\allocHati = 1 - \allocLB_{\reserveagent}(\price)$ and $\allocLBi(\price) \leq \allocHati \leq \allocUBi(\price)$.
        Moreover, $\max \PriceH$ is the highest possible per-unit price in all pure Nash equilibrium.
        
    \end{enumerate}
\end{theorem}

By analyzing a two-item budgeted-agent instance from \Cref{example:budgeted agents}, \Cref{sec:pure nash budgeted agents} shows the non-existence of pure Nash equilibrium for general instances. In contrast, \Cref{thm:pure nash single item} confirms the existence of pure Nash equilibrium for single-item instances.

Since $\allocLBi(\price) \leq \allocUBi(\price)$ for all $\price\in\reals_+$ (\Cref{lem:allocLB allocUB}), it is straightforward to verify that $\PriceL \subseteq \PriceH$. Moreover, for every $\price\in\PriceL$, it can be constructed as both low-price equilibrium and high-price equilibrium (since condition ``$\exists\reserveagent$, $\allocLB_{\reserveagent}(\price)+\sum_{i\not= \reserveagent} \allocUBi(\price) \geq 1$'' is satisfied trivially). 
A natural question is whether there exists high-price equilibrium with per-unit price $\price\in\PriceH$ such that there exists no low-price equilibrium with the same per-unit price, i.e., $\price \in \PriceL$.
The answer is yes.
Consider \Cref{example:linear agents}:
\begin{align*}
\begin{array}{ll}
  \allocLBi(\price) = \max\left\{1 - \frac{\price}{\vali}, 0\right\},
    \qquad   
    &  
    \allocUBi(\price) = \indicator{\price \leq \vali}
    \\
    \PriceL = \left\{\frac{\val_1\val_2}{\val_1+\val_2}\right\},
     & 
     \PriceH = \left[\frac{\val_1\val_2}{\val_1+\val_2}, \val_1\right].
\end{array}
\end{align*}
Though the reported messages are slightly different, equilibrium with per-unit price $\val_2$ (resp.\ $\frac{\val_1\val_2}{\val_1+\val_2}$) described in \Cref{claim:linear efficient} (resp.\ \Cref{claim:linear inefficient}) is equivalent to a high-price (resp.\ low-price) equilibrium. 
In fact, we can extend \Cref{example:linear agents} and construct natural scenario with multiple pure Nash equilibria. In those equilibria, the per-unit prices of inner FPPE are different. Consequently, a fixed agent's utility is different in different equilibrium. See the proof of the proposition in \Cref{apx:singleitembudgetedagentmultinash}.

\begin{restatable}{proposition}{singleitembudgetedagentmultinash}
\label{prop:multiple pure nash for budgeted agents}
    In the metagame, for budgeted (or linear) agents with type $\{(\val_i, \truewealth_i)\}$,
    if $\truewealth_i > \frac{1}{4}\val_i$ for all agents, then there exists multiple pure Nash equilibrium. Specifically, besides low-price equilibrium and high-price equilibrium with per-unit price in $\PriceL$, there also exists high-price equilibrium with per-unit price in $\PriceH \backslash \PriceL$.
\end{restatable}

By utilizing the monotonicity of $\allocLBi(\cdot), \allocUBi(\cdot)$, 
we can compute $\PriceL, \PriceH$ and thus construct equilibrium in \Cref{thm:pure nash single item} in polynomial time. 

\begin{restatable}{proposition}{computeeqlb}
\label{prop:single item eqlb computation}
    In the metagame with a single item,
    there exists a polynomial time algorithm that 
    computes pure Nash equilibrium.
\end{restatable}
The formal proof of the proposition is deferred to \Cref{apx:computeeqlb}. 
In \Cref{apx:best response computation},
we also present a polynomial time algorithm 
to compute the best response for metagame with a single item.

\subsection{Proof of \texorpdfstring{\Cref{thm:pure nash single item}}{}}

% In the beginning of the analysis, we make the following observation about $\allocLBi(\price)$ and $\allocUBi(\price)$.

% \begin{lemma}
% \label{lem:allocLB allocUB}
% For every agent $i\in[n]$, functions $\allocLBi(\price)$ and $\allocUBi(\price)$ satisfy the following properties:
% \begin{itemize}
%     \item Both $\allocLBi(\price)$ and $\allocUBi(\price)$ are continuous and weakly decreasing in $\price$.
%     % \item $\allocUBi(\price)$ is continuous and weakly decreasing in $\price$.
%     \item Both $\allocLBi(0) = \allocUBi(0) = 1$, 
%             $\lim_{\price\rightarrow \infty}\allocLBi(\price) = \lim_{\price\rightarrow \infty}\allocUBi(\price) = 0$,
%         and $\allocLBi(\alloci) \leq \allocUBi(\alloci)$ for all $\alloci\in\reals_+$.
%     % \item If $\allocLBi(\price) > 0$, then $\valuefunctionderivativei(\allocLBi(\price))\cdot (1 - \allocLBi(\price)) = \price \cdot \moneycostderivativei(\price\allocLBi(\price))$.
% \end{itemize}
% \end{lemma}
% \begin{proof}
%     All two properties hold since valuation function $\valuefunctioni$ is differentiable, weakly increasing, weakly concave, $\valuefunctioni(0) = 0$,
% and money cost function $\moneycosti$ is differentiable, weakly increasing, weakly convex and $\moneycosti(0) = 0$. 
% \end{proof}

We analyze the low-price (high-price) equilibrium separately.

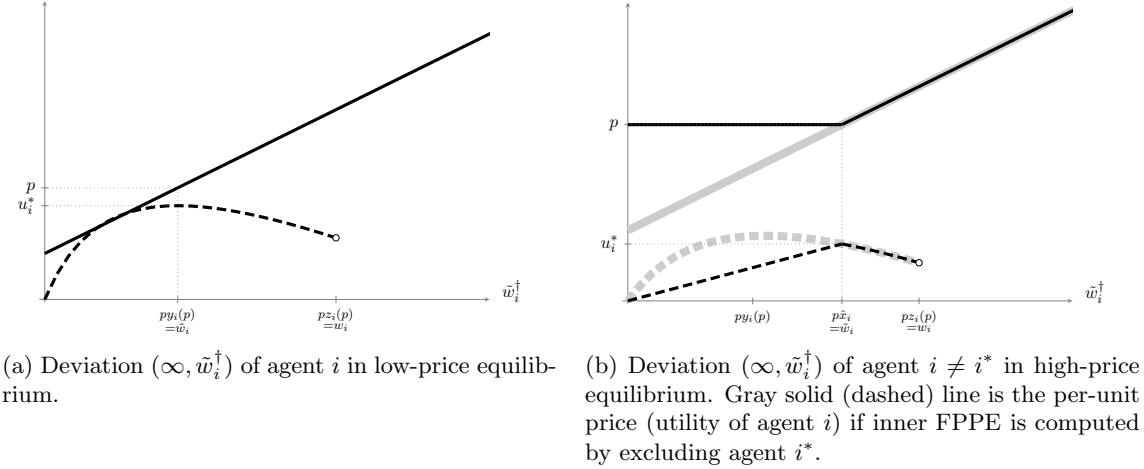
\begin{figure}[hbt]
\centering
\subfloat[Deviation $(\infty,\reportwealth_i\primed)$ of agent $i$ in low-price equilibrium.]{
\begin{tikzpicture}[scale=0.6, transform shape]
\begin{axis}[
axis line style=gray,
axis lines=middle,
xlabel = {$\reportwealth_i\primed$},
xlabel style={at={(axis description cs:1.05,-0.02)}, anchor=south},
xtick={0, 0.24221453, 0.53045},
ytick={0, 0.3460210084, 0.4117647058823529},
xticklabels={0, $\substack{\price\allocLBi(\price) \\ = \reportwealth_i}$, $\substack{ \price\allocUBi(\price) \\ =\truewealth_i }$},
yticklabels={0, $\util_i^*$, $\price$},
xmin=-0.01,xmax=0.81,ymin=-0.01,ymax=1.1,
width=0.7\textwidth,
height=0.5\textwidth]

\addplot[domain=0:1, line width=0.7mm] (x, {0.16955 + x});

\addplot[domain=0:0.53045, dashed, dash pattern=on 6pt off 3pt on 6pt off 3pt, line width=0.7mm] (x, {x/(0.16955 + x)*1-x});

\addplot[color=black, only marks, style={mark=*, fill=white}] coordinates {(0.53045,0.22733553025208542)};

\addplot[domain=0:0.24221453, dotted, gray] (x, {0.4117647058823529});
% \addplot[domain=0.24221453:1, dotted, gray] (x, {0.4117647058823529});
\addplot[domain=0:0.24221453, dotted, gray] (x, {0.3460210084});
% \addplot[domain=0.24221453:1, dotted, gray] (x, {0.3460210084});
\draw [dotted, gray] (axis cs: 0.24221453,0)--(axis cs:0.24221453, 0.4117647058823529);

\end{axis}

\end{tikzpicture}
\label{fig:pure nash single item deviation low price}
}
\quad
\subfloat[Deviation $(\infty,\reportwealth_i\primed)$ of agent $i\not=\reserveagent$ in high-price equilibrium. Gray solid (dashed) line is
the per-unit price (utility of agent $i$)
if inner FPPE is computed by excluding agent $\reserveagent$.]{
% p = 0.65

% xUnderbar1 = 1 - 0.65 = 0.35
% xUnderbar2 = 1 - 0.65 / 0.7 = 0.07142857142857129

% xBar1 = 1
% xBar2 = 1

% xHat1 = 0.6
% xHat2 = 0.4

\begin{tikzpicture}[scale=0.6, transform shape]
\begin{axis}[
axis line style=gray,
axis lines=middle,
xlabel = {$\reportwealth_i\primed$},
xlabel style={at={(axis description cs:1.05,-0.02)}, anchor=south},
xtick={0, 0.22749999999999998, 0.39, 0.53045},
ytick={0, 0.20999999999999996, 0.65},
xticklabels={0, $\substack{\price\allocLBi(\price)}$, $\substack{ \price\allocHati\\=\reportwealth_i}$,
$\substack{\price\allocUBi(\price)\\=\truewealth_i }$},
yticklabels={0, $\util_i^*$, $\price$},
xmin=-0.01,xmax=0.81,ymin=-0.01,ymax=1.1,
width=0.7\textwidth,
height=0.5\textwidth]

\addplot[gray!40!white, domain=0.0:1, line width=1.8mm] (x, {0.26 + x});

\addplot[domain=0.39:1, line width=0.7mm] (x, {0.26 + x});
\addplot[domain=0:0.39, line width=0.7mm] (x, {0.65});

\addplot[gray!40!white, domain=0:0.53045, dashed, dash pattern=on 6pt off 3pt on 6pt off 3pt, line width=1.8mm] (x, {x/(0.26 + x)*1-x});

\addplot[domain=0.39:0.53045, dashed, dash pattern=on 6pt off 3pt on 6pt off 3pt, line width=0.7mm] (x, {x/(0.26 + x)*1-x});

\addplot[domain=0:0.39, dashed, dash pattern=on 6pt off 3pt on 6pt off 3pt, line width=0.7mm] (x, {x/(0.65)*1-x});

\addplot[color=black, only marks, style={mark=*, fill=white}] coordinates {(0.53045,0.1406234391802138)};

\addplot[domain=0:0.39, dotted, gray] (x, {0.65});
\addplot[domain=0:0.39, dotted, gray] (x, {0.20999999999999996});
\draw [dotted, gray] (axis cs: 0.39,0)--(axis cs:0.39, 0.65);

\end{axis}

\end{tikzpicture}
\label{fig:pure nash single item deviation high price}
}
\caption{Graphical illustration of agent $i$'s deviation in
the pure Nash equilibrium analysis for \Cref{thm:pure nash single item}. The solid (dashed) line is the per-unit price (utility of agent $i$). The utility of agent~$i$ in the equilibrium is $\util_i^*$. In this graphical example, agent $i$ has
hard budget $\truewealth_i$ and thus her utility becomes negative infinite for all $\reportwealth_i\primed > \truewealth_i$.}
\label{fig:pure nash single item deviation}
\end{figure}

\paragraph{Low-price equilibrium}
Since both $\allocLBi(\price)$ and $\allocUBi(\price)$ is continuous and weakly decreasing with $\allocLBi(0) = \allocUBi(0) = 1$ and 
$\lim_{\price\rightarrow \infty}\allocLBi(\price) = \lim_{\price\rightarrow \infty}\allocUBi(\price) = 0$ (\Cref{lem:allocLB allocUB}), set $\PriceL$ is a non-empty subinterval. 

Fix an arbitrary $\price\in\priceL$.
Now consider the reported message profile $\{(\reportval_i \triangleq \infty,
\reportwealth_i \triangleq \price \allocLBi(\price))\}_{i\in[n]}$.
By construction, the inner FPPE has per-unit price $\sum_{i\in[n]}\reportwealth_i = \sum_{i\in[n]}\price\allocLBi(\price) = \price$, and every agent~$i$ 
has allocation $\allocLBi(\price)$.
Next we verify that this message profile forms a pure Nash equilibrium.
A graphical illustration of the argument can be found at \Cref{fig:pure nash single item deviation low price}.
Fix an arbitrary agent $i\in[n]$. By \Cref{lem:reporting infinite value},
it is sufficient to argue that 
\begin{align*}
% \label{eq:single item low price eqlb}
    \reportwealth_i \in \argmax_{\reportwealth_i\primed}
        \util_i(\infty,\reportwealth_i\primed, \reportval_{-i}, \reportwealth_{-i})
\end{align*}
Note that by deviating to message $(\reportval_i\primed = \infty, \reportwealth_i\primed)$, the new inner FPPE 
has per-unit price $\price\primed = \reportwealth_i\primed + \sum_{\ell\in[n]:\ell\not= i}\reportwealth_\ell$,
allocation $\alloci\primed = \frac{\reportwealth_i\primed}{\reportwealth_i\primed + \sum_{\ell\in[n]:\ell\not= i}\reportwealth_\ell}$ 
and payment $\paymenti\primed = \reportwealth_i\primed$. Therefore,
the utility $\util_i(\infty,\reportwealth_i\primed, \reportval_{-i}, \reportwealth_{-i})$
has the following closed-form:
for every $\reportwealth_i\primed \leq \truewealth_i$,
\begin{align*}
     \util_i(\infty,\reportwealth_i\primed, \reportval_{-i}, \reportwealth_{-i}) 
    &=
     \valuefunctioni\left(
     \frac{\reportwealth_i\primed}{\reportwealth_i\primed + \sum_{\ell\in[n]:\ell\not= i}\reportwealth_\ell}
     \right)
     -
     \moneycosti\left(\reportwealth_i\primed\right)
\end{align*}
which is concave in $\reportwealth_i\primed$.
Taking the derivative of the right-hand side, we obtain 
\begin{align*}
    \frac{\sum_{\ell\in[n]:\ell\not= i}\reportwealth_\ell}{\left(\reportwealth_i\primed + \sum_{\ell\in[n]:\ell\not= i}\reportwealth_\ell\right)^2}
    \cdot 
    \valuefunctionderivativei\left(
     \frac{\reportwealth_i\primed}{\reportwealth_i\primed + \sum_{\ell\in[n]:\ell\not= i}\reportwealth_\ell}
     \right)
     -
     \moneycostderivativei\left(\reportwealth_i\primed\right) 
\end{align*}
which can be re-written as 
\begin{align*}
    \valuefunctionderivativei(\alloci\primed)\cdot \frac{(1-\alloci\primed)}{\price\primed} - \moneycostderivativei(\price\primed\alloci\primed)
\end{align*}
Recall $\allocLBi(\price) \triangleq \frac{\truewealth_i}{\price}\wedge\min\left\{\alloci\in[0, 1]: \valuefunctionderivativei(\alloci)\cdot (1-\alloci) \leq \price \cdot \moneycostderivativei(\price\alloci)\right\}$.
If $\price\allocLBi(\price) < \truewealth_i$,
then $\reportwealth_i\primed = \price\allocLBi(\price)$ is the maximizer of $\util_i(\infty,\reportwealth_i\primed, \reportval_{-i}, \reportwealth_{-i})$,
since 
$\util_i(\infty,\reportwealth_i\primed, \reportval_{-i}, \reportwealth_{-i}) = \valuefunctioni(
\sfrac{\reportwealth_i\primed}{(\reportwealth_i\primed + \sum_{\ell\in[n]:\ell\not= i}\reportwealth_\ell)}
) - \moneycosti(\reportwealth_i\primed)$
for $\reportwealth_i\primed \in [0, \truewealth_i]$ and the right-hand side is concave in $\reportwealth_i\primed$.
Similarly, if $\price\allocLBi(\price) = \truewealth_i$,
then $\util_i(\infty,\reportwealth_i\primed, \reportval_{-i}, \reportwealth_{-i})$
is weakly increasing and weakly positive in $\reportwealth_i\primed \in[0, \truewealth_i]$,
and thus
$\reportwealth_i\primed = \price\allocLBi(\price) = \truewealth_i$ is again the maximizer attained on the boundary point.

Finally, we argue that there exists no pure Nash equilibrium whose the inner FPPE 
has a per-unit price strictly smaller than $\min\PriceL$ by contradiction.
Suppose there exists a pure Nash equilibrium with per-unit price $\underbar\price < \min\PriceL$.
Due to the monotonicity of $\allocLBi(\cdot)$ (\Cref{lem:allocLB allocUB}),
it guarantees that $\sum_{i\in[n]}\allocLBi(\underbar\price) > 1$.
Thus, there exists an agent $i$ with allocation $\alloci < \allocLBi(\underbar\price)$.
A similar argument as the one in \Cref{lem:poa pure allocation price relation} guarantees that agent $i$ strictly prefers
to increases her reported budget by a sufficiently small amount, which is a contradiction as desired.

\paragraph{High-price equilibrium}
Since both $\allocLBi(\price)$ and $\allocUBi(\price)$ is continuous and weakly decreasing with $\allocLBi(0) = \allocUBi(0) = 1$ and 
$\lim_{\price\rightarrow \infty}\allocLBi(\price) = \lim_{\price\rightarrow \infty}\allocUBi(\price) = 0$ (\Cref{lem:allocLB allocUB}), set $\PriceH$ is a non-empty subinterval. 

Fix an arbitrary $\price\in\priceH$.
Suppose there exists agent $\reserveagent$ such that $\allocLB_{\reserveagent}(\price) + \sum_{i\in[n]:i\not= \reserveagent} \allocUBi(\price) \geq 1$.
Now consider the reported message profile $\{(\reportval_i,\reportwealth_i)\}_{i\in[n]}$ constructed as 
\begin{align*}
i = \reserveagent:\qquad &\reportval_i \triangleq \price,~
\reportwealth_i \triangleq \infty,
\\
\forall i\not= \reserveagent:\qquad &\reportval_i \triangleq \infty,~
\reportwealth_i \triangleq \price  \allocHati
\end{align*}
where $\{\allocHati\}_{i\not= \reserveagent}$ is an arbitrary solution such that 
$\sum_{i\not=\reserveagent}\allocHati = 1 - \allocLB_{\reserveagent}(\price)$ and $\allocLBi(\price) \leq \allocHati \leq \allocUBi(\price)$.
By construction, the inner FPPE has per-unit price $\price$, and every agent $i\not=\reserveagent$ ($i = \reserveagent$)
has allocation $\allocHati$ (allocation $\allocLB_{\reserveagent}(\price)$).
Next we verify that this message profile forms a pure Nash equilibrium.
A graphical illustration of the argument can be found at \Cref{fig:pure nash single item deviation high price}.
Fix an arbitrary agent $i\in[n]$. By \Cref{lem:reporting infinite value},
it is sufficient to argue that 
\begin{align*}
% \label{eq:single item low price eqlb}
    \reportwealth_i \in \argmax_{\reportwealth_i\primed}
        \util_i(\infty,\reportwealth_i\primed, \reportval_{-i}, \reportwealth_{-i})
\end{align*}
We analyze agent $i\not=\reserveagent$ and $i=\reserveagent$ separately.

Consider agent $i\not=\reserveagent$.
By deviating to message $(\reportval_i\primed = \infty, \reportwealth_i\primed)$ with $\reportwealth_i\primed < \reportwealth_i$, the new inner FPPE 
has the same per-unit price $\price\primed = \price$,
with smaller 
allocation $\alloci\primed = \sfrac{\reportwealth_i\primed}{\price}$ 
and payment $\paymenti\primed = \reportwealth_i\primed$.\footnote{Note that this is different from low-price equilibrium argument, where the new per-unit price decreases. In the high-price equilibrium, the per-unit price remains the same after deviation $\reportwealth_i\primed < \reportwealth_i$, since there exist agent $i'$ in $\lagent$ who reports $(\reportval_{i'} = \price, \reportwealth_{i'} = \infty)$ and avoids the decrease of the per-unit price.}
Note that $\alloci\primed = \sfrac{\reportwealth_i\primed}{\price} < \sfrac{\reportwealth_i}{\price} = \allocHati \leq \allocUBi(\price)$.
The construction of $\allocUBi(\price)$ implies that the utility of agent $i$ 
weakly decreases by deviating to message $(\reportval_i\primed = \infty, \reportwealth_i\primed)$ with $\reportwealth_i\primed < \reportwealth_i$.
On the other side, by deviating to to message $(\reportval_i\primed = \infty, \reportwealth_i\primed)$ with $\reportwealth_i\primed \geq \reportwealth_i$,
using a similar argument as the one for low-price equilibrium, 
it can be verified that the utility is concave and decreasing for $\reportwealth_i\primed \in [\reportwealth_i,\truewealth_i]$ and such deviation is not profitable.
Therefore, we obtain $\reportwealth_i \in \argmax_{\reportwealth_i\primed}
\util_i(\infty,\reportwealth_i\primed, \reportval_{-i}, \reportwealth_{-i})$
as desired. 

Consider agent $\reserveagent$.
By construction,
agent $i$ reports message $(\reportval_i = \price, \reportwealth_i = \infty)$ and has allocation  $\allocLB_{\reserveagent}(\price)$ and payment $\price\allocLB_{\reserveagent}(\price)$.
Note that the utility of agent $\reserveagent$ is equivalent to the utility by reporting  $(\reportval_i\primed = \infty, \reportwealth_i\primed = \price\allocLB_{\reserveagent}(\price))$.
Following the same argument as the one for the low-price equilibrium, it maximizes her utility due to the concavity of the utility as a function of $\reportwealth_i\primed$.

Finally, we argue that there exists no pure Nash equilibrium whose inner FPPE 
has a per-unit price strictly larger than $\max\PriceH$ by contradiction.
Suppose there exists a pure Nash equilibrium with per-unit price $\bar\price > \max\PriceL$.
Due to the monotonicity of $\allocUBi(\cdot)$ (\Cref{lem:allocLB allocUB}),
it guarantees that $\sum_{i\in[n]}\allocUBi(\bar\price) < 1$.
Thus, there exists an agent $i$ with allocation $\alloci > \allocUBi(\bar\price)$.
The construction of $\allocUBi(\bar\price)$
implies that $\valuefunctionderivativei(\alloci) < \bar\price\cdot \moneycostderivativei(\bar\price\alloci)$.
Thus, agent $i$'s utility can be strictly increases if she decreases
her reported budget by a sufficiently small amount, which is a contradiction as desired.

}
% \newpage

% \section{Conclusion}
% \input{Paper/conclusion}

% \newpage

\bibliographystyle{apalike}
	\bibliography{refs.bib}
\appendix
\section*{Appendix}
\section{Original Definition of First-price Pacing Equilibrium}
\label{apx:FPPE original definition}

The original definition of the first-price pacing equilibrium in \cite{CKPSSSW-22} is as follows.

\begin{definition}
\label{def:FPPE old}
For budgeted agents with types $\{(\val_i, \truewealth_i)\}_{i\in[n]}$, a \emph{first-price pacing equilibrium (FPPE)} is a tuple $(\pacescalar, \alloc)$ of pacing multipliers $\pacescalar_i \in [0, 1]$ for each item $j$, and allocation $\alloci\in[0,1]^m$ for each agent $i$ that satisfies the following properties:
\begin{enumerate}
    \item \textsl{(price)} per-unit price $\pricej = \max_{i\in[n]} \pacescalar_i\vali\ctrij$
    \item \textsl{(items go to highest agents)} if $\allocij > 0$, then $\pacescalar_i\vali\ctrij = \max_{i'\in[n]}\pacescalar_{i'}\val_{i'}\ctr_{i'j}$
    \item \textsl{(budget-feasible)} $\sum_{j\in[m]}\allocij\pricej \leq \truewealthi$
    \item \textsl{(demanded items sold completely)} if $\pricej > 0$, then $\sum_{i\in[n]}\allocij = 1$
    \item \textsl{(no overselling)} $\sum_{i\in[n]}\allocij \leq 1$
    \item \textsl{(no unnecessary pacing)} if $\sum_{j\in[m]}\allocij\pricej < \truewealthi$, then $\pacescalar = 1$.
\end{enumerate}
\end{definition}
We verify the equivalence between \Cref{def:FPPE} and \Cref{def:FPPE old}. 

From \Cref{def:FPPE old} to \Cref{def:FPPE}, consider the same allocation, and let the per-unit price $\pricej$ in \Cref{def:FPPE} be $\pricej=\max_{i\in[n]}\pacescalar_i\vali\ctrij$ due to the ``price'' property in \Cref{def:FPPE old}. The ``highest band-per-buck'' property in \Cref{def:FPPE} is implied by the ``price'' property and ``items go to highest agents'' property in \Cref{def:FPPE old}. The ``supply feasibility'' property in \Cref{def:FPPE} is implied by the ``demanded item sold completely'' property and ``no overselling'' property. The ``budget feasibility'' property in \Cref{def:FPPE} is implied by the ``budget-feasible'' and `no unnecessary pacing'' property in \Cref{def:FPPE old}. 

From \Cref{def:FPPE} to \Cref{def:FPPE old}, consider the same allocation, and let the pacing multiplier $\pacescalar_i$ in \Cref{def:FPPE old} be $\pacescalar_i = 1 \wedge (\min_{j\in[m]}\sfrac{\pricej}{\vali\ctrij})\in[0, 1]$. Then the ``price'' property and ``items go to highest agents'' property in \Cref{def:FPPE old} are implied by the pacing multiplier construction and ``highest bang-per-buck'' property in \Cref{def:FPPE}. 
The ``budget-feasible'' property and ``no unnecessary pacing'' property in \Cref{def:FPPE old} in \Cref{def:FPPE old} are implied by the ``budget feasibility'' property and ``payment calculation'' property in \Cref{def:FPPE}. The ``demanded items sold completely'' property and ``no overselling'' property in \Cref{def:FPPE old} are implied by ``supply feasibility'' property in \Cref{def:FPPE}.

\section{Omitted Proofs}
\label{apx:proofs}

\subsection{Proof of \texorpdfstring{\Cref{lem:compute utility}}{}}
\label{apx:computeutility}
\computeutility*

\begin{proof}
    Let $\alloci$ be the allocation of agent $i$.
    By definition, given allocation $\alloci$ and payment $\paymenti$,
    agent~$i$'s utility $\util_i(\alloci,\paymenti)$ is
    \begin{align*}
        \util_i(\alloci,\paymenti) =
    \left\{
    \begin{array}{ll}
    \valuefunctioni\left(
    \sum_{j\in[m]}
    \ctrij\allocij\right)
     -
    \moneycosti\left(\paymenti\right)&
    \quad \text{if } \paymenti \leq \truewealthi\\
    -\infty     &
     \quad \text{if } \paymenti > \truewealthi
    \end{array}
    \right.
    \end{align*}
    Note that the lemma statement is satisfied immediately if $\paymenti > \truewealthi$.

    Now we consider the case where $\paymenti\leq \truewealthi$.
    Define $\purchaseseti \triangleq \argmax_{j\in[m]} \frac{\ctrij}{\pricej}$
    as the subset of items that achieves the highest bang-per-buck for agent $i$.
    The ``highest bang-per-buck'' property of FPPE ensures that
    $\allocij > 0$ if $j \in \purchaseseti$.
    Hence,
    \begin{align*}
        \sum_{j\in[m]} \ctrij\allocij
        =
        \sum_{j\in[m]} \frac{\ctrij}{\pricej}\cdot \pricej\allocij
        =
        \sum_{j\in \purchaseseti} \frac{\ctrij}{\pricej}\cdot \pricej\allocij
        =
        \left(\max_{j\in[m]}\frac{\ctrij}{\pricej}\right)\cdot \paymenti
    \end{align*}
    where the last equality holds due to the definition of $\purchaseseti$ and
    the ``payment calculation'' property of FPPE, i.e., $\paymenti = \sum_{j\in[m]}\pricej\allocij = \sum_{j\in \purchaseseti}\pricej\allocij$.
\end{proof}

\subsection{Proof of \texorpdfstring{\Cref{lem:reporting infinite value}}{}}
\label{apx:reportinginfinitevalue}
\reportinginfinitevalue*

To prove \Cref{lem:reporting infinite value}, we need the following technical lemma.

\begin{lemma}
\label{lem:finite price}
    In the metagame, for every pure Nash equilibrium and the per-unit prices $\price$ of its inner FPPE, it satisfies that $\pricej \not= \infty$ for every item $j$.
\end{lemma}
\begin{proof}
    We prove the statement by contradiction. 
    Suppose there exists item $j$ such that $\pricej =\infty$.
    Due to the ``supply feasibility'' property and ``payment calculation'' property in the definition of FPPE, 
    there exists an agent whose payment is $\infty$ and utility is $-\infty$, which contradicts the assumption that the per-unit prices $\price$ is induced by an equilibrium.
\end{proof} 

Now we are ready to prove \Cref{lem:reporting infinite value}.

\begin{proof}[Proof of \Cref{lem:reporting infinite value}]
    The first equality in the lemma statement holds due to the definition of pure Nash equilibrium.
    For the second equality, it is sufficient to show that there exists $\reportwealth_i\primed$ such that 
    $\util_i(\reportval_i,\reportwealth_i, \reportval_{-i}, \reportwealth_{-i}) = \util_i(\infty,\reportwealth_i\primed, \reportval_{-i}, \reportwealth_{-i})$.
    In particular, consider the construction that $\reportwealth_i\primed \triangleq \paymenti(\reportval_i,\reportwealth_i, \reportval_{-i}, \reportwealth_{-i})$,
    where $\paymenti(\reportval_i,\reportwealth_i, \reportval_{-i}, \reportwealth_{-i})$
    is the payment of agent $i$ in equilibrium $\{(\reportval_i, \reportwealth_i), (\reportval_{-i}, \reportwealth_{-i})\}$.
    By checking the properties in the definition of FPPE one by one, we can verify that for each item $j$, the unique per-unit price $\price_j$ induced by $\{(\reportval_i, \reportwealth_i), (\reportval_{-i}, \reportwealth_{-i})\}$ is the same as the unique per-unit price $\price_j\primed$ induced by $\{(\infty, \reportwealth_i\primed), (\reportval_{-i}, \reportwealth_{-i})\}$.
    Invoking \Cref{lem:finite price}, it ensures that $\pricej\primed = \pricej < \infty$.
    Therefore, due to the ``budget feasibility'' property, agent $i$ exhausts her reported budget,
    i.e., $\paymenti(\infty,\reportwealth_i\primed , \reportval_{-i}, \reportwealth_{-i}) =
    \reportwealth_i\primed = \paymenti(\reportval_i,\reportwealth_i, \reportval_{-i}, \reportwealth_{-i})$.
    Consequently, $\util_i(\reportval_i,\reportwealth_i, \reportval_{-i}, \reportwealth_{-i}) = \util_i(\infty,\reportwealth_i\primed, \reportval_{-i}, \reportwealth_{-i})$
    due to \Cref{lem:compute utility}.
\end{proof}

\subsection{Proof of \texorpdfstring{\Cref{lem:agent in favor tie-breaking}}{}}
\label{apx:favortiebreaking}
\favortiebreaking*
\begin{proof}
    Invoking \Cref{lem:reporting infinite value}, we have
    \begin{align*}
        \util_i^\allocselectionrule(\reportval_i, \reportwealth_i, \reportval_{-i}, \reportwealth_{-i})
        =
        \max_{\reportwealth_i\primed}
        \util_i^\allocselectionrule(\infty, \reportwealth_i\primed, \reportval_{-i}, \reportwealth_{-i})
    \end{align*}
    Using a similar argument as the one in \Cref{lem:reporting infinite value},
    for every tie-breaking rule $\allocselectionrule'$,
    \begin{align*}
        \util_i^{\allocselectionrule'}(\reportval_i, \reportwealth_i, \reportval_{-i}, \reportwealth_{-i})
        \leq
        \max_{\reportwealth_i\primed}
        \util_i^{\allocselectionrule'}(\infty, \reportwealth_i\primed, \reportval_{-i}, \reportwealth_{-i})
    \end{align*}
    Finally, since under message profile 
    $\{(\infty, \reportwealth_i\primed), (\reportval_{-i}, \reportwealth_{-i})\}$,
    agent $i$ exhausts her reported budget $\reportwealth_i\primed$ due to the ``budget feasibility'' property of FPPE, 
    \Cref{lem:compute utility,lem:FPPE uniqueness}
    ensure that the per-unit prices as well as agent $i$'s utility do not depend on the specific choice of tie-breaking rule.
    Putting all pieces together,
    for every tie-breaking rule $\allocselectionrule'$,
    \begin{align*}
        \util_i^{\allocselectionrule'}(\reportval_i, \reportwealth_i, \reportval_{-i}, \reportwealth_{-i})
        \leq
        \max_{\reportwealth_i\primed}
        \util_i^{\allocselectionrule'}(\infty, \reportwealth_i\primed, \reportval_{-i}, \reportwealth_{-i})
        =
        \max_{\reportwealth_i\primed}
        \util_i^{\allocselectionrule}(\infty, \reportwealth_i\primed, \reportval_{-i}, \reportwealth_{-i})
        =
        \util_i^\allocselectionrule(\reportval_i, \reportwealth_i, \reportval_{-i}, \reportwealth_{-i})
    \end{align*}
    as desired.
\end{proof}

\subsection{Proof of \texorpdfstring{\Cref{claim:budget two allocation}}{}}
\label{apx:budgettwoalloc}
\budgettwoalloc*
\begin{proof}
        We prove this by contradiction. Suppose there exists a pure Nash equilibrium as desired.
        Without loss of generality, we assume $\alloc_{11} = 1$ and $\alloc_{12} > 0$, Namely, agent 1 receives her favored item 1 entirely and a positive fraction of item 2.
        Let $\price_1, \price_2$ be the per-unit prices of the inner FPPE. The ``highest bang-per-buck'' property implies $\frac{1}{2}\price_1 = \price_2 \triangleq \price$.
        Moreover, let $\payment_i$ be the payments for each agent $i$ in the equilibrium. Below we consider three cases depending agents' payments.

        \paragraph{Case (i) agent $1$ does not exhaust her true budget $\truewealth_1$, i.e., $\payment_1 < \truewealth_1$.}
        In this case, we know that agent 2 does not exhaust her true budget as well, i.e., $\payment_2 < \truewealth_2$. 
        Otherwise, the per-unit price $\price_1$ of item 1 is at least $\price_1 = 2 \price_2 \geq 2\truewealth_2$ and thus agent $1$ receives negative utility due to the violation of her budget constraint.
        Since both agents do not exhaust their budgets, invoking \Cref{lem:poa pure allocation price relation}, we have
        \begin{align*}
            \frac{\val_1}{1} \leq \frac{\price_1 + \price_2}{(1-\alloc_{11})\ctr_{11} + (1-\alloc_{12})\ctr_{12}},
            \quad 
            \frac{\val_2}{1} \leq \frac{\price_2}{(1 - \alloc_{22})\ctr_{22}}
        \end{align*}
        which can be simplified as 
        \begin{align*}
            \price_{11} \geq 2\max\left\{\frac{2}{3}\alloc_{22}, 1 - \alloc_{22}\right\} \geq \frac{4}{5}
        \end{align*}
        where the second inequality by considering all possible $\alloc_{22}$. 
        Notice that this leads to a contradiction since agent 1 receives a negative utility due to the violation of her budget constraint.
        
        \paragraph{Case (ii) agent 1 exhausts her true budget $\truewealth_1$, i.e., $\payment_1 = \truewealth_1$.}
        Since agent $1$ exhausts her true budget $\truewealth_1$, we have 
        \begin{align*}
            \truewealth_1 = \payment_1 
            = \price_1\alloc_{11} + \price_2 \alloc_{12}
            = (2 + \alloc_{12})\price_2
        \end{align*}
        
        Using the same argument as the one in case (i), 
        we know that agent 2 does not exhaust her true budget. 
        Using a similar argument as the one in \Cref{lem:poa pure allocation price relation}, we consider a deviation of agent 2 by
        increasing her reported budget for a sufficiently small amount $\epsilon > 0$ and let reported value be infinity. Under this deviation, agent $1$ still wins strictly positive fraction for both items. Thus, the new per-unit prices of inner FPPE under the deviation is $\price_1\primed = \price_1 + \frac{2}{3}\epsilon$ 
        and $\price_2\primed = \price_2 + \frac{1}{3}\epsilon$.
        Given that this is not a profitable deviation, we have
        \begin{align*}
            \val_2 \ctr_{22}\alloc_{22} - \price_2 \alloc_{22} 
            \geq 
            \val_2 \ctr_{22} \frac{\price_2\alloc_{22} + \epsilon}{\price_2 + \frac{1}{3}\epsilon} - (\price_2\alloc_{22} + \epsilon)
        \end{align*}
        Letting $\epsilon$ approach zero, it can be simplified as 
        \begin{align*}
            \price_2 \geq 1 - \frac{1}{3}\alloc_{22}
        \end{align*}
        Together with $\wealth_1 = (2 + \alloc_{12})\price_2$, we obtain 
        $\frac{1}{2(3-\alloc_{22})} \geq 1-\frac{1}{3}\alloc_{22}$ which is not satisfied for all $\alloc_{22} \in [0, 1]$ and thus leads to a contradiction as desired.
\end{proof}

\subsection{Proof of \texorpdfstring{\Cref{lem:fppe increasing budget}}{}}
\label{apx:increasingBudget}
\increasingBudget*

The proof of \Cref{lem:fppe increasing budget} relies on the following concept and lemmas in \cite{CKPSSSW-22}.

\begin{definition}[Budget-feasible first-price pacing multiplier]
For budgeted agents with types $\{(\val_i, \truewealth_i)\}_{i\in[n]}$,
a set of \emph{budget-feasible first-price pacing multipliers (BFPM)} is a tuple $(\pacescalar, \alloc)$, of
pacing multipliers $\pacescalar \in [0, 1]$ for each bidder $i \in N$, and factional allocation $\allocij\in  [0, 1]$ for agent
$i \in [n]$ and item $j \in [m]$ with the following properties:
\begin{itemize}
    \item (Price) Per-unit price $\price_j = \max_{i} \pacescalar_i \val_{ij}$
    \item (Highest bid wins) 
    If $\alloc_{ij} > 0$, then $\pacescalar_i \val_{ij} = \price_j$
    \item (Budget feasible) $\sum_{j} \alloc_{ij}\price_j \leq \wealth_i$
    \item (Demanded goods sold completely)
    If $\price_j > 0$, then $\sum_{i}\alloc_{ij} = 1$
    \item (No overselling) $\sum_{i}\alloc_{ij}  \leq 1$
\end{itemize}
\end{definition}

As we mentioned in \Cref{sec:prelim}, FPPE can be interpreted as a BFPM where the pacing multipliers are defined as $\pacescalar_i = \min\{\max_{j\in[m]}\sfrac{\vali\ctr_{ij}}{\price_{j}}, 1\}\in[0, 1]$ for each agent $i$.

\begin{lemma}[\citealp{CKPSSSW-22}]
\label{lem:BFPM dominance}
For any set of budgeted agents, FPPE Pareto-dominates all BFPMs in pacing multipliers, per-unit prices, and revenue.
\end{lemma}

\begin{lemma}[\citealp{CKPSSSW-22}]
\label{lem:fppe adding item}
\label{lem:fppe adding agent}
    In an FPPE, adding an item (agent) weakly increases revenue.
\end{lemma}

Now, we are ready to prove \Cref{lem:fppe increasing budget}.

\begin{proof}[Proof of \Cref{lem:fppe increasing budget}]
    Let $(\pacescalar, \alloc)$/$(\pacescalar\primed, \alloc\primed)$ 
    be the FPPE before/after
    agent $i$'s budget increasing.
    Let $\price$/$\price\primed$
    be their induced per-unit prices respectively.
    It can be verified that $(\pacescalar, \alloc)$
    is still a valid BFPM after budget increasing.
    By \Cref{lem:BFPM dominance}, $\pacescalar\primed_\ell\geq \pacescalar_\ell$
    for every agent $\ell$,
    and thus per-unit price $\price_j\primed \geq \price_j$ for every item $j$.

    Now consider an instance $\instance\doubleprimed$ where we exclude agent $i$
    while holding all other parameters fixed.
    By \Cref{lem:fppe adding agent},
    the revenue of FPPE in this instance $\instance\doubleprimed$ 
    is at most 
    the revenue of FPPE in the original instance where agent $i$ has budget $\wealth_i$,
    i.e., at most $\sum_j \price_j$.
    It is sufficient to show the revenue of FPPE in $\instance\doubleprimed$ is 
    at least $\sum_j \price_j\primed - \wealth_i\primed$.
    To see this, consider the following pacing multipliers $\pacescalar\doubleprimed\triangleq \pacescalar\primed_{-i}$ and 
    allocations $\alloc\doubleprimed \triangleq \alloc\primed_{-i}$.
    The tuple $(\pacescalar\doubleprimed, \alloc\doubleprimed)$ satisfies all
    constraints of BFPM for instance $\instance\doubleprimed$ except 
    (demand goods sold completely).
    Note that we can further decompose every items into smaller pieces and discard 
    pieces with no allocation under $\allocs\doubleprimed$.
    In this way, $(\pacescalar\doubleprimed, \allocs\doubleprimed)$ becomes a valid
    BFPM after removing those (fractional) items,
    and its revenue is at least $\sum_j \price_j\primed - \wealth_i\primed$.
    Invoking \Cref{lem:BFPM dominance} and \Cref{lem:fppe adding item} finishes the proof.
\end{proof}

\subsection{Proof of \texorpdfstring{\Cref{thm:poa bayes}}{}}
\label{apx:poabayes}
\poabayes*

\newcommand{\type}{z}
\newcommand{\typei}{\type_i}
\newcommand{\typeMinusI}{\type_{-i}}
\newcommand{\optexpectalloc}{X^*}
\newcommand{\optexpectallocij}{\optexpectalloc_{ij}}

Note that \Cref{example:pure lower bound} also serves as a lower bound for \Cref{thm:poa bayes}.
For the upper bound, we use a similar argument as the one in \Cref{thm:poa mixed}.
In this analysis, we introduce auxiliary notation $\type_i \triangleq (\valuefunc_i, \truewealth_i,\moneycost_i)$ to denote agent $i$'s realized type, and notation $\typeMinusI$ for other agents $-i$ similarly.

\begin{proof}[Proof of \Cref{thm:poa bayes}]
   Fix an arbitrary Bayesian Nash equilibrium $\strategy$. 
Fix an arbitrary agent $i$. 
Let $\optalloc(\typei, \typeMinusI)$ be the optimal allocation that maximizes the liquid welfare when 
agent $i$ has type $\typei$, and other agents have types $\typeMinusI$.
For every type $\typei$, 
define $\optexpectalloc_i(\typei) \triangleq \expect[\typeMinusI \sim \typedist_{-i}]{\optalloc_i(\typei, \typeMinusI)}$,
For every type $\typei$,
let $\randomalloc(\typei),\randomprice(\typei)$ be the randomized allocation and per-unit prices of the inner FPPE when agent~$i$ has type $\typei$ under the equilibrium. 
In both $\randomalloc(\typei),\randomprice(\typei)$, the randomness is taken over agent $i$'s message given strategy $\strategy_i$, other agents' types $\typeMinusI$ 
and their messages given strategy $\strategy_{-i}$.
Similarly, let $\randomalloc\doubleprimed,\randomprice\doubleprimed$ be the randomized allocation and per-unit prices of the inner FPPE when agent~$i$ reports $(\reportval_i = 0, \reportwealth_i = 0)$ and all other agents use their equilibrium strategy $\strategy_{-i}$.
By \Cref{lem:fppe increasing budget}, $\randomprice\primed$
is stochastically dominated by $\randomprice(\typei)$ for all $\typei$.

In the first step of our argument, we prove the following inequality
for every type $\typei$ for agent $i$:
    \begin{align}
    \label{eq:multi item bayes wtp bound}
        \frac{1}{2}\wtp_i(\optexpectalloc_i(\typei)) \leq 
        \moneycost_i^{-1}(\util_i(\typei)) + 
        \sum_{j\in[m]}\expect[\pricej\sim\randompricej\doubleprimed]{\pricej}\optexpectallocij(\typei) 
    \end{align}
    where $\util_i(\typei)$ is the expected utility of agent $i$ with type $\typei$ in the equilibrium.
    We proof inequality~\eqref{eq:multi item bayes wtp bound}
    using the same argument as the one in \Cref{thm:poa mixed}.
    Specifically, 
    we construct a specific deviation based on $\optexpectalloc_i(\typei)$ for agent $i$ with type $\typei$.
    Let $\tildePj\doubleprimed \triangleq \expect[\pricej\doubleprimed\sim\randompricej\doubleprimed]{\pricej\doubleprimed}$
    for every item $j\in[m]$.
    Consider the following deviation for agent $i$
    where agent $i$ deterministically reports 
      $\reportval_i\primed = \infty,
        \reportwealth_i\primed = 
        \sum_{j\in[m]} 
        \tildePj\doubleprimed
        \optexpectallocij(\typei)$.
    Let $\randomalloc\primed,\randomprice\primed$ be the randomized allocation and per-unit prices of the inner FPPE after this deviation of agent $i$.
    Since agent $i$ reports $\val_i\primed = \infty$, 
    she always exhausts her reported budget$~\reportwealth_i\primed$.
    
    In the remaining analysis,
    we couple the reported message of all other agents except agent~$i$ in the FPPE inducing randomized per-unit prices $\randomprice\doubleprimed$, and the FPPE inducing randomized 
    allocation~$\randomalloc\primed$ and per-unit prices $\randomprice\primed$.
    It is important to note that 
    \Cref{lem:fppe increasing budget} implies
    \begin{align*}
        \pricej\doubleprimed \leq \pricej\primed
        \text{ for every $j\in[m]$, and }
        \sum_{j\in[m]}\pricej\primed \leq 
        \sum_{j\in[m]}\pricej\doubleprimed + \reportwealth_i\primed
    \end{align*}
    for every realization of agents' message profile.

    Define auxiliary random variable 
    $\boldsymbol\tau_i \triangleq 
    \frac{\reportwealth_i\primed}
    {\sum_{j\in[m]}
    \randompricej\doubleprimed\optexpectallocij(\typei)} \geq 0$.
    By definitions of $\boldsymbol\tau_i$ and $\reportwealth_i\primed$, $\expect{\frac{1}{\boldsymbol\tau_i}} = 1$.
    The expected utility $\util_i$
    of agent $i$ in the equilibrium can be lowerbounded as her expected utility under this deviation. Namely,
    \begin{align*}
        \util_i 
        &\overset{(a)}{\geq} 
        \expect
        [\price\primed\sim\randomprice\primed]
        {
        \valuefunc_i\left(
        \left(
        \max_{j\in[m]}
        \frac{\ctrij}{\pricej\primed}
        \right)
        \reportwealth_i\primed
        \right)
        -
        \moneycost_i\left(\reportwealth_i\primed\right)}
        \\
        &\overset{(b)}{\geq} 
        \expect
        [\price\primed\sim\randomprice\primed,\tau_i\sim\boldsymbol\tau_i]
        {
        \valuefunc_i\left(
        \left(
        \max_{j\in[m]}
        \frac{\ctrij}{\pricej\primed}
        \right)
        \left(
        \sum_{j\in[m]}
        \frac{1}{\tau_i + 1}
        \tau_i\pricej\doubleprimed\optexpectallocij(\typei)
        +
        \frac{\tau_i}{\tau_i+1}
        \left(\pricej\primed - \pricej\doubleprimed\right)
        \right)
        \right)}
        -
        \moneycost_i\left(
        \sum_{j\in[m]}
        \tildePj\doubleprimed
        \optexpectallocij(\typei) 
        \right)
        \\
        &\overset{(c)}{\geq} 
        \expect
        [\price\primed\sim\randomprice\primed,\tau_i\sim\boldsymbol\tau_i]
        {
        \valuefunc_i\left(
        \left(
        \max_{j\in[m]}
        \frac{\ctrij}{\pricej\primed}
        \right)
        \left(
        \sum_{j\in[m]}
        \frac{\tau_i}{\tau_i + 1}
        \pricej\primed
        \optexpectallocij(\typei)
        \right)
        \right)}
        -
        \moneycost_i\left(
        \sum_{j\in[m]}
        \tildePj\doubleprimed
        \optexpectallocij(\typei)  
        \right)
        \\
        &\overset{(d)}{\geq} 
        \expect
        [\price\primed\sim\randomprice\primed,\tau_i\sim\boldsymbol\tau_i]
        {
        \valuefunc_i\left(
        \sum_{j\in[m]}
        \frac{\tau_i}{\tau_i + 1}
        \frac{\ctrij}{\pricej\primed}
        \pricej\primed
        \optexpectallocij(\typei)
        \right)}
        -
        \moneycost_i\left(
        \sum_{j\in[m]}
        \tildePj\doubleprimed
        \optexpectallocij(\typei)\right)
        \\
        &\overset{(e)}{\geq} 
        \expect
        [\tau_i\sim\boldsymbol\tau_i]
        {
        \frac{\tau_i}{\tau_i + 1}
        \valuefunc_i\left(
        \sum_{j\in[m]}
        {\ctrij}
        \optallocij
        \right)}
        -
        \moneycost_i\left(
        \sum_{j\in[m]}  
        \tildePj\doubleprimed
        \optalloc_{ij}\right)
        \\
        &\overset{(f)}{\geq} 
        {
        \frac{1}
        {\expect{\frac{1}{\boldsymbol\tau_i}}
        +1 }
        \valuefunc_i\left(
        \sum_{j\in[m]}
        {\ctrij}
        \optexpectallocij(\typei)
        \right)}
        -
        \moneycost_i\left(
        \sum_{j\in[m]}
        \tildePj\doubleprimed
        \optexpectallocij(\typei)  \right)
        \\
        &\overset{(g)}{=} 
        \frac{1}{2}
        \valuefunc_i\left(
        \sum_{j\in[m]}
        {\ctrij}
        \optexpectallocij(\typei)
        \right)
        -
        \moneycost_i\left(
        \sum_{j\in[m]}
        \tildePj\doubleprimed
        \optexpectallocij(\typei)  \right)
    \end{align*}
    where inequality~(a) holds due to \Cref{lem:compute utility};
    inequality~(b) holds since $\reportwealth_i\primed = \tau_i\sum_{j\in[m]}\pricej\doubleprimed\optexpectallocij(\typei)$
    and $\wealth_i\primed \geq \sum_{j\in[m]}\pricej\primed - \pricej\doubleprimed$ for every realization of 
    $\pricej\primed,\pricej\doubleprimed,\tau_i$;
     inequality~(c) holds since $\optexpectallocij(\typei) \leq 1$;
     inequality~(d) holds by algebra;
     {inequality~(e) holds due to the concavity of 
     valuation function $\valuefunc_i$ and $\valuefunc_i(0) = 0$};
     inequality~(f) holds due to Jensen's inequality;
     and equality~(g) holds since $\expect{\frac{1}{\boldsymbol\tau_i}} = 1$
     by definition.

    We are ready to prove inequality~\eqref{eq:multi item bayes wtp bound} as follows,\footnote{Similar to the proof of \Cref{thm:poa mixed}, this analysis does not use the differentiability of money cost function $\moneycost_i$. Thus, without loss of generality we assume $\truewealth_i = \infty$ and allow $\moneycost:\reals_+ \rightarrow \realsinf$ to simplify the presentation.}
    \begin{align*}
        \frac{1}{2}\wtp_i(\optexpectalloc_i(\typei)) 
        &\overset{(a)}{=}
        \frac{1}{2}\moneycost_i^{-1}\left(
        \valuefunc_i\left(
        \sum_{j\in[m]}
        {\ctrij}
        \optexpectallocij(\typei)
        \right)
        \right)
        \\
        &\overset{(b)}{\leq}
        \moneycost_i^{-1}\left(
        \frac{1}{2}
        \valuefunc_i\left(
        \sum_{j\in[m]}
        {\ctrij}
        \optexpectallocij(\typei)
        \right)
        \right)
        \\
        &\overset{(c)}{\leq}
        \moneycost_i^{-1}\left(
        \util_i(\typei) + 
        \moneycost_i\left(\sum_{j\in[m]} 
        \tildePj\doubleprimed
        \optexpectallocij(\typei) \right)
        \right)
        \\
        &\overset{(d)}{\leq}
        \moneycost_i^{-1}\left(
        \util_i(\typei)  \right)
        +
        \sum_{j\in[m]}
        \tildePj\doubleprimed
        \optexpectallocij(\typei)  
    \end{align*}
    where 
    equality~(a) holds due to the definition of $\wtp_i$;
    inequalities~(b) and (d) hold due to the concavity of $\moneycost_i^{-1}$;
    and
    inequality~(c) holds due to the monotonicity of $\moneycost_i^{-1}$ and the lower bound of the expected utility $\util_i$ obtained above.

    Given inequality~\eqref{eq:multi item bayes wtp bound}, we are able to show the upper bound of the PoA.
    First, for agent $i$ with type $\typei$,
      \begin{align*}
        \expect[\typeMinusI \sim \typedist_{-i}]{\wtp_i(\optalloc_i(\typei, \typeMinusI))} 
	&\overset{(a)}{\leq} 
	\wtp_i\left(\expect[\typeMinusI \sim \typedist_{-i}]{\{\optalloc_i(\typei, \typeMinusI)}
	\right)
	\\
	&=
	\wtp_i\left(\optexpectalloc_i(\typei)\right)
	\\
        &\leq 
        2
        \left(
        \moneycost_i^{-1}(\util_i(\typei))
        +
        \sum_{j\in[m]}\tildePj\doubleprimed\optexpectallocij(\typei) 
        \right)
    \end{align*}
    where 
inequality~(a) holds due to the concavity of $\wtp_i$.
Then,
\begin{align*}
\expect[\type\sim\typedist]{\wtp(\optalloc(\type))} 
&= 
\sum_{i\in[n]}\expect[\typei\sim\typedist_i]{\expect[\typeMinusI \sim \typedist_{-i}]{\wtp_i(\optalloc_i(\typei, \typeMinusI))}}
\\
&\leq
\sum_{i\in[n]}\expect[\typei\sim\typedist_i]{
2
        \left(
        \moneycost_i^{-1}(\util_i(\typei))
        +
        \sum_{j\in[m]}\tildePj\doubleprimed\optexpectallocij(\typei) 
        \right)
}
\\
&\leq 
2\left(
\sum_{i\in[n]}\expect[\typei\sim\typedist_i]{
\moneycost_i^{-1}(\util_i(\typei))
}
+
\sum_{i\in[n]}\sum_{j\in[m]}\expect[\typei\sim\typedist_i]{\tildePj\doubleprimed\optexpectallocij(\typei)}
\right)
\\
&\leq 
4\sum_{i\in[n]}
\expect[\typei\sim\typedist_i]{\wtp(\randomalloc_i(\typei)}
\end{align*}
where the last inequality hold
since $\moneycost_i^{-1}(\util_i(\typei))\leq \moneycost_i^{-1}
    \left(\expect[\alloci\sim\randomalloc_i(\typei)]{\valuefunc_i\left(\sum_{j\in[m]}\allocij\ctrij\right)}\right)
    =
    \wtp_i(\randomalloc_i(\typei))$, and
    $\sum_{i\in[n]}\sum_{j\in[m]}\expect[\typei\sim\typedist_i]{\tildePj\doubleprimed\optexpectallocij(\typei)} 
    = \sum_{j\in[m]}{\tildePj}
    \leq \sum_{i\in[m]}\sum_{j\in[m]}\expect[\typei\sim\typedist_i]
    {\expect[\paymenti\sim\randompayment_i(\typei)]{\paymenti}}
    \leq  \wtp(\allocs)$ where the last two inequalities hold due to the stochastic dominance between $\randomprice\doubleprimed$ and $\randomprice(\type_i)$ for all type $\typei$ and the non-negative utility of every agent in the equilibrium.
\end{proof}

\subsection{Proof of \texorpdfstring{\Cref{prop:multiple pure nash for budgeted agents}}{}}
\singleitembudgetedagentmultinash*
\label{apx:singleitembudgetedagentmultinash}
\begin{proof}
    For every agent $i$, we have
    \begin{align*}
        \allocLBi(\price) &\overset{(a)}{=} \left(1 \wedge \left(1 - \frac{\price}{\vali}\right)\right)\cdot \indicator{\vali \geq \price}
        =
        \left(1 - \frac{\price}{\vali}\right)\cdot \indicator{\vali \geq \price}
        ,\qquad
        % \\
        \allocUBi(\price) = \left(1 \wedge \frac{\truewealth_i}{\price}\right)\cdot \indicator{\vali\geq \price}
    \end{align*}
    where equality~(a) holds since $\truewealthi > \frac{1}{4}\vali$.
    Moreover, $\allocLBi(\price) < \allocUBi(\price)$ for every $\price \leq \vali$.
    Thus, $\PriceL \subsetneq \PriceH$ and $\PriceL = \{\priceL\}$ is a singleton.

    Below we argue that there exists high-price equilibrium with per-unit price in $\PriceH \backslash\PriceL$.
    By the definition of $\PriceL$, there exists at least two distinct agents $i\primed, i\doubleprimed$ such that $\allocLB_{i\primed}(\priceL), \allocLB_{i\doubleprimed}(\priceL) > 0$, $\val_{i\primed}, \val_{i\doubleprimed} \geq \priceL$ and thus $\allocLB_{i\primed}(\priceL) < \allocUB_{i\primed}(\priceL)$, $\allocLB_{i\doubleprimed}(\priceL) < \allocUB_{i\doubleprimed}(\priceL)$. 
    Therefore, 
    \begin{align*}
        \allocLB_{i\primed}(\priceL) + \sum\nolimits_{i\not=i\primed}\allocUBi(\priceL) 
        > 
        \allocLB_{i\primed}(\priceL) + \allocLB_{i\doubleprimed}(\priceL) +
        \sum\nolimits_{i\not=i\primed,i\doubleprimed}\allocUBi(\priceL)
        \geq 
        \sum\nolimits_{i\in[n]}\allocLBi(\priceL) = 1
    \end{align*}
    Since both $\allocLBi(\price)$ and $\allocUBi(\price)$ are continuous, for sufficiently small $\epsilon$, we have
    \begin{align*}
        \allocLB_{i\primed}(\priceL + \epsilon) + \sum\nolimits_{i\not=i\primed}\allocUBi(\priceL + \epsilon) > 1
    \end{align*}
    which is the condition for high-price equilibrium existence with per-unit price $\priceL + \epsilon \in \PriceH\backslash\PriceL$ as desired.
\end{proof} 

\subsection{Proof of \texorpdfstring{\Cref{prop:single item eqlb computation}}{}}
\computeeqlb*
\label{apx:computeeqlb}
\begin{proof}
    By \Cref{lem:allocLB allocUB}, both $\allocLB_i(\price)$ and $\allocUB_i(\price)$ are weakly decreasing in $\price$.
    Thus, $\PriceL$ and consequently low-price equilibrium in \Cref{prop:single item eqlb computation} can be computed (e.g., via binary search) in polynomial time. Similarly, $\PriceH$ can be computed (e.g., via binary search) in polynomial time. Moreover, price $\price\in\PriceH$ that satisfies conditions $|\lagent(\price)| \geq 1$ and $\sum_{i\in\hagent(\price)}\allocUB_(\price)\geq 1$ can also be identified in polynomial time. Therefore, the high-price equilibrium can be computed in polynomial time as well.
\end{proof}

\section{Best Response Computation for Single-item Instances}
\label{apx:best response computation}

In this section, we focus on the metagame for single-item instances and explain how to compute the best response given other agents' reported message in polynomial time. Similar to \Cref{sec:single item pure nash}, we assume $\ctrij = 1$ for all agents and drop it as well as subscript index $j$ for the item without loss of generality.

\begin{proposition}
    In the metagame with a single item, given other agents' reported message $(\reportval_{-i}, \reportwealth_{-i})$,
    there exists a polynomial time algorithm that 
    computes the best response $(\reportval_i, \reportwealth_i)$ of agent $i$.
\end{proposition}
\begin{proof}
    We focus on the best response computation of agent $1$. 
    Without loss of generality, we assume $\reportval_2 \leq \reportval_3 \leq \dots \leq \reportval_n$.
    By \Cref{lem:reporting infinite value}, it suffices to consider reported message $(\reportval_1 = \infty, \reportwealth_1)$.
    By definition, the per-unit price $\price$ of the inner FPPE is continuous, weakly increasing in $\reportwealth_1$ and agent $1$ has allocation $\alloc_1 = \sfrac{\reportwealth_1}{\price}$ and payment $\payment_1 = \reportwealth_1$. Though agent $1$'s utility $\util_1$ is not globally concave in $\reportwealth_1$, it can be divided into at most $2n-1$ pieces where each piece is concave in $\reportwealth_1$. See \Cref{fig:linear agent 1 inefficient deviation} for an graphical illustration. Specifically, in each piece,  price $\price$ either (i) increases linearly in $\reportwealth_i$, or (ii) stays constant and is equal to $\reportval_i$ for some agent $i$. In both cases, it can be verified that agent $1$'s utility is concave, since her valuation function $\valuefunction_1$ (money cost function $\moneycost_i$) is concave (convex). By utilizing this piecewise concavity, the best $\reportwealth_1$ in each piece can be solved efficiently. Since there are at most $2n-1$ pieces, the optimal $\reportwealth_1$ can be computed in polynomial time as desired.
\end{proof}

\end{document}